\documentclass[%
 aip,
 jmp,
 amsmath,amssymb,
 reprint, onecolumn
]{revtex4-1}

\usepackage{graphicx}
\usepackage{dcolumn}
\usepackage{bm}
\usepackage[utf8]{inputenc}
\usepackage[T1]{fontenc}
\usepackage{mathptmx}
\usepackage{etoolbox}

\usepackage{amsthm}
\usepackage{enumitem}
\usepackage{mathtools}
\usepackage{braket}
\usepackage{tikz}
\usetikzlibrary{arrows.meta}
\usepackage{hyperref}
\hypersetup{colorlinks=true,linkcolor=blue,citecolor=blue,urlcolor=blue,
  pdftitle={Algebraic quantum kinematics: Galilean covariance with positive energy confines unequal-time commutation to null sets},
  pdfauthor={Leonardo A. Pach\'on}}

\newtheorem{theorem}{Theorem}
\newtheorem{proposition}{Proposition}
\newtheorem{lemma}{Lemma}
\newtheorem{corollary}{Corollary}
\newtheorem{alemma}{Lemma}

\theoremstyle{definition}
\newtheorem{definition}{Definition}
\theoremstyle{remark}
\newtheorem{remark}{Remark}

\makeatletter
\def\@email#1#2{%
 \endgroup
 \patchcmd{\titleblock@produce}
  {\frontmatter@RRAPformat}
  {\frontmatter@RRAPformat{\produce@RRAP{*#1\href{mailto:#2}{#2}}}\frontmatter@RRAPformat}
  {}{}
}%
\makeatother

\newcommand{\R}{\mathbb{R}}
\newcommand{\C}{\mathbb{C}}
\newcommand{\Z}{\mathbb{Z}}
\newcommand{\N}{\mathbb{N}}

\newcommand{\A}{\mathcal{A}}
\newcommand{\F}{\mathcal{F}}
\newcommand{\Hil}{\mathcal{H}}
\newcommand{\D}{\mathcal{D}}
\newcommand{\Op}{\mathcal{O}}
\newcommand{\K}{\mathcal{K}}
\newcommand{\one}{\mathbb{1}}

\newcommand{\rmd}{\mathrm{d}}             
\newcommand{\rmi}{\mathrm{i}}             
\newcommand{\rme}{\mathrm{e}}             
\DeclareMathOperator{\spec}{spec}         
\DeclareMathOperator{\supp}{supp}         
\newcommand{\hH}{\hat{H}}
\newcommand{\hM}{\hat{M}}
\newcommand{\hP}{\hat{\mathbf{P}}}
\newcommand{\hK}{\hat{\mathbf{K}}}
\newcommand{\psio}{\hat{\psi}_0}
\newcommand{\psiod}{\hat{\psi}_0^\dagger}
\newcommand{\Uint}{\hat{U}_{\mathrm{int}}}
\newcommand{\bx}{\mathbf{x}}
\newcommand{\by}{\mathbf{y}}
\newcommand{\ba}{\mathbf{a}}
\newcommand{\bv}{\mathbf{v}}
\newcommand{\bp}{\mathbf{p}}
\newcommand{\evol}[1]{\rme^{\rmi\hH #1/\hbar}}
\newcommand{\evolm}[1]{\rme^{-\rmi\hH #1/\hbar}}

\begin{document}

\preprint{}

\title[Algebraic quantum kinematics]{%
Algebraic quantum kinematics: Galilean covariance with positive energy
confines unequal-time commutation to null sets}

\author{Leonardo A.\ Pach\'{o}n}
\affiliation{guane Enterprises, R+D+I Unit, Medell\'{i}n 050010, Colombia}
\email{leonardo@guane.com.co}

\date{\today}

\pacs{03.65.Ta, 03.70.+k, 11.10.Cd}
\keywords{algebraic quantum field theory; Galilean covariance;
Bargmann central extension; unequal-time locality;
Galilean transitivity; commutation sets; quantum speed limits;
kinematical groups; Newton--Hooke groups; Lieb--Robinson bounds}


\begin{abstract}
We ask which kinematical groups allow local algebras of canonical
quantum fields to commute at unequal times. The Galilei group
acts transitively on pairs of non-simultaneous events with a given time
separation, so a Galilean-invariant commutation relation between
sharp-time local algebras at unequal times also holds at coincident
positions. At arbitrarily short separations it contradicts the
canonical commutation relations, with no spectral assumption. Under unitary Bargmann covariance and positivity of the
energy modulo the central mass, the vacuum is annihilated by the
annihilation fields, hence not separating for local algebras, and
the vacuum two-point function of a canonical field factorises into the
free Schr\"odinger propagator and the characteristic function of an
internal-energy distribution. Sharp-time algebras at times
$t_1\neq t_2$ then commute only if the internal state created by the
field is orthogonal to its own evolution over $|t_2-t_1|$. These
separations form a closed null set, bounded below by the
Mandelstam--Tamm and Margolus--Levitin times and empty for sharp or
generic internal energy, as for Galilean many-body systems with bounded
pair interactions; any finite set of separations occurs. Algebras of open
spacetime regions containing time-smeared fields never commute, and at
unequal times the norm of vacuum commutators does not decrease with
spatial separation. Positivity cannot be dropped: without it, the slice
algebras and the algebras of open regions of a net with unitary Bargmann
covariance whose fields annihilate the vacuum can commute at all time
separations beyond any given one.
Among spatially isotropic homogeneous spacetimes with absolute time,
only the galilean ones are transitive at fixed time separation, up to
discrete exceptions in the two oscillating cases, which a harmonic trap
realises in the Newton--Hooke one. In bounded external potentials
commutation sets stay null and avoid short times. KMS analyticity cannot
replace the one-sidedness of the vacuum sector, but non-commutation
transfers from the vacuum to every representation of a simple
C*-algebra on which the dynamics acts by automorphisms, including free
Bose and Fermi gases at all temperatures and the Galilei-covariant
cutoff dynamics of interacting fermions of Narnhofer and Thirring.
Without such a vacuum representation the conclusions can fail even in
thermal equilibrium, as a quasi-free Fermi system of infinite density
shows.
\end{abstract}

\maketitle

\section{Introduction}\label{sec:intro}

\subsection{The question}
\label{sec:motivation}

In the operator-algebraic formulation of quantum field
theory~\cite{HaagKastler1964,Haag1996} the kinematic group acts by
automorphisms of a net of local algebras, and locality is a property of
the net: the algebras of regions that cannot influence each other
commute. In relativistic theories this includes regions at unequal
times. The Lorentzian interval separates the orbits of pairs of events
under the Poincar\'e group, so a spacelike pair (``there, later'') is
invariantly distinguished from a timelike one (``here, later''). In
Galilean quantum field theory local commutativity has been imposed only
at equal times, as what remains of spacelike separation once the light
cone is flattened onto the hyperplanes of constant
time~\cite{LevyLeblond1967,Dadashev1985,PucciniVucetich2004}. We ask
whether this restriction is a convention or a necessity:

\begin{quote}
\emph{Question.} Can a net of local algebras generated by canonical
fields---in the creation--annihilation form or in the
$(\hat\phi,\hat\pi)$ form---be covariant under the Galilei group and
satisfy a non-trivial commutation requirement between regions at unequal
times? More generally, for which kinematical groups is such
unequal-time locality possible?
\end{quote}

The question belongs to what may be called algebraic quantum
kinematics. The canonical commutator enters through the algebra it
generates, the kinematic group through its action by automorphisms, and
the spectrum condition only where it is stated. Keeping the three apart
shows which conclusions need which input.

\subsection{Results}
\label{sec:results}

\paragraph{Kinematic exclusion.}
The Galilei group acts transitively on pairs of non-simultaneous events
with a given time separation (Lemma~\ref{lem:galilean_transitivity}):
there is no absolute space, and ``here, later'' and ``there, later'' are
the same configuration seen from two inertial frames. A
Galilean-invariant commutation requirement between sharp-time (slice)
algebras at a time separation $\tau$ therefore holds at every spatial
separation, coincident positions included. If it holds for arbitrarily
small $\tau$, continuity contradicts the canonical commutator in either
of its forms. No spectral condition and no property of the vacuum beyond
time-translation invariance is used (Theorem~\ref{thm:sr_selection},
Corollary~\ref{cor:no_unequal_time_locality}). The restriction to short
times is sharp: without positivity, commutation can hold at all
separations beyond any given one
(Proposition~\ref{prop:attainable}(e)).

\paragraph{The vacuum sector.}
We then assume that the Galilei group is implemented by a unitary
representation of its Bargmann central extension, with central mass
$\hM$, and that $\hH+\lambda\hM\geq 0$ for some $\lambda\in\R$, that
is, positivity of the energy modulo the central charge; with an
invariant vacuum these are the hypotheses~(H) of
Definition~\ref{def:H}. Under these assumptions the mass spectrum is
non-negative, every vector invariant
under spatial translations has zero mass, and the time-zero fields of
positive mass charge annihilate the vacuum
(Lemmas~\ref{lem:zero_mass_vacuum}--\ref{lem:vacuum_annihilation}).
Hermitian canonical fields without mass charge are excluded
(Proposition~\ref{prop:hermitian_zero_mass}), and the vacuum is not
separating for the local field algebras
(Corollary~\ref{thm:galilean_RS}).

\paragraph{Commutation sets.}
The vacuum two-point function of a canonical field of mass $m_0$
factorises into the free Schr\"odinger propagator and the characteristic
function $\widehat{\nu_\xi}$ of an internal-energy probability measure
(Lemma~\ref{lem:GKL}). Two slice algebras at times $t_1\neq t_2$
therefore commute only if $\widehat{\nu_\xi}$ vanishes at
$(t_2-t_1)/\hbar$ (Theorem~\ref{thm:commutation_set}). In words, the
internal state $\xi$ of the created particle must be orthogonal to its
own evolution over $|t_2-t_1|$. The set of such time separations is a
closed null set bounded below by the Mandelstam--Tamm and
Margolus--Levitin times, and it is empty when the internal energy is
sharp. The inclusion is optimal: for every admissible internal-energy
distribution there is a net that commutes exactly at those time
separations. Exact commutation is non-generic, but every finite set of
time separations occurs, and so do every discrete set $\{z_j\}$ with
$\sum_j1/z_j<\infty$ and some sets with a finite accumulation point; for
internal energies of bounded range the number of commutation times
grows at most linearly (Proposition~\ref{prop:attainable}). Unequal-time
commutation of slice algebras is therefore possible in a Galilean net
with positive energy, but only at a closed null set of time separations
away from zero. What is excluded is commutation at short times, with no
spectral assumption (Theorem~\ref{thm:sr_selection}), and, under
positivity, commutation on sets of positive measure
(Corollary~\ref{thm:sr_selection_analytic}) and between open spacetime
regions (Corollary~\ref{cor:4D}). Positivity cannot be dropped: without
it, a net with unitary Bargmann covariance whose fields annihilate the
vacuum can commute at all time separations beyond any given $\tau_0>0$
(Proposition~\ref{prop:attainable}(e)).

\paragraph{Consequences.}
Under (H) and a natural link~(T) between the algebras of open spacetime
regions and time-smeared fields, no two such algebras commute, not even
at overlapping times (Corollary~\ref{cor:4D}). In every Galilean
many-body system with free kinetic energy and a bounded pair
interaction, no two non-empty slice algebras at unequal times commute,
and for fermions their odd generators do not all anticommute
(Corollary~\ref{cor:pair}). The results extend to Fermi fields and to
gauge-invariant observable algebras (Proposition~\ref{prop:fermi_obs}).
At a fixed time separation $\tau\neq0$, the norm of the vacuum
commutator between the field algebras of two regions does not depend on
their spatial separation, so continuum
Lieb--Robinson bounds must break boost invariance
(Proposition~\ref{prop:antiLR}).

\paragraph{Classification.}
Among the spatially isotropic homogeneous kinematical and aristotelian
spacetimes with absolute time classified by Figueroa-O'Farrill and
Prohazka~\cite{FigueroaOFarrillProhazka2019}, transitivity at fixed time
separation holds exactly for the galilean ones. The exceptions are the
zeros of the determinant of a two-point boundary problem for the boost
trajectories, which are non-trivial only in the oscillating
(anti-de Sitter) Newton--Hooke cases. For the isotropic harmonic trap
they are realised: at half periods the slice algebras of regions that
are disjoint from each other's reflection through the trap centre commute
(Theorem~\ref{thm:classification}, Proposition~\ref{prop:trap}).
Minkowski, Carroll and static kinematics admit canonical nets with
non-trivial unequal-time locality.

\paragraph{External potentials.}
Without boosts the transitivity argument is not available, and the
harmonic trap shows that refocusing can make slice algebras commute. For
quasi-free fields in an external potential that is locally bounded and
bounded below, or of Coulomb type, commutation sets are still closed and
null, and algebras of open spacetime regions never commute. For bounded potentials no
commutation occurs below a time separation that depends on the potential
and on the size of the regions, but not on the distance between them
(Proposition~\ref{prop:external}). The proofs rest on a unique
continuation lemma for the spectral family of a Schr\"odinger operator
(Lemma~\ref{lem:spectral_UC}).

\paragraph{Thermal states.}
Analyticity in the KMS strip cannot replace the one-sidedness that
positivity and Bargmann covariance give the vacuum sector. We give a
positive-energy canonical field with KMS states at every temperature
whose commutators vanish for all time separations above a threshold
(Proposition~\ref{prop:KMS_no}). Non-commutation does transfer from the
vacuum representation to every representation of a simple C*-algebra on
which the dynamics acts by automorphisms, which covers the free Bose and
Fermi gases at all temperatures (Proposition~\ref{prop:transfer},
Corollaries~\ref{cor:free_gases}--\ref{cor:GNRS}). It also covers, in
all their KMS states, the Galilei-covariant cutoff dynamics of fermions
of Narnhofer and Thirring~\cite{NarnhoferThirring1991}
(Corollary~\ref{cor:NT}). Without a
representation satisfying (H) this can fail even in equilibrium: a
quasi-free Fermi system with internal levels unbounded below has KMS
states at every temperature whose odd fields anticommute at all
separations beyond a threshold, in every representation; its density is
infinite, and at finite density, as for quasi-free bosons in
equilibrium, positivity is restored
(Proposition~\ref{prop:fermi_medium}). For Bose gases in the
resolvent-algebra framework the algebra is not simple, and
non-commutation transfers to the representations whose kernel is
contained in that of the one-particle representation. What remains open
is stated in Remark~\ref{rem:massive_vacuum}.

\subsection{Relation to previous work}
\label{sec:previous}

\paragraph{Galilean axiomatics.}
L\'evy-Leblond~\cite{LevyLeblond1967} formulated Galilean quantum field
theory with local commutativity at equal times. He observed that a
non-degenerate Galilean-invariant vacuum has zero mass, that Hermitian
Galilean fields are massless, and that the CPT, spin--statistics and
Haag theorems do not hold. Dadashev~\cite{Dadashev1985} gave a
Wightman-type axiomatics for time-zero fields. In it the vacuum space is
the joint kernel of the spatial momenta, assumed one-dimensional, the
particle-number operators are assumed non-negative, and the
annihilation fields are shown to annihilate the vacuum (Proposition~3
there). Puccini and Vucetich~\cite{PucciniVucetich2004} give another
axiomatics with equal-time local commutativity.
Lemma~\ref{lem:joint_spectrum} below derives the non-negativity of the
mass from positivity of the energy modulo the central charge and boost
covariance, and Lemma~\ref{lem:vacuum_annihilation} then recovers
Dadashev's Proposition~3. When the time-zero fields integrate to a Weyl
system, with $\D$ a core for the fields, $\hM/m_0$ is a number operator
for it in the sense of Chaiken~\cite{Chaiken1968}. By Lemma~\ref{lem:joint_spectrum} it is
bounded below, so the Weyl system is a direct sum of Fock
representations (Theorem~1.3 of Ref.~\onlinecite{Chaiken1968}). We use
only the elementary consequence $\psio(g)\Omega=0$. Our Lemmas~\ref{lem:zero_mass_vacuum}--\ref{lem:vacuum_annihilation}
and Proposition~\ref{prop:hermitian_zero_mass} are therefore known in
substance. They are proved here, in the form needed below, to keep the
paper self-contained.

\paragraph{Reeh--Schlieder property.}
Requardt~\cite{Requardt1982} showed that the spectrum condition,
analyticity of vacuum expectation values, clustering and the cyclicity
of the vacuum for local algebras hold in Galilei-invariant theories with
zero mean particle density. Bain~\cite{Bain2011} pointed out that in a classical
spacetime the causal complement of a region of non-zero duration is
empty. The relativistic derivation of the separating property is
therefore unavailable, and the vacuum can be cyclic for such regions
without being separating. Corollary~\ref{thm:galilean_RS} states the
failure of the separating property explicitly. Modular theory relative
to the vacuum is then unavailable for local algebras, while faithful
states such as KMS states keep their modular
structure~\cite{HaagHugenholtzWinnink1967}.

\paragraph{Localization and causality.}
Malament's theorem~\cite{Malament1996}, in the form given by Halvorson
and Clifton~\cite{HalvorsonClifton2002}, concerns localization systems
for a single particle over an affine spacetime. Such a system consists
of projections attached to bounded regions of hyperplanes, together with
a unitary representation of the translations. It is trivial if it is
translation covariant, satisfies localizability, has energy bounded
below for every timelike translation, obeys microcausality for disjoint
regions at short time separations, and the spacetime has ``no absolute
velocity'' (every spacelike translation is a difference of timelike
ones). The analytic input is a lemma of Borchers~\cite{Borchers1967}.
Example~1 of Ref.~\onlinecite{HalvorsonClifton2002} shows that the
conclusion fails in Newtonian spacetime with a distinguished timelike
direction, where the energy condition holds trivially. Without a
distinguished direction the energy condition fails in that example,
because the generator of a timelike translation is not bounded below.
Theorem~\ref{thm:sr_selection} is a field-theoretic counterpart for the
Galilei group, and it differs in four respects. (i) It concerns algebras
generated by canonical fields, not localization projections, and no
localizability is assumed. (ii) It uses no spectral condition: the
contradiction comes from the canonical commutator, by continuity. In a
Galilean theory the energy condition of
Refs.~\onlinecite{Malament1996,HalvorsonClifton2002} is in any case not
available. The generator of the timelike translation with velocity
$\bv$ is $\hH-\bv\cdot\hP$, and under positive energy
$\hH-\bv\cdot\hP\geq-\hM|\bv|^2/2$ (Lemma~\ref{lem:joint_spectrum}).
This bound holds sector by sector in the mass, but it is not uniform in
the Fock space of massive particles, where $\hM$ is unbounded.
(iii) It excludes arbitrary Galilean-invariant commutation relations,
and Galilean invariance enters through the boosts, not through
translations alone. (iv) Under positive energy,
Theorem~\ref{thm:commutation_set} identifies a closed null set of time
separations that contains every separation at which slice algebras
commute, and shows that, for every internal-energy distribution, some
net commutes at exactly these separations. The analytic step of
Theorem~\ref{thm:commutation_set} is the one
behind Hegerfeldt's instantaneous spreading~\cite{Hegerfeldt1974,
Hegerfeldt1985,Hegerfeldt1998} and his analysis of Fermi's two-atom
problem~\cite{Hegerfeldt1994}. Buchholz and
Yngvason~\cite{BuchholzYngvason1994} showed that in relativistic quantum
field theory the apparent conflict with causality disappears once local
properties of states are compared through expectation values of local
observables, whose algebras commute at spacelike separation.
Theorem~\ref{thm:commutation_set} and Corollary~\ref{cor:4D} show that
in a Galilean theory with positive energy the local algebras themselves fail to commute at
unequal times: slice algebras outside a closed null set of time
separations, and algebras of open spacetime regions at all
separations. For related results on the localization of states in
quantum field theory see Refs.~\onlinecite{Knight1961,Schlieder1965}.

\paragraph{Kinematical groups.}
Bacry and L\'evy-Leblond~\cite{BacryLevyLeblond1968} listed the possible
kinematics in $3+1$ dimensions, and Bacry and Nuyts~\cite{BacryNuyts1986}
classified the kinematical Lie algebras. Figueroa-O'Farrill and
Prohazka~\cite{FigueroaOFarrillProhazka2019} classified the
simply-connected spatially isotropic homogeneous spacetimes in all
dimensions, including torsional galilean ones, which are studied further
in Ref.~\onlinecite{FigueroaOFarrillGormezVandenBleeken2023}. On Newton--Hooke
spacetimes see Refs.~\onlinecite{AldrovandiEtAl1999,GibbonsPatricot2003}.
On Carroll kinematics see Refs.~\onlinecite{DuvalGibbonsHorvathyZhang2014,
deBoerHartongObersSybesmaVandoren2022}, for an algebraic treatment of
free Carrollian fields Ref.~\onlinecite{FredenhagenProhazkaTiefenbacher2026},
and on ultralocal models Ref.~\onlinecite{Klauder1970}. We are not aware
of an analysis of unequal-time locality across these classifications.

\paragraph{Schr\"odinger symmetry.}
For fields covariant under the Schr\"odinger group, the boost Ward
identity fixes the vacuum two-point function up to a function of the
time difference~\cite{Hagen1972,Henkel1994,NishidaSon2007}.
Lemma~\ref{lem:GKL} is the operator form of this identity for canonical
fields under the Bargmann group alone. It identifies the function of time
as the characteristic function of an internal-energy measure, and this
identification is what makes the commutation set computable.

\paragraph{Interacting Galilean theories and propagation bounds.}
Interacting Galilean models include the Lee
models~\cite{LevyLeblond1967,Schrader1968}, fermions with
Galilei-invariant cutoff interactions~\cite{NarnhoferThirring1990,
NarnhoferThirring1991,Jakel1991}, and Bose systems with pair potentials
treated with the resolvent algebra~\cite{Buchholz2018,Buchholz2020,
Buchholz2026}. For UV-regularized fermion interactions, with the
Galilei-covariant cutoff of Narnhofer and
Thirring~\cite{NarnhoferThirring1990,NarnhoferThirring1991} or with a
Gaussian smearing~\cite{GebertNachtergaeleReschkeSims2020},
the infinite-volume dynamics is a group of automorphisms of the CAR
algebra. Lieb--Robinson
bounds~\cite{LiebRobinson1972,NachtergaeleSims2006,HastingsKoma2006}
have been extended to continuum
systems~\cite{GebertNachtergaeleReschkeSims2020,HinrichsLemmSiebert2025,
BachmannDeNittis2025}. Proposition~\ref{prop:antiLR} explains why these
extensions bound interaction-picture differences or use test functions of
controlled momentum-space localization. The quantum speed limits that
appear in Theorem~\ref{thm:commutation_set} are those of Mandelstam and
Tamm~\cite{MandelstamTamm1991} and of Margolus and
Levitin~\cite{MargolusLevitin1998}; see also
Refs.~\onlinecite{Fleming1973,DeffnerCampbell2017}.

\subsection{What is new}
\label{sec:new}

To our knowledge, the following are new:
(i) the exclusion of arbitrary Galilean-invariant unequal-time
commutation relations realised at short times, for both
canonical-commutator forms and without spectral assumptions
(Theorem~\ref{thm:sr_selection}), whose kinematic input is elementary;
(ii) the characterisation of commutation sets through the
internal-energy measure, with the quantum-speed-limit bounds and its
optimality (Lemma~\ref{lem:GKL}, Theorem~\ref{thm:commutation_set}), and
its consequences in Sec.~\ref{sec:consequences};
(iii) the genericity and realisability of commutation sets, and the
necessity of positivity (Proposition~\ref{prop:attainable});
(iv) the classification over spatially isotropic homogeneous spacetimes
with absolute time, its realisation in a harmonic trap, and the
extension to external potentials (Secs.~\ref{sec:external}
and~\ref{sec:classification});
(v) the analysis of thermal states in Sec.~\ref{sec:media}, including
the transfer to non-simple algebras, the cutoff dynamics of Narnhofer
and Thirring, and a Fermi system in equilibrium without positivity.
Section~\ref{sec:setting} fixes the setting,
Sec.~\ref{sec:hegerfeldt_rigidity} proves the kinematic exclusion,
Sec.~\ref{sec:reeh_schlieder_strand} analyses the vacuum sector,
Sec.~\ref{sec:commutation_sets} proves the commutation-set theorem,
Sec.~\ref{sec:consequences} draws its consequences,
Sec.~\ref{sec:classification} classifies the kinematics with absolute
time, Sec.~\ref{sec:media} treats thermal states and media at rest, and
Sec.~\ref{sec:conclusion} concludes. Auxiliary lemmas are collected in
the Appendix. Figure~\ref{fig:structure} shows the main dependencies
between the results and the hypotheses that each of them uses.

\begin{figure}[htb]
\centering
\begin{tikzpicture}[font=\footnotesize, >={Stealth[length=2mm]},
  box/.style={rounded corners=2pt, align=center, inner sep=3pt,
    minimum height=9mm, text width=40mm,
    execute at begin node={\hyphenpenalty=10000\exhyphenpenalty=10000}},
  narrow/.style={text width=32mm},
  hyp/.style={box, draw=gray!70!black, dashed, fill=gray!10},
  der/.style={box, draw, fill=green!12},
  thm/.style={box, draw, line width=0.8pt, fill=blue!12, font=\footnotesize\bfseries},
  arr/.style={->, gray!70!black, line width=0.5pt}]
\def\figXa{2.05}\def\figXb{6.4}\def\figXc{11.15}\def\figXd{15.45}\def\figXn{10.65}
\def\figYa{0}\def\figYb{-1.65}\def\figYc{-3.3}\def\figYd{-4.95}\def\figYe{-6.6}
\node[hyp] (a0) at (\figXa,\figYa) {Canonical net (S1)--(S3), no spectral condition};
\node[der] (a1) at (\figXa,\figYb)
  {Lemma~\ref{lem:galilean_transitivity}\\ Galilean transitivity};
\node[thm] (a2) at (\figXa,\figYc)
  {Theorem~\ref{thm:sr_selection}, Corollary~\ref{cor:no_unequal_time_locality}\\
   \mdseries no locality at short times};
\node[thm] (a3) at (\figXa,\figYd)
  {Theorem~\ref{thm:classification}\\ \mdseries kinematical classification};
\node[thm] (a4) at (\figXa,\figYe) {Proposition~\ref{prop:trap}\\ \mdseries harmonic trap};
\node[hyp] (b0) at (\figXb,\figYa)
  {(H1)--(H4): Bargmann covariance, invariant vacuum, positivity modulo the mass};
\node[der] (b1) at (\figXb,\figYb)
  {Lemmas~\ref{lem:zero_mass_vacuum}--\ref{lem:vacuum_annihilation}\\
   zero mass, covariance, joint spectrum, vacuum annihilation};
\node[der] (b2) at (\figXb,\figYc)
  {Lemma~\ref{lem:GKL}\\ K\"all\'en--Lehmann factorisation};
\node[thm] (b3) at (\figXb,\figYd)
  {Theorem~\ref{thm:commutation_set}\\ \mdseries commutation sets};
\node[thm] (b4) at (\figXb,\figYe)
  {Corollary~\ref{thm:sr_selection_analytic}, Proposition~\ref{prop:attainable}\\
   \mdseries null and attainable commutation sets};
\node[hyp] (c0) at (\figXc,\figYa) {(T): open regions and time-smeared fields};
\node[thm, narrow] (c1) at (\figXn,\figYb)
  {Proposition~\ref{prop:hermitian_zero_mass}, Corollary~\ref{thm:galilean_RS}\\
   \mdseries Hermitian fields, vacuum not separating};
\node[thm, narrow] (c2) at (\figXn,\figYc)
  {Proposition~\ref{prop:antiLR}\\ \mdseries distance-independent commutators};
\node[thm] (c3) at (\figXc,\figYd)
  {Corollary~\ref{cor:4D}\\ \mdseries no commuting open regions};
\node[thm] (c4) at (\figXc,\figYe)
  {Proposition~\ref{prop:fermi_obs}\\ \mdseries Fermi fields, observable algebras};
\node[hyp] (d0) at (\figXd,\figYa) {Explicit models with free kinetic energy};
\node[thm] (d1) at (\figXd,\figYb)
  {Corollary~\ref{cor:pair}, Proposition~\ref{prop:external}\\
   \mdseries pair interactions; external potentials, with Lemma~\ref{lem:spectral_UC}};
\node[hyp] (d3) at (\figXd,\figYd) {KMS analyticity, C*-algebras with automorphisms};
\node[thm] (d4) at (\figXd,\figYe)
  {Propositions~\ref{prop:KMS_no}--\ref{prop:fermi_medium},
   Corollaries~\ref{cor:free_gases}--\ref{cor:NT}\\
   \mdseries thermal states, transfer, media};
\foreach \s/\t in {a0/a1,a1/a2,a2/a3,a3/a4,b0/b1,b1/b2,b2/b3,b3/b4,c3/c4,d0/d1,d3/d4}
  \draw[arr] (\s) -- (\t);
\draw[arr] (b1) -- (c1);
\draw[arr] (b2) -- (c2);
\draw[arr, dashed] (c1.north |- c0.south) -- (c1.north);
\draw[arr] ([xshift=15.5mm]c0.south) -- ([xshift=15.5mm]c3.north);
\draw[arr] ([yshift=-2.5mm]b2.east) -- ++(0.3,0) |- (c3.west);
\draw[arr] ([yshift=-2mm]b3.east) -- ++(0.14,0) |- (c4.west);
\end{tikzpicture}
\caption{Main dependencies between the results and the hypotheses that
each of them uses. Dashed boxes: hypotheses; green: lemmas; bold boxes:
results. The kinematic exclusion (first column) uses no spectral
condition; the classification extends it to the other galilean
kinematics, not to the aristotelian ones, and the harmonic trap shows
that it is sharp. The second column contains the commutation sets under
(H1)--(H4); Proposition~\ref{prop:attainable} is a statement about
probability measures, read through
Theorem~\ref{thm:commutation_set}(iv); its part~(e) drops (H4) and shows
that it cannot be dropped. The third column contains further
consequences of (H1)--(H4). The statements on open spacetime regions,
Corollary~\ref{cor:4D} and the last part of
Corollary~\ref{thm:galilean_RS} (dashed arrow), need the link~(T) in
addition, and Proposition~\ref{prop:fermi_obs} extends
Theorems~\ref{thm:sr_selection} and~\ref{thm:commutation_set} and
Corollaries~\ref{thm:sr_selection_analytic} and~\ref{cor:4D} to Fermi
fields and observable algebras. The last column contains the many-body
and external-potential models, whose proofs re-use those of
Theorem~\ref{thm:commutation_set} and Corollary~\ref{cor:4D}, and the
thermal states and media at rest;
Proposition~\ref{prop:fermi_medium}(iii) uses
Proposition~\ref{prop:fermi_obs}(a) and
Corollary~\ref{thm:sr_selection_analytic}. The general Lemmas
\ref{lem:truncation}--\ref{lem:fock_factorization} of the Appendix enter
proofs in all columns, Lemma~\ref{lem:truncation} already in
Theorem~\ref{thm:sr_selection}, and are not drawn.}
\label{fig:structure}
\end{figure}

\section{Setting}\label{sec:setting}

\subsection{Galilei group and slice regions}
\label{sec:galilean_covariance}

The Galilei group $\mathcal{G}$ consists of the elements
$g=(R,\bv,\ba,s)$, with $R\in SO(3)$, $\bv,\ba\in\R^3$ and $s\in\R$,
acting on events by
\begin{equation}\label{eq:galilei_action}
g\cdot(\bx,t)=(R\bx+\bv t+\ba,\;t+s).
\end{equation}
A \emph{slice region} is a pair $(\Op,t)$ with $\Op\subset\R^3$
non-empty, bounded and open and $t\in\R$, identified with
$\Op\times\{t\}$. The group acts on slice regions by
$g\cdot(\Op,t)=(R\Op+\bv t+\ba,\,t+s)$ and maps balls to balls of the
same radius. For a function $g$ on $\R^3$ we write $g_\ba:=g(\cdot-\ba)$,
and $(f,g)=\int\overline{f(\bx)}g(\bx)\,\rmd^3x$.

\subsection{Galilean canonical nets}
\label{sec:canonical_nets}

For a family $\mathfrak X$ of closed, densely defined operators on a
Hilbert space, the von Neumann algebra \emph{generated} by $\mathfrak X$
is the smallest von Neumann algebra $\mathcal{M}$ with which every
$X\in\mathfrak X$ is affiliated, i.e.\ such that $CX\subset XC$ for all
$C\in\mathcal{M}'$. It is the commutant of the von Neumann algebra of
bounded operators $C$ with $CX\subset XC$ and $C^*X\subset XC^*$ for all
$X\in\mathfrak X$.

\begin{definition}[Galilean canonical net]\label{def:canonical_net}
A \emph{Galilean canonical net} in the creation--annihilation form
consists of the following data.
\begin{enumerate}[label=(S\arabic*)]
\item \emph{Fields and slice algebras.} A Hilbert space $\Hil$, a dense
subspace $\D\subset\Hil$, a unit vector $\Omega\in\D$ and a self-adjoint
operator $\hH$ with $\evol{t}\Omega=\Omega$ for all $t\in\R$. For every
$g\in C_c^\infty(\R^3)$ there are operators $\psio(g)$ and $\psiod(g)$
with domain $\D$ that map $\D$ into itself. The map
$g\mapsto\psio(g)$ is antilinear, $g\mapsto\psiod(g)$ is linear, and
$\braket{u,\psio(g)w}=\braket{\psiod(g)u,w}$ for $u,w\in\D$; in
particular $\psiod(g)\subset\psio(g)^*$ and both operators are
closable. For a non-empty bounded open $\Op\subset\R^3$ the
\emph{slice algebra} $\F(\Op,0)$ is the von Neumann algebra generated by
the closures of $\psio(g)$ and $\psiod(g)$, $g\in C_c^\infty(\Op)$, and
$\F(\Op,t):=\evol{t}\F(\Op,0)\evolm{t}$.
\item \emph{Galilean covariance.} For every $g\in\mathcal{G}$ there is
a $*$-automorphism $\alpha_g$ of the C*-algebra generated by all slice
algebras such that $\alpha_g(\F(\Op,t))=\F(g\cdot(\Op,t))$ for every
slice region $(\Op,t)$.
\item \emph{Canonical commutator.} On $\D$,
$[\psio(f),\psiod(g)]=\hbar\,(f,g)$ and $[\psio(f),\psio(g)]=0$.
\end{enumerate}
In the \emph{$(\hat\phi,\hat\pi)$ form} the fields are replaced by
symmetric operators $\hat\phi_0(g)$, $\hat\pi_0(g)$, $g\in
C_c^\infty(\R^3)$ real, with domain $\D$, mapping $\D$ into itself and
real-linear in $g$, which satisfy $[\hat\phi_0(f),\hat\pi_0(g)]=
\rmi\hbar\,(f,g)$ and $[\hat\phi_0(f),\hat\phi_0(g)]=
[\hat\pi_0(f),\hat\pi_0(g)]=0$ on $\D$. The slice algebras are then
generated by their closures.
\end{definition}

We call data satisfying (S1) and (S3) a \emph{canonical net}; the
covariance~(S2) makes it Galilean. Isotony, $\F(\Op_1,t)\subset
\F(\Op_2,t)$ for $\Op_1\subset\Op_2$, holds by construction. Equal-time locality, $[\F(B_1,t),\F(B_2,t)]=\{0\}$ for
$B_1,B_2$ with disjoint closures, holds in all examples below but is not
used in any proof. Covariance~(S2) is needed only in
Corollary~\ref{cor:no_unequal_time_locality}. In Sec.~\ref{sec:reeh_schlieder_strand}
it is derived from a unitary representation of the Bargmann group
(Lemma~\ref{lem:covariance}). No commutation between algebras of
four-dimensional regions is assumed. Those algebras are treated in
Corollary~\ref{cor:4D}.

\subsection{Unequal-time locality relations}

\begin{definition}[Galilean unequal-time locality relation]
\label{def:galilean_locality_relation}
Let $\mathcal{L}$ be a set of unordered pairs
$\{(\Op_1,t_1),(\Op_2,t_2)\}$ of slice regions with $t_1\neq t_2$. We
call $\mathcal{L}$ a \emph{Galilean unequal-time locality relation} for
the net if
\begin{enumerate}[label=(L\arabic*)]
\item \emph{Microcausality:} $\{(\Op_1,t_1),(\Op_2,t_2)\}\in
\mathcal{L}$ implies $[\F(\Op_1,t_1),\F(\Op_2,t_2)]=\{0\}$;
\item \emph{Galilean invariance:} $\{(\Op_1,t_1),(\Op_2,t_2)\}\in
\mathcal{L}$ implies $\{g\cdot(\Op_1,t_1),g\cdot(\Op_2,t_2)\}\in
\mathcal{L}$ for every $g\in\mathcal{G}$;
\item \emph{Inclusion-monotonicity:} $\{(\Op_1,t_1),(\Op_2,t_2)\}\in
\mathcal{L}$ and $\emptyset\neq\Op_i'\subset\Op_i$ open imply
$\{(\Op_1',t_1),(\Op_2',t_2)\}\in\mathcal{L}$;
\item \emph{Non-triviality:} $\mathcal{L}\neq\emptyset$.
\end{enumerate}
Its set of \emph{realised time separations} is
\begin{equation}\label{eq:realised_time_separations}
T(\mathcal{L}):=\bigl\{\,|t_2-t_1|:\{(\Op_1,t_1),(\Op_2,t_2)\}\in
\mathcal{L}\,\bigr\}\subset(0,\infty),
\end{equation}
and the relation is \emph{realised at short times} if
$\inf T(\mathcal{L})=0$.
\end{definition}

\begin{remark}[Scope of the definition]\label{rem:locality_relation_scope}
(L1) and (L3) are satisfied by relativistic microcausality on the slice
regions of Minkowski spacetime, and (L2) replaces Poincar\'e invariance.
The definition therefore formalises ``microcausality with respect to
some Galilean-invariant notion of spacelike separation at unequal
times'' without prescribing that notion. A relation between regions of
$\R^3\times\R$ whose time projections overlap contains, after slicing,
pairs at arbitrarily small time separations. One with disjoint time
projections contains an open interval of time separations. Algebras of
open spacetime regions are treated directly in Corollary~\ref{cor:4D},
without assuming that slice algebras are contained in them.
\end{remark}

\section{Galilean transitivity and the kinematic exclusion}
\label{sec:hegerfeldt_rigidity}

\begin{lemma}[Galilean transitivity: no absolute space]
\label{lem:galilean_transitivity}
\begin{enumerate}[label=(\alph*)]
\item Let $e_i=(\bx_i,t_i)$ and $e_i'=(\bx_i',t_i')$, $i=1,2$, be events
with $t_2-t_1=t_2'-t_1'\neq 0$. Then there is $g\in\mathcal{G}$, a boost
composed with a translation, such that $g\cdot e_i=e_i'$ for $i=1,2$.
\item Let $\mathcal{L}$ be a Galilean unequal-time locality relation and
$\tau\in T(\mathcal{L})$. Then there is $r>0$ such that, for both signs,
\begin{equation}\label{eq:ball_commutation}
[\F(B_r(\bx),t),\,\F(B_r(\by),t\pm\tau)]=\{0\}
\qquad\text{for all }\bx,\by\in\R^3,\ t\in\R,
\end{equation}
where $B_r(\bx)$ is the open ball of radius $r$ about $\bx$.
\item If, in addition, the slice algebras are additive---$\F(\Op,t)$ is
generated by the algebras $\F(B,t)$ of the balls $B\subset\Op$ of radius
at most $\rho$, for every $\rho>0$---then, for every
$\tau\in T(\mathcal{L})$ and both signs,
$[\F(\Op_1,t),\F(\Op_2,t\pm\tau)]=\{0\}$ for all $t\in\R$ and all
non-empty bounded open $\Op_1,\Op_2\subset\R^3$, not necessarily
disjoint.
\end{enumerate}
\end{lemma}

\begin{proof}
(a) Put $\Delta t:=t_2-t_1$, $\bv:=[(\bx_2'-\bx_1')-(\bx_2-\bx_1)]/\Delta
t$, and let $g\cdot(\bx,t)=(\bx+\bv(t-t_1)+\bx_1'-\bx_1,\,t+t_1'-t_1)$.
Then $g\cdot e_1=e_1'$ and $g\cdot e_2=(\bx_2+\bv\Delta
t+\bx_1'-\bx_1,\,t_2')=e_2'$.

(b) Choose $\{(\Op_1^\circ,t_1^\circ),(\Op_2^\circ,t_2^\circ)\}\in
\mathcal{L}$ with $|t_2^\circ-t_1^\circ|=\tau$, points
$\bx_i^\circ\in\Op_i^\circ$ and $r>0$ with
$B_r(\bx_i^\circ)\subset\Op_i^\circ$. By (L3) the pair of balls
$\{(B_r(\bx_1^\circ),t_1^\circ),(B_r(\bx_2^\circ),t_2^\circ)\}$ belongs
to $\mathcal{L}$. Given $\bx,\by,t$ and a sign, label the pair so that
$t_2^\circ-t_1^\circ=\pm\tau$ and apply (a) to the events
$(\bx_i^\circ,t_i^\circ)$ and $(\bx,t)$, $(\by,t\pm\tau)$. The resulting
$g$ maps the two balls onto $(B_r(\bx),t)$ and $(B_r(\by),t\pm\tau)$.
By (L2) the image pair belongs to $\mathcal{L}$, and (L1)
gives~\eqref{eq:ball_commutation}.

(c) Every ball of radius at most $r$ lies in a ball of radius $r$, so by
isotony and~\eqref{eq:ball_commutation} each $\F(B,t)$ commutes with
each $\F(B',t\pm\tau)$ for balls $B,B'$ of radius at most $r$. Two
applications of additivity give the claim.
\end{proof}

Figure~\ref{fig:worldlines} shows the geometry of part~(a). In Galilean
spacetime every event at a time $t\neq t_1$ lies on a free-fall world
line through $e_1$; in Minkowski space only the events inside the light
cone do, and in the oscillating spacetime S10 of
Sec.~\ref{sec:classification} the free-fall world lines through $e_1$
refocus.

\begin{figure}[htb]
\centering
\begin{tikzpicture}[
  font=\small, >={Stealth[length=2mm]},
  wl/.style={blue!70!black, line width=0.5pt},
  wlr/.style={red!70!black, line width=0.9pt},
  slice/.style={gray, dashed, line width=0.4pt},
  ev/.style={circle, fill=black, inner sep=1.3pt},
  evr/.style={circle, fill=red!70!black, inner sep=1.4pt},
  evo/.style={circle, draw=red!70!black, fill=white, inner sep=1.4pt},
  ttl/.style={font=\small\bfseries, align=center}]
\def\figW{2.3}\def\figH{2.3}
\begin{scope}
  \begin{scope}
    \clip (-\figW,-0.05) rectangle (\figW,\figH);
    \fill[blue!6] (0,0) -- (3,3) -- (-3,3) -- cycle;
    \foreach \v in {-0.8,-0.55,-0.3,0,0.3,0.55,0.8}
      \draw[wl] (0,0) -- ({3*\v},3);
    \draw[gray!70!black, line width=0.6pt] (0,0) -- (3,3) (0,0) -- (-3,3);
  \end{scope}
  \draw[->, gray] (-\figW,0) -- (\figW+0.15,0) node[below left=-1pt, black] {$\bx$};
  \draw[->, gray] (0,-0.2) -- (0,\figH+0.2) node[left, black] {$ct$};
  \draw[slice] (-\figW,1.1) -- (\figW,1.1);
  \node[ev, label={[font=\footnotesize]below right:$e_1$}] at (0,0) {};
  \node[evr, label={[font=\footnotesize, red!70!black]below:$e_2$}] at (1.75,1.1) {};
  \node[font=\footnotesize, gray!60!black, rotate=-45] at (-1.52,1.3) {light cone};
  \node[ttl] at (0,\figH+0.6) {(a) Minkowski};
\end{scope}
\begin{scope}[xshift=5.5cm]
  \begin{scope}
    \clip (-\figW,-0.05) rectangle (\figW,\figH);
    \fill[blue!6] (-\figW,0) rectangle (\figW,\figH);
    \foreach \v in {-6,-3,-1,-0.55,-0.25,0,0.25,0.55,1,3,6}
      \draw[wl] (0,0) -- ({3*\v},3);
    \draw[wlr] (0,0) -- ({3*1.59},3);
  \end{scope}
  \draw[->, gray] (-\figW,0) -- (\figW+0.15,0) node[below left=-1pt, black] {$\bx$};
  \draw[->, gray] (0,-0.2) -- (0,\figH+0.2) node[left, black] {$t$};
  \draw[slice] (-\figW,1.1) -- (\figW,1.1);
  \node[ev, label={[font=\footnotesize]below right:$e_1$}] at (0,0) {};
  \node[evr, label={[font=\footnotesize, red!70!black]below right:$e_2$}] at (1.75,1.1) {};
  \node[evo] at (0,1.1) {};
  \draw[->, red!70!black, shorten >=3pt, shorten <=3pt] (1.75,1.1)
    to[bend right=40] node[above, font=\footnotesize] {boost} (0,1.1);
  \node[ttl] at (0,\figH+0.6) {(b) Galilean (S7)};
\end{scope}
\begin{scope}[xshift=11cm]
  \begin{scope}
    \clip (-\figW,-0.05) rectangle (\figW,\figH+0.05);
    \foreach \a in {-1.2,-0.8,-0.4,0,0.4,0.8,1.2}
      \draw[wl, domain=0:6.2832, samples=90, smooth, variable=\s]
        plot ({cos(\s r)+\a*sin(\s r)}, {\figH*\s/6.2832});
  \end{scope}
  \draw[->, gray] (-\figW,0) -- (\figW+0.15,0) node[below left=-1pt, black] {$\bx$};
  \draw[->, gray] (0,-0.2) -- (0,\figH+0.2) node[left, black] {$\omega t$};
  \draw[slice] (-\figW,\figH/2) -- (\figW,\figH/2) node[right, font=\footnotesize, black] {$\pi$};
  \draw[slice] (-\figW,\figH) -- (\figW,\figH) node[right, font=\footnotesize, black] {$2\pi$};
  \node[ev, label={[font=\footnotesize]below:$e_1$}] at (1,0) {};
  \node[ev, fill=blue!70!black] at (-1,\figH/2) {};
  \node[ev, fill=blue!70!black] at (1,\figH) {};
  \node[font=\footnotesize, blue!70!black, fill=white, inner sep=1pt, anchor=east]
    at (-1.62,\figH/2+0.28) {$-\bx_1$};
  \node[evr, label={[font=\footnotesize, red!70!black]above:$e_2$}] at (1.7,\figH/2) {};
  \node[ttl] at (0,\figH+0.6) {(c) galilean AdS (S10)};
\end{scope}
\end{tikzpicture}
\caption{Free-fall world lines through an event $e_1=(\bx_1,t_1)$: the
geodesics through $e_1$, that is, the images of one of them under the
boosts that fix $e_1$. Schematic, in one space dimension; only times
$t\geq t_1$ are shown. (a) In Minkowski space they fill the inside of
the light cone. An event $e_2$ spacelike to $e_1$ lies on none of them,
and no Poincar\'e transformation carries the pair $(e_1,e_2)$ to a pair
at coincident positions. (b) In Galilean spacetime (S7) they fill every
time slice $t\neq t_1$, and a boost that fixes $e_1$ carries $e_2$ to the
event above $e_1$ (Lemma~\ref{lem:galilean_transitivity}(a)). A
commutation relation at time separation $t_2-t_1$ therefore reaches
coincident positions (Lemma~\ref{lem:galilean_transitivity}(b)); if
this happens at arbitrarily short time separations, it contradicts the
canonical commutator (Theorem~\ref{thm:sr_selection}). (c) In the
galilean anti-de Sitter, or oscillating Newton--Hooke, spacetime S10 of
Table~\ref{tab:kinematics} they are the geodesics of the Newton--Hooke
connection, realised by the trajectories in the harmonic trap of
Proposition~\ref{prop:trap}, drawn for $t_1=0$. They refocus at
$(-1)^n\bx_1$ at the conjugate times $\omega t=n\pi$, where $w_M$
vanishes; an event such as $e_2$ at such a time lies on none of them,
and there slice algebras can commute.}
\label{fig:worldlines}
\end{figure}

\begin{theorem}[Galilean-exclusion theorem, kinematic form]
\label{thm:sr_selection}
Let a canonical net, in the creation--annihilation form or in the
$(\hat\phi,\hat\pi)$ form, admit a Galilean unequal-time locality
relation $\mathcal{L}$ realised at short times. Then:
\begin{enumerate}[label=(\roman*)]
\item for every non-empty bounded open $\Op\subset\R^3$ and every
non-zero $g\in C_c^\infty(\Op)$ (real, in the $(\hat\phi,\hat\pi)$
form),
$\braket{\Omega,[A_g,B_g]\Omega}=0$, where $(A_g,B_g)=
(\psio(g),\psiod(g))$ or $(\hat\phi_0(g),\hat\pi_0(g))$. This
contradicts~(S3), by which
$\braket{\Omega,[A_g,B_g]\Omega}=\hbar\|g\|^2$ or
$\rmi\hbar\|g\|^2$;
\item if the slice algebras are additive, every slice algebra is
abelian; this conclusion uses neither $\Omega$ nor~(S3).
\end{enumerate}
Consequently no such net exists. No spectral condition is assumed, and
the vacuum is used only through its time-translation invariance.
\end{theorem}

\begin{proof}
Choose $\tau_n\in T(\mathcal{L})$ with $\tau_n\to0$ and let $r_n$ be the
radius of Lemma~\ref{lem:galilean_transitivity}(b) for $\tau=\tau_n$.
We write the creation--annihilation form; the $(\hat\phi,\hat\pi)$ form
is identical.

(i) For $f,h\in C_c^\infty(\R^3)$ and $\tau\in\R$ put
\begin{equation}\label{eq:F_witness}
F_{f,h}(\tau):=\braket{\psiod(f)\Omega,\evol{\tau}\psiod(h)\Omega}
-\braket{\evol{\tau}\psio(h)\Omega,\psio(f)\Omega}.
\end{equation}
The function $F_{f,h}$ is continuous, and it is sesquilinear in $(f,h)$
because the fields are additive on $\D$. By~(S3) and the adjointness
relation in~(S1), $F_{g,g}(0)=\braket{\Omega,[\psio(g),\psiod(g)]\Omega}
=\hbar\|g\|^2$. Fix $n$ and a finite partition of unity subordinate to
a cover of $\supp g$ by balls $B_{r_n}(\by_j)$, so that
$g=\sum_jg_j$ with $g_j\in C_c^\infty(B_{r_n}(\by_j))$. Then
$F_{g,g}(\tau_n)=\sum_{j,k}F_{g_j,g_k}(\tau_n)$. The closure $A$ of
$\psio(g_j)$ is affiliated with $\F(B_{r_n}(\by_j),0)$, and $B:=
\evol{\tau_n}\overline{\psiod(g_k)}\evolm{\tau_n}$ is affiliated with
$\F(B_{r_n}(\by_k),\tau_n)$. These algebras commute
by~\eqref{eq:ball_commutation}. Since $\evolm{\tau_n}\Omega=\Omega\in
\D$, the vector $\Omega$ lies in the domains of $A$, $A^*$, $B$ and
$B^*$, with $A^*\Omega=\psiod(g_j)\Omega$,
$B\Omega=\evol{\tau_n}\psiod(g_k)\Omega$ and
$B^*\Omega=\evol{\tau_n}\psio(g_k)\Omega$. Lemma~\ref{lem:truncation}
gives $\braket{A^*\Omega,B\Omega}=\braket{B^*\Omega,A\Omega}$, that is,
$F_{g_j,g_k}(\tau_n)=0$. Hence $F_{g,g}(\tau_n)=0$ for all $n$, and by
continuity $F_{g,g}(0)=0$, contradicting $F_{g,g}(0)=\hbar\|g\|^2$.

(ii) By Lemma~\ref{lem:galilean_transitivity}(c) with $\Op_1=\Op_2=\Op$,
$[X,\evol{\tau_n}Y\evolm{\tau_n}]=0$ for all $X,Y\in\F(\Op,t)$ and all
$n$. As $\tau\to0$, $\evol{\tau}Y\evolm{\tau}\to Y$ strongly, so $[X,Y]=0$.
An abelian slice algebra contradicts~(S3): by Lemma~\ref{lem:truncation}
the vacuum expectation of the commutator of two closed operators
affiliated with an abelian von Neumann algebra vanishes.
\end{proof}

\begin{corollary}[No unequal-time locality in Galilean-covariant
canonical nets]\label{cor:no_unequal_time_locality}
Let a Galilean canonical net satisfy (S1)--(S3). If there are non-empty
bounded open sets $\Op_1^{(n)},\Op_2^{(n)}\subset\R^3$ and times
$t_1^{(n)}\neq t_2^{(n)}$ with $|t_2^{(n)}-t_1^{(n)}|\to 0$ and
\[
[\F(\Op_1^{(n)},t_1^{(n)}),\F(\Op_2^{(n)},t_2^{(n)})]=\{0\}
\qquad\text{for every }n,
\]
then the conclusions of Theorem~\ref{thm:sr_selection} hold, and no
such net exists.
\end{corollary}

\begin{proof}
The set of all unordered pairs of slice regions at unequal times whose
algebras commute satisfies (L1) by definition, (L3) by isotony, (L2) by
covariance~(S2), since automorphisms preserve commutation, and (L4)
together with $\inf T(\mathcal{L})=0$ by hypothesis.
\end{proof}

\begin{remark}[What is used, and why the Lorentzian case escapes]
\label{rem:where_positivity_enters}
Theorem~\ref{thm:sr_selection} uses the Galilean invariance of the
relation, the continuity of the time evolution, the canonical commutator
and a time-invariant vector in the domain of the fields, and nothing
else. The Poincar\'e group does not act
transitively on pairs of events with a fixed time separation: the
interval $|\Delta\bx|^2-c^2\Delta t^2$ separates its orbits, and it is
this invariant that lets relativistic microcausality be non-trivial at
unequal times. In the Carroll group spatial distance plays that role,
and groups without boosts have invariant cones of every slope.
Section~\ref{sec:classification} makes this precise for the spatially
isotropic homogeneous spacetimes with absolute time. The proof uses the
boosts only through~(L2), that is, through the invariance of the set of
commuting pairs. It therefore applies whenever the boosts preserve the
commutation of slice algebras, for instance when they act by
automorphisms of a faithfully represented C*-algebra of local fields
whose image generates the slice algebras, whether or not they are
unitarily implemented, as in states of non-zero density
(Sec.~\ref{sec:media}).
\end{remark}

\section{The vacuum sector under Bargmann covariance}
\label{sec:reeh_schlieder_strand}

\subsection{Hypotheses}

Continuous projective unitary representations of the Galilei group
lift to unitary representations of the Bargmann central extension of its
universal covering group~\cite{Bargmann1954,LevyLeblond1971}. We treat
scalar fields and rotations in $SO(3)$. For fields with spin, (B5)
and~(H2) involve $SU(2)$ and a matrix representation; the rotations are
used only in the second part of Lemma~\ref{lem:zero_mass_vacuum}, in
Lemma~\ref{lem:covariance} and in
Theorem~\ref{thm:commutation_set}(iii). We write $\evol{t}$ for
the time translations, $T(\ba)=\rme^{-\rmi\ba\cdot\hP/\hbar}$ for the
spatial translations, $B(\bv)=\rme^{\rmi\bv\cdot\hK/\hbar}$ for the
boosts, $U(R)$ for the rotations and $\rme^{\rmi\theta\hM/\hbar}$ for
the central one-parameter group generated by the mass operator $\hM$.
In the following, a \emph{unitary Bargmann representation} is a family
of strongly continuous unitary groups of this kind satisfying
\begin{enumerate}[label=(B\arabic*)]
\item $\evol{t}$, $T(\ba)$ and $\rme^{\rmi\theta\hM/\hbar}$ commute with
each other, and $\rme^{\rmi\theta\hM/\hbar}$ commutes with $B(\bv)$ and
$U(R)$;
\item $B(\bv)B(\mathbf{w})=B(\bv+\mathbf{w})$;
\item $B(\bv)^{-1}T(\ba)B(\bv)=\rme^{-\rmi\,\ba\cdot\bv\,\hM/\hbar}\,
T(\ba)$;
\item $B(\bv)^{-1}\evol{t}B(\bv)=\rme^{\rmi|\bv|^2t\hM/2\hbar}\,
T(-\bv t)\,\evol{t}$;
\item $U(R)$ commutes with $\evol{t}$, $U(R)T(\ba)U(R)^{-1}=T(R\ba)$ and
$U(R)B(\bv)U(R)^{-1}=B(R\bv)$.
\end{enumerate}
These are the Weyl form of the commutation relations
$[\hat K_j,\hat P_k]=\rmi\hbar\delta_{jk}\hM$ and
$[\hat K_j,\hH]=\rmi\hbar\hat P_j$. In this convention $B(\bv)$ raises
the velocity by $\bv$: in generator form,
$B(\bv)^{-1}\hP B(\bv)=\hP+\hM\bv$ and
$B(\bv)^{-1}\hH B(\bv)=\hH+\bv\cdot\hP+\tfrac12|\bv|^2\hM$. For a
single particle of mass $m_0$ on $L^2(\R^3)$, $T(\ba)$ is the
translation $\tau_\ba g=g_\ba$, $B(\bv)$ is multiplication by
$e_\bv(\bx):=\rme^{\rmi m_0\bv\cdot\bx/\hbar}$, and $\hH=h_0:=
-\hbar^2\Delta/2m_0$.

\begin{definition}[Hypotheses (H)]\label{def:H}
A canonical net satisfies (H) if the following hold.
\begin{enumerate}[label=(H\arabic*)]
\item $\hH$ is the generator of the time translations of a unitary
Bargmann representation, and $\D$ is invariant under $T(\ba)$, $B(\bv)$,
$U(R)$ and $\rme^{\rmi\theta\hM/\hbar}$.
\item On $\D$,
\begin{gather*}
T(\ba)\psio(g)T(\ba)^{-1}=\psio(g_\ba),\qquad
B(\bv)\psio(g)B(\bv)^{-1}=\psio(e_\bv g),\qquad
U(R)\psio(g)U(R)^{-1}=\psio(g\circ R^{-1}),\\
\rme^{\rmi\theta\hM/\hbar}\psio(g)\rme^{-\rmi\theta\hM/\hbar}=
\rme^{-\rmi\theta m_0/\hbar}\psio(g),
\end{gather*}
for some $m_0>0$, the Bargmann charge of the field.
\item $\Omega$ is invariant under $\evol{t}$, $T(\ba)$, $B(\bv)$ and
$U(R)$.
\item There is $\lambda\in\R$ such that $\hH+\lambda\hM\geq 0$.
\end{enumerate}
\end{definition}

The operator in (H4) is defined by the joint spectral calculus of the
commuting self-adjoint operators $\hH$ and $\hM$. By the adjointness
relation in~(S1), (H2) implies the corresponding relations for $\psiod$,
with $\rme^{+\rmi\theta m_0/\hbar}$ in the last one. In particular
$\psiod(g)$ maps $\D\cap\ker(\hM-m)$ into $\ker(\hM-m-m_0)$.
Hypothesis~(H4) is discussed in Remark~\ref{rem:central_shift}.

\subsection{Mass, covariance and the vacuum}

\begin{lemma}[Translation-invariant vectors have zero Bargmann mass]
\label{lem:zero_mass_vacuum}
Let (B1)--(B3) hold and let $\Omega\in\Hil$ be invariant under all
spatial translations. Then $\hM\Omega=0$. If, moreover, (B5) holds,
$\Omega$ is, up to scalars, the only vector invariant under the spatial
translations, and it is invariant under rotations, then
$B(\bv)\Omega=\Omega$ for all $\bv$.
\end{lemma}

\begin{proof}
By (B3), $T(\ba)B(\bv)=B(\bv)\rme^{-\rmi\ba\cdot\bv\hM/\hbar}T(\ba)$.
Hence $\braket{\Omega,B(\bv)\Omega}=\braket{T(\ba)\Omega,
T(\ba)B(\bv)\Omega}=\braket{\Omega,B(\bv)\rme^{-\rmi\ba\cdot\bv
\hM/\hbar}\Omega}$ for every $\ba$. Take $\bv\neq0$, $\ba=s\bv$, and
average over $s\in[0,S]$. As $S\to\infty$ the averages of
$\rme^{-\rmi s|\bv|^2\hM/\hbar}$ converge strongly to the spectral
projection $\Pi_0$ of $\hM$ onto $\ker\hM$, by the spectral theorem and
dominated convergence. Thus $\braket{\Omega,B(\bv)\Omega}=\braket{\Omega,
B(\bv)\Pi_0\Omega}$ for all $\bv\neq0$. Letting $\bv\to0$ and using strong
continuity gives $\|\Omega\|^2=\|\Pi_0\Omega\|^2$, that is, $\hM\Omega=0$.
For the second statement, $B(\bv)\Omega\in\ker\hM$ because $\hM$ is
central, and on $\ker\hM$ relation (B3) says that $B(\bv)$ commutes with
the spatial translations. Hence $B(\bv)\Omega$ is invariant under them,
and by uniqueness $B(\bv)\Omega=c(\bv)\Omega$. By (B2) and strong
continuity $c$ is a continuous character of $\R^3$, so
$c(\bv)=\rme^{\rmi\mathbf{w}\cdot\bv}$ for some $\mathbf{w}\in\R^3$. By
(B5) and rotation invariance, $c(R\bv)=\braket{U(R)^{-1}\Omega,
B(\bv)U(R)^{-1}\Omega}=c(\bv)$ for all $R$, which forces $\mathbf{w}=0$.
\end{proof}

L\'evy-Leblond observed that a non-degenerate Galilean-invariant vacuum
has zero mass~\cite{LevyLeblond1967}. Lemma~\ref{lem:zero_mass_vacuum}
needs only invariance under spatial translations.

\begin{lemma}[Covariance of slice algebras]\label{lem:covariance}
Let (S1) and (H1)--(H2) hold. Then
\begin{gather*}
T(\ba)\F(\Op,t)T(\ba)^{-1}=\F(\Op+\ba,t),\qquad
B(\bv)\F(\Op,t)B(\bv)^{-1}=\F(\Op+\bv t,t),\qquad
U(R)\F(\Op,t)U(R)^{-1}=\F(R\Op,t),\\
\rme^{\rmi\theta\hM/\hbar}\F(\Op,t)\rme^{-\rmi\theta\hM/\hbar}=
\F(\Op,t),\qquad \evol{s}\F(\Op,t)\evolm{s}=\F(\Op,t+s).
\end{gather*}
In particular (S2) holds.
\end{lemma}

\begin{proof}
Conjugation by a unitary $W$ maps the von Neumann algebra generated by
a family of closed operators onto the one generated by the conjugated
family, and conjugation commutes with closure. Scalar factors of
modulus one do not change the generated algebra. For $W\in\{T(\ba),
U(R),\rme^{\rmi\theta\hM/\hbar}\}$ and $t=0$ the claims therefore follow
from (H2), since $\D$ is invariant under $W$. For $B(\bv)$ at $t=0$,
(H2) gives $B(\bv)\F(\Op,0)B(\bv)^{-1}=\F(\Op,0)$, because
multiplication by $e_\bv$ does not change supports. The unitaries $T(\ba)$,
$U(R)$ and $\rme^{\rmi\theta\hM/\hbar}$ commute with $\evol{t}$, which
gives the claims at time $t$. For the boosts, (B4) with $\bv$ replaced
by $-\bv$ and (B2) give $B(\bv)\evol{t}B(\bv)^{-1}=
\rme^{\rmi|\bv|^2t\hM/2\hbar}T(\bv t)\evol{t}$. Hence
\[
B(\bv)\F(\Op,t)B(\bv)^{-1}=\rme^{\rmi|\bv|^2t\hM/2\hbar}T(\bv t)\,
\F(\Op,t)\,T(\bv t)^{-1}\rme^{-\rmi|\bv|^2t\hM/2\hbar}
=\F(\Op+\bv t,t).
\]
Every $g\in\mathcal{G}$ is a product of a translation, a boost, a
rotation and a time translation, and the corresponding product of
unitaries implements $\alpha_g$.
\end{proof}

\begin{lemma}[Joint spectrum]\label{lem:joint_spectrum}
Let (B1)--(B4) and (H4) hold, and let $E$ be the joint spectral measure
on $\R^5$ of the commuting self-adjoint operators
$(\hH,\hat P_1,\hat P_2,\hat P_3,\hM)$. Then $E$ is supported in
\begin{equation}\label{eq:S_lambda}
S_\lambda:=\bigl\{(\varepsilon,\bp,\mu):\ \mu\geq0,\
\varepsilon+\lambda\mu\geq0,\ 2\mu(\varepsilon+\lambda\mu)\geq|\bp|^2
\bigr\}.
\end{equation}
Consequently $\spec\hM\subset[0,\infty)$, $\hP=0$ on $\ker\hM$, and
$\hH-|\hP|^2/2m\geq-\lambda m$ on $\ker(\hM-m)$ for every $m>0$.
\end{lemma}

\begin{proof}
For $\bv\in\R^3$ let $\Phi_\bv(\varepsilon,\bp,\mu):=(\varepsilon+\bv\cdot
\bp+\tfrac12|\bv|^2\mu,\ \bp+\mu\bv,\ \mu)$. By (B1), (B3) and (B4),
\[
B(\bv)^{-1}\rme^{\rmi(t\hH-\ba\cdot\hP+\theta\hM)/\hbar}B(\bv)
=\rme^{\rmi(t(\hH+\bv\cdot\hP+\frac12|\bv|^2\hM)-\ba\cdot(\hP+\bv\hM)
+\theta\hM)/\hbar}
\]
for all $(t,\ba,\theta)$. Taking expectation values in any vector and
using the uniqueness of a finite measure with given Fourier transform, we
get $B(\bv)^{-1}E(\Delta)B(\bv)=E(\Phi_\bv^{-1}(\Delta))$ for every
Borel set $\Delta$. By (H4), $E$ is supported in $C:=\{\varepsilon+
\lambda\mu\geq0\}$. Hence $E(\Phi_\bv^{-1}(\R^5\setminus C))=
B(\bv)^{-1}E(\R^5\setminus C)B(\bv)=0$, so $E$ is supported in
$\Phi_\bv^{-1}(C)$ for every $\bv$. Intersecting over $\bv\in\mathbb{Q}^3$,
and using continuity in $\bv$, $E$ is supported in the set of
$(\varepsilon,\bp,\mu)$ with $\varepsilon+\bv\cdot\bp+
(\frac12|\bv|^2+\lambda)\mu\geq0$ for all $\bv\in\R^3$. For $\mu<0$ this
fails for large $|\bv|$. For $\mu=0$ it holds if and only if $\bp=0$ and
$\varepsilon\geq0$. For $\mu>0$ the minimum over $\bv$ is
$\varepsilon+\lambda\mu-|\bp|^2/2\mu$. This set is $S_\lambda$.
\end{proof}

\begin{lemma}[Vacuum annihilation]\label{lem:vacuum_annihilation}
Under (H1)--(H4), $\hM\Omega=0$ and $\psio(g)\Omega=0$ for every
$g\in C_c^\infty(\R^3)$. Moreover $\psiod(g)\Omega\in\Hil_{m_0}:=
\ker(\hM-m_0)$ and $\|\psiod(g)\Omega\|^2=\hbar\|g\|^2$.
\end{lemma}

\begin{proof}
$\hM\Omega=0$ by Lemma~\ref{lem:zero_mass_vacuum}. By (H2),
$\rme^{\rmi\theta\hM/\hbar}\psio(g)\Omega=\rme^{-\rmi\theta
m_0/\hbar}\psio(g)\Omega$ for all $\theta$, so $\psio(g)\Omega\in
\ker(\hM+m_0)$. This space is $\{0\}$ by Lemma~\ref{lem:joint_spectrum},
since $m_0>0$. Likewise $\psiod(g)\Omega\in\ker(\hM-m_0)$, and by~(S3)
$\|\psiod(g)\Omega\|^2=\braket{\Omega,\psio(g)\psiod(g)\Omega}=
\hbar\|g\|^2+\|\psio(g)\Omega\|^2=\hbar\|g\|^2$.
\end{proof}

Lemma~\ref{lem:vacuum_annihilation} is Proposition~3 of
Ref.~\onlinecite{Dadashev1985}, where non-negativity of the particle
number is an axiom; here it follows from (H4) and boost covariance
through Lemma~\ref{lem:joint_spectrum}. When the fields integrate to a
Weyl system, with $\D$ a core for the fields, $\hM/m_0$ is a number
operator for it in the sense of
Chaiken, bounded below, and the Weyl system is a direct sum of Fock
representations (Ref.~\onlinecite{Chaiken1968}, Theorem~1.3).

\begin{proposition}[Hermitian canonical fields]
\label{prop:hermitian_zero_mass}
Let a unitary Bargmann representation satisfy (H4), and let $\Omega$ be
a unit vector invariant under the spatial translations. Let
$\hat\phi_0(f)$, $\hat\pi_0(f)$, $f\in C_c^\infty(\R^3)$ real, be
symmetric operators with a common dense domain $\D\ni\Omega$, invariant
under $T(\ba)$ and $\rme^{\rmi\theta\hM/\hbar}$, such that: (a) they map
$\D$ into itself and are real-linear in $f$; (b) $T(\ba)\hat\phi_0(f)
T(\ba)^{-1}=\hat\phi_0(f_\ba)$ and $T(\ba)\hat\pi_0(f)T(\ba)^{-1}=
\hat\pi_0(f_\ba)$ on $\D$; (c) they commute with
$\rme^{\rmi\theta\hM/\hbar}$ on $\D$, that is, they carry no Bargmann
charge; (d) $f\mapsto\hat\phi_0(f)\Omega$ and
$f\mapsto\hat\pi_0(f)\Omega$ are continuous from $C_c^\infty(\R^3)$,
with its inductive-limit topology, to $\Hil$; (e)
$[\hat\phi_0(f),\hat\pi_0(g)]=\rmi\hbar\,(f,g)$ on $\D$. Then no such
system exists.
\end{proposition}

\begin{proof}
$\hM\Omega=0$ by Lemma~\ref{lem:zero_mass_vacuum}. By (c),
$\hat\phi_0(f)\Omega$ and $\hat\pi_0(f)\Omega$ lie in $\ker\hM$, where
$\hP=0$ by Lemma~\ref{lem:joint_spectrum}. The translations therefore
act trivially on these vectors, and by (b)
$\hat\phi_0(f_\ba)\Omega=T(\ba)\hat\phi_0(f)\Omega=\hat\phi_0(f)\Omega$
for all $\ba$. Let $f$ be real with $\int f=0$. By
Lemma~\ref{lem:divergence}, $f=\sum_i\partial_ih_i$ with real
$h_i\in C_c^\infty(\R^3)$, and $\partial_ih_i$ is the limit in
$C_c^\infty(\R^3)$ of $\epsilon^{-1}(h_{i,-\epsilon\mathbf{e}_i}-h_i)$
as $\epsilon\to0$. Each difference quotient is annihilated by
$\hat\phi_0(\cdot)\Omega$, so by (d) $\hat\phi_0(f)\Omega=0$, and likewise
$\hat\pi_0(f)\Omega=0$. Hence $\braket{\Omega,[\hat\phi_0(f),\hat\pi_0(f)]
\Omega}=\braket{\hat\phi_0(f)\Omega,\hat\pi_0(f)\Omega}-
\braket{\hat\pi_0(f)\Omega,\hat\phi_0(f)\Omega}=0$, whereas (e) gives
$\rmi\hbar\|f\|^2\neq0$ for $f\neq0$.
\end{proof}

\begin{remark}[Hermitian fields in non-relativistic physics]
\label{rem:phi_pi_open}
Proposition~\ref{prop:hermitian_zero_mass} uses no locality: over a
translation-invariant vacuum, chargeless canonical fields would create
massless Galilean excitations, and by Lemma~\ref{lem:joint_spectrum}
these carry no momentum. L\'evy-Leblond~\cite{LevyLeblond1967} and
Dadashev (Ref.~\onlinecite{Dadashev1985}, Proposition~1) had already
observed that Hermitian Galilean fields carry no mass. The Hermitian Bose
fields of non-relativistic physics---phonons, sound, surface
waves---describe excitations of a medium, whose state is not invariant
under unitarily implemented boosts; they lie outside the hypotheses
(Sec.~\ref{sec:media}).
\end{remark}

\begin{corollary}[The vacuum is not separating]\label{thm:galilean_RS}
Under (H1)--(H4), for every non-empty bounded open $\Op\subset\R^3$ and
every $t\in\R$ there is a non-zero projection $Q\in\F(\Op,t)$ with
$Q\Omega=0$. Thus $\Omega$ is not separating for $\F(\Op,t)$, and the
same holds for the algebras $\F(\Op)$ of open spacetime regions under
hypothesis~(T) of Corollary~\ref{cor:4D}.
\end{corollary}

\begin{proof}
Take $g\in C_c^\infty(\Op)$, $g\neq0$, and let $A$ be the closure of
$\evol{t}\psio(g)\evolm{t}$, affiliated with $\F(\Op,t)$. By
Lemma~\ref{lem:vacuum_annihilation}, $A\Omega=0$. Also $A\neq0$, since
$\psio(g)\psiod(g)\Omega=\hbar\|g\|^2\Omega\neq0$. Let $Q$ be the
spectral projection of $|A|$ for $(0,\infty)$. Then $Q\in\F(\Op,t)$,
$Q\neq0$, and $Q\Omega=0$ because $|A|\Omega=0$. For $\F(\Op)$ the same
argument applies to the operators $\hat\psi(\chi\otimes g)$ of~(T).
These annihilate $\Omega$, and they are non-zero for suitable $\chi$,
because their adjoints map $\Omega$ to $\int\chi(t)\evol{t}
\psiod(g)\Omega\,\rmd t$.
\end{proof}

The Tomita--Takesaki construction needs a cyclic and separating vector.
The vacuum therefore defines no modular objects for the local field
algebras of these nets, while faithful states such as KMS states
do~\cite{HaagHugenholtzWinnink1967}. Cyclicity of the vacuum for regions
of non-zero duration is a separate matter~\cite{Requardt1982,Bain2011}.

\subsection{The Galilean K\"all\'en--Lehmann factorisation}

\begin{lemma}[Galilean K\"all\'en--Lehmann factorisation]\label{lem:GKL}
Assume (H1)--(H4), let $\Hil_{m_0}=\ker(\hM-m_0)$ and $h_0=
-\hbar^2\Delta/2m_0$ on $L^2(\R^3)$. There are a Hilbert space $\K$, a
self-adjoint operator $U_\K\geq-\lambda m_0$ on $\K$, a unitary
$W:\Hil_{m_0}\to L^2(\R^3)\otimes\K$ and a unit vector $\xi\in\K$ such
that
\begin{enumerate}[label=(\roman*)]
\item $WT(\ba)W^{-1}=\tau_\ba\otimes\one$, $WB(\bv)W^{-1}=e_\bv\otimes
\one$ and $W\evol{t}W^{-1}=\rme^{\rmi th_0/\hbar}\otimes
\rme^{\rmi tU_\K/\hbar}$ on $\Hil_{m_0}$;
\item $W\psiod(g)\Omega=\sqrt\hbar\,g\otimes\xi$ for every $g\in
C_c^\infty(\R^3)$.
\end{enumerate}
Consequently, with $\nu_\xi$ the spectral measure of $U_\K$ in $\xi$, a
probability measure on $[-\lambda m_0,\infty)$, and
$\widehat{\nu_\xi}(s):=\int\rme^{\rmi Es}\,\rmd\nu_\xi(E)$,
\begin{equation}\label{eq:GKL}
\braket{\psiod(f)\Omega,\evol{\tau}\psiod(g)\Omega}
=\hbar\,\bigl(f,\rme^{\rmi\tau h_0/\hbar}g\bigr)\,
\widehat{\nu_\xi}(\tau/\hbar)
\qquad(f,g\in C_c^\infty(\R^3),\ \tau\in\R).
\end{equation}
If $\psiod(g)\Omega\in D(\hH)$ for some $g\neq0$, then $\xi\in D(U_\K)$.
\end{lemma}

\begin{proof}
\emph{Weyl system.} The subspace $\Hil_{m_0}$ is non-zero by
Lemma~\ref{lem:vacuum_annihilation}. It reduces the representation,
since $\hM$ is central, and on it (B3) reads $T(\ba)B(\bv)=
\rme^{-\rmi m_0\ba\cdot\bv/\hbar}B(\bv)T(\ba)$. The same relation holds
for the Schr\"odinger pair $(\tau_\ba,e_\bv)$ on $L^2(\R^3)$. By the
Stone--von Neumann theorem, which needs no separability
assumption~\cite{vonNeumann1931,Folland1989,Rosenberg2004}, there is a
unitary $W:\Hil_{m_0}\to L^2(\R^3)\otimes\K$ with
$WT(\ba)W^{-1}=\tau_\ba\otimes\one$ and $WB(\bv)W^{-1}=e_\bv\otimes\one$.

\emph{The created vector.} Put $Vg:=\hbar^{-1/2}W\psiod(g)\Omega$. By
Lemma~\ref{lem:vacuum_annihilation}, $V$ is isometric on
$C_c^\infty(\R^3)$ and extends to an isometry of $L^2(\R^3)$. By (H2)
for $\psiod$ and (H3), $V\tau_\ba=(\tau_\ba\otimes\one)V$ and
$Ve_\bv=(e_\bv\otimes\one)V$. Lemma~\ref{lem:schur} gives a unit vector
$\xi\in\K$ with $Vg=g\otimes\xi$.

\emph{Internal energy.} By Lemma~\ref{lem:joint_spectrum}, the joint
spectral measure of $(\hH,\hP)$ on $\Hil_{m_0}$ is supported in
$\{\varepsilon-|\bp|^2/2m_0\geq-\lambda m_0\}$. Let $\Uint:=\hH-
|\hP|^2/2m_0$ on $\Hil_{m_0}$, defined by the joint spectral calculus.
It is self-adjoint and bounded below by $-\lambda m_0$. It commutes with
$T(\ba)$. It also commutes with $B(\bv)$: in the proof of
Lemma~\ref{lem:joint_spectrum}, $B(\bv)^{-1}u(\hH,\hP,\hM)B(\bv)=
(u\circ\Phi_\bv)(\hH,\hP,\hM)$ for bounded Borel $u$, and
$(\varepsilon+\bv\cdot\bp+\frac12|\bv|^2m_0)-|\bp+m_0\bv|^2/2m_0=
\varepsilon-|\bp|^2/2m_0$. The spectral projections of $W\Uint W^{-1}$
therefore commute with $\tau_\ba\otimes\one$ and $e_\bv\otimes\one$.
The Schr\"odinger pair is irreducible~\cite{Folland1989}, so these
operators generate $\mathcal{B}(L^2(\R^3))\otimes\one$, whose commutant
is $\one\otimes\mathcal{B}(\K)$. Hence $W\Uint W^{-1}=\one\otimes U_\K$
with $U_\K\geq-\lambda m_0$. Since $WT(\ba)W^{-1}=
\rme^{-\rmi\ba\cdot\bp/\hbar}\otimes\one$ with $\bp=-\rmi\hbar\nabla$,
we have $W|\hP|^2W^{-1}/2m_0=h_0\otimes\one$, and
$W\hH W^{-1}=h_0\otimes\one+\one\otimes U_\K$ (strongly commuting sum).
This proves~(i), and~(ii) is the previous step. Equation~\eqref{eq:GKL}
follows from $\braket{g_1\otimes\xi,(\rme^{\rmi\tau h_0/\hbar}\otimes
\rme^{\rmi\tau U_\K/\hbar})g_2\otimes\xi}=(g_1,\rme^{\rmi\tau h_0/\hbar}
g_2)\braket{\xi,\rme^{\rmi\tau U_\K/\hbar}\xi}$.

\emph{Moments.} Put $A:=h_0\otimes\one\geq0$ and $B:=\one\otimes
(U_\K+\lambda m_0)\geq0$. They commute strongly, and for a vector $x$
with joint spectral measure $\mu_x$ on $[0,\infty)^2$ one has
$\|(A+B)x\|^2=\int(\alpha+\beta)^2\rmd\mu_x\geq\int\beta^2\rmd\mu_x=
\|Bx\|^2$. Hence $D(A+B)\subset D(B)$. If $\psiod(g)\Omega\in D(\hH)$,
then $g\otimes\xi\in D(B)$, and so $\xi\in D(U_\K)$.
\end{proof}

The measure $\nu_\xi$ depends only on the net and not on the choice of
$W$: by~\eqref{eq:GKL}, $\widehat{\nu_\xi}(\tau/\hbar)$ is the ratio of
the vacuum correlation on the left to $\hbar(f,\rme^{\rmi\tau h_0/\hbar}
g)$ for any $f,g$ for which the latter does not vanish, and such $f,g$
exist by Lemma~\ref{lem:free_nonvanishing}.

\begin{remark}[Positivity modulo the central charge]\label{rem:central_shift}
Replacing $\hH$ by $\hH+\lambda\hM$ is an automorphism of the Bargmann
Lie algebra, because $\hM$ is central. The unitary groups
$\evol{t}\rme^{\rmi\lambda t\hM/\hbar}$, $T(\ba)$, $B(\bv)$, $U(R)$ and
$\rme^{\rmi\theta\hM/\hbar}$ again satisfy (B1)--(B5); relation (B4) is
unaffected because $\hM$ is central. They define the same slice
algebras, because $\rme^{\rmi\lambda t\hM/\hbar}$ acts on the time-zero
fields by phases. Hypothesis~(H4) therefore says that the energy is
non-negative up to this automorphism, that is, up to a choice of the
zero of the internal energy that is linear in the
mass~\cite{LevyLeblond1971}. A non-affine function $f(\hM)$ would not do:
since $\psio(g)$ lowers the mass by $m_0$, $\rme^{\rmi tf(\hM)/\hbar}
\psio(g)\rme^{-\rmi tf(\hM)/\hbar}=\psio(g)\,\rme^{\rmi t[f(\hM-m_0)-
f(\hM)]/\hbar}$, a phase that depends on $\hM$, and the slice algebras
at $t\neq0$ change in general. The shift moves $U_\K$ by $\lambda m_0$
and multiplies $\widehat{\nu_\xi}(s)$ by $\rme^{\rmi\lambda m_0s}$. It
leaves unchanged the zero set of $\widehat{\nu_\xi}$, the standard
deviation $\Delta_\xi U$ and the difference $\braket{U}_\xi-E_0$ that
enter Theorem~\ref{thm:commutation_set}. All statements below are
invariant under it.

The proofs of Lemmas~\ref{lem:joint_spectrum}--\ref{lem:GKL},
Proposition~\ref{prop:hermitian_zero_mass},
Corollary~\ref{thm:galilean_RS} and
Theorem~\ref{thm:commutation_set}(i)--(iii) use~(H4) only through the
support of the joint spectral measure of $(\hH,\hM)$ in $\{\varepsilon+
\ell(\mu)\geq0\}$ for some real Borel function $\ell$, that is, through a
lower bound on the energy in each mass sector. They hold verbatim under
that weaker hypothesis, with $\lambda\mu$ replaced by $\ell(\mu)$ in
$S_\lambda$ and $\lambda m_0$ by $\ell(m_0)$ in Lemma~\ref{lem:GKL}. It
covers many-body systems with bounded pair interactions that are not
H-stable, that is, whose energy on $n$ particles is not bounded below by
$-Bn$ for any constant $B$ (Corollary~\ref{cor:pair}).
\end{remark}

\begin{remark}[Relation to the boost Ward identity]
For fields covariant under the Schr\"odinger group, the boost Ward
identity fixes the two-point function up to a function of the time
difference~\cite{Hagen1972,Henkel1994}; for a primary operator of scaling
dimension $\Delta_{\mathcal{O}}$ it is $t^{-\Delta_{\mathcal{O}}}$ times
the free Gaussian phase~\cite{NishidaSon2007}.
Equation~\eqref{eq:GKL} is the corresponding operator statement for a
canonical field under the Bargmann group alone, with the function of
time identified as $\widehat{\nu_\xi}$. When the mass-$m_0$ sector
contains several particle species, or bound and scattering states of
the same total mass, $\nu_\xi$ records how $\psiod$ distributes the
created state over their internal energies.
\end{remark}

\section{Commutation sets}\label{sec:commutation_sets}

Under (H), let $E_0:=\min\supp\nu_\xi$, let $\braket{U}_\xi:=\int
E\,\rmd\nu_\xi(E)\in(-\infty,\infty]$ be the mean of $U_\K$ in $\xi$, and
let $\Delta_\xi U\in[0,\infty]$ be its standard deviation, with
$\Delta_\xi U:=\infty$ if $\braket{U}_\xi=\infty$. Both are finite if
$\psiod(g)\Omega\in D(\hH)$ for some $g\neq0$ (Lemma~\ref{lem:GKL}). Since $\nu_\xi$ is a real measure,
$\widehat{\nu_\xi}(-s)=\overline{\widehat{\nu_\xi}(s)}$.

\begin{theorem}[Commutation sets of Galilean canonical fields]
\label{thm:commutation_set}
Assume (H1)--(H4), and let
$Z_\xi:=\{\tau>0:\ \widehat{\nu_\xi}(\tau/\hbar)=0\}$.
\begin{enumerate}[label=(\roman*)]
\item If $\Op_1,\Op_2\subset\R^3$ are non-empty, bounded and open,
$t_1\neq t_2$ and $[\F(\Op_1,t_1),\F(\Op_2,t_2)]=\{0\}$, then
$|t_2-t_1|\in Z_\xi$. In particular $T(\mathcal{L})\subset Z_\xi$ for
every Galilean unequal-time locality relation, and more generally for
every set of commuting pairs of slice regions, invariant or not.
\item $Z_\xi$ is closed, has Lebesgue measure zero, and $\inf
Z_\xi>0$; moreover,
\begin{equation}\label{eq:QSL}
\inf Z_\xi\;\geq\;\max\Bigl\{\frac{\pi\hbar}{2\,\Delta_\xi U},\
\frac{\pi\hbar}{2\,(\braket{U}_\xi-E_0)}\Bigr\},
\end{equation}
with the conventions $\pi\hbar/0=\infty$ and $\pi\hbar/\infty=0$.
\item If $\xi$ is an eigenvector of $U_\K$---in particular, if the
Bargmann representation on $\Hil_{m_0}$ is irreducible---then
$Z_\xi=\emptyset$: no two non-empty slice algebras at unequal times
commute.
\item (Optimality.) Let $\nu$ be a probability measure on $\R$ with
$\inf\supp\nu>-\infty$. There is a net satisfying (H1)--(H4) with
$\nu_\xi=\nu$ in which, for every $\tau\neq0$,
$[\F(\Op_1,0),\F(\Op_2,\tau)]=\{0\}$ holds for all non-empty bounded
open $\Op_1,\Op_2$ if $\widehat\nu(\tau/\hbar)=0$, and for no such pair
otherwise.
\end{enumerate}
\end{theorem}

\begin{proof}
(i) Conjugating by $\evolm{t_1}$ we may take $t_1=0$ and $\tau:=t_2$.
Let $f\in C_c^\infty(\Op_1)$ and $g\in C_c^\infty(\Op_2)$, and let $A$ be
the closure of $\psio(f)$ and $B:=\evol{\tau}\overline{\psiod(g)}
\evolm{\tau}$. They are affiliated with the commuting algebras
$\F(\Op_1,0)$ and $\F(\Op_2,\tau)$, and $\Omega$ lies in the domains of
$A,A^*,B,B^*$, with $A\Omega=0$ and $B^*\Omega=\evol{\tau}\psio(g)
\Omega=0$ by Lemma~\ref{lem:vacuum_annihilation}.
Lemma~\ref{lem:truncation} gives
$\braket{\psiod(f)\Omega,\evol{\tau}\psiod(g)\Omega}=0$. By
Lemma~\ref{lem:GKL}, $(f,\rme^{\rmi\tau h_0/\hbar}g)\,
\widehat{\nu_\xi}(\tau/\hbar)=0$ for all such $f,g$. By
Lemma~\ref{lem:free_nonvanishing} the first factor is non-zero for
suitable $f,g$, so $\widehat{\nu_\xi}(\tau/\hbar)=0$.

(ii) $Z_\xi$ is closed because $\widehat{\nu_\xi}$ is continuous. The
function $G(z):=\int\rme^{\rmi(E-E_0)z}\rmd\nu_\xi(E)$ is holomorphic in
$\{\mathrm{Im}\,z>0\}$, bounded by $1$ and continuous on the closed
half-plane, and $|G(s)|=|\widehat{\nu_\xi}(s)|$ for real $s$. Since
$G(0)=1$, $G\not\equiv0$. The boundary values of a bounded holomorphic
function on a half-plane that is not identically zero vanish only on a
set of measure zero (theorem of F.~and M.~Riesz; see
Ref.~\onlinecite{Koosis1998}). Hence $Z_\xi$ is null. Continuity and
$\widehat{\nu_\xi}(0)=1$ give $\inf Z_\xi>0$. Lemma~\ref{lem:speed_limits},
applied to $\nu_\xi$ at $s=\tau/\hbar$, gives~\eqref{eq:QSL} when the
moments are finite; otherwise the corresponding bound is $0$.

(iii) If $U_\K\xi=E\xi$, then $\nu_\xi=\delta_E$ and
$|\widehat{\nu_\xi}|\equiv1$. If the representation on $\Hil_{m_0}$ is
irreducible, $\Uint$ commutes with it, since it is invariant under
$T(\ba)$, $B(\bv)$, $U(R)$, $\evol{t}$ and the centre. By Schur's lemma
$\Uint$ is then a multiple of the identity, and $\xi$ is an eigenvector.

(iv) Let $\K:=L^2(\R,\nu)$, let $U_\K$ be multiplication by $E$ and
$\xi\equiv1$. On the one-particle space $\mathfrak{h}:=L^2(\R^3)\otimes
\K$ let $h:=h_0\otimes\one+\one\otimes U_\K$, let translations,
rotations and boosts act on the first factor as in the
Schr\"odinger representation of mass $m_0$, and let the central group
act as $\rme^{\rmi\theta m_0/\hbar}$. These unitary groups satisfy
(B1)--(B5). On the bosonic Fock space over $\mathfrak{h}$ take the
second-quantised representation, $\hH:=\rmd\Gamma(h)$, $\hM:=
m_0\rmd\Gamma(\one)$, the Fock vacuum $\Omega$, the finite-particle
vectors as $\D$, and $\psio(g):=\sqrt\hbar\,a(g\otimes\xi)$. Then
(S1), (S3) and (H1)--(H3) hold. Since $h\geq E_*:=\inf\supp\nu$,
$\rmd\Gamma(h)\geq E_*\,\rmd\Gamma(\one)$, and (H4) holds with
$\lambda:=\max(0,-E_*/m_0)$. By construction $\nu_\xi=\nu$. Moreover
$\evol{\tau}a(k)\evolm{\tau}=a(\rme^{\rmi\tau h/\hbar}k)$, so
$\F(\Op_2,\tau)$ is generated by the closures of $a(k)$, $a^*(k)$ with
$k\in\rme^{\rmi\tau h/\hbar}(C_c^\infty(\Op_2)\otimes\xi)$. If
$\widehat\nu(\tau/\hbar)=0$, then $\braket{f\otimes\xi,
\rme^{\rmi\tau h/\hbar}g\otimes\xi}=(f,\rme^{\rmi\tau h_0/\hbar}g)
\widehat\nu(\tau/\hbar)=0$. The one-particle subspaces carrying the two
algebras are then orthogonal, and the algebras commute by
Lemma~\ref{lem:fock_factorization}. Otherwise they do not commute, by
part~(i).
\end{proof}

The bounds~\eqref{eq:QSL} involve only the spread of the internal
energy, which is smaller than that of $\hH$ in the created state: in
$\psiod(g)\Omega$ the variance of $\hH$ is that of $h_0$ in $g$ plus
$(\Delta_\xi U)^2$. Sets of positive measure are excluded:

\begin{corollary}[Null commutation sets]\label{thm:sr_selection_analytic}
Under (H1)--(H4) no set of commuting pairs of non-empty slice regions at
unequal times---in particular no Galilean unequal-time locality
relation---realises a set of time separations of positive Lebesgue
measure.
\end{corollary}

\begin{proof}
By Theorem~\ref{thm:commutation_set}(i) the realised separations lie in
$Z_\xi$, which is null by part~(ii).
\end{proof}

\begin{remark}[Two species; saturation]\label{rem:two_species}
Let $\hat a,\hat b$ be free Bose fields of the same mass $m_0$,
normalised as $\psio$, with one-particle Hamiltonians $h_0$ and $h_0+E$,
$E>0$, and let
$\hat\psi:=(\hat a+\hat b)/\sqrt2$. This is the net of
Theorem~\ref{thm:commutation_set}(iv) with $\nu=\frac12(\delta_0+
\delta_E)$, so $\widehat{\nu}(s)=\frac12(1+\rme^{\rmi Es})$ and
$Z_\xi=\{(2n+1)\pi\hbar/E:\ n\geq0\}$. All slice algebras at these time
separations commute, and those at any other unequal times do not. Here
$\Delta_\xi U=\braket{U}_\xi-E_0=E/2$, so both bounds
in~\eqref{eq:QSL} equal $\pi\hbar/E=\inf Z_\xi$. The commutation is a
property of sharp times: the algebras of open spacetime regions of this
net never commute (Corollary~\ref{cor:4D}). Figure~\ref{fig:commutation}
shows $|\widehat{\nu_\xi}|$ for this and two other distributions.
\end{remark}

\begin{figure}[htb]
\centering
\begin{tikzpicture}[font=\small, >={Stealth[length=2mm]}, xscale=0.82, yscale=3.6]
\def\figX{13.2}
\draw[->, gray] (0,0) -- (\figX+0.3,0) node[right, black] {$E\tau/\hbar$};
\draw[->, gray] (0,0) -- (0,1.13) node[left, black] {$|\widehat{\nu_\xi}(\tau/\hbar)|$};
\foreach \k/\lab in {1/\pi,2/2\pi,3/3\pi,4/4\pi}
  \draw[gray] ({\k*pi},0) -- ({\k*pi},-0.03) node[below, black, font=\footnotesize] {$\lab$};
\draw[gray] (0,1) -- (-0.1,1) node[left, black, font=\footnotesize] {$1$};
\draw[gray] (0,0) -- (-0.1,0) node[left, black, font=\footnotesize] {$0$};
\draw[gray, line width=0.4pt] ({pi},0) -- ({pi},1.05)
  node[above, font=\scriptsize, black] {$\pi$};
\draw[gray, line width=0.4pt] ({sqrt(3)*pi},0) -- ({sqrt(3)*pi},1.05)
  node[above, font=\scriptsize, black] {$\sqrt3\,\pi$};
\draw[green!45!black, line width=0.9pt, dash pattern=on 4pt off 1.5pt on 1pt off 1.5pt,
  domain=0:\figX, samples=200, smooth]
  plot (\x, {sqrt(0.58+0.42*cos(\x r))});
\draw[blue!70!black, line width=0.9pt, dash pattern=on 3.5pt off 2pt]
  (0,1) -- plot[domain=0.02:\figX, samples=260] (\x, {abs(sin(\x/2 r)/(\x/2))});
\draw[red!70!black, line width=0.9pt, domain=0:\figX, samples=260]
  plot (\x, {abs(cos(\x/2 r))});
\foreach \z in {1,3} \node[circle, fill=red!70!black, inner sep=1.5pt] at ({\z*pi},0) {};
\foreach \z in {2,4} \node[circle, fill=blue!70!black, inner sep=1.5pt] at ({\z*pi},0) {};
\begin{scope}[shift={(0.4,1.3)}]
  \draw[red!70!black, line width=0.9pt] (0,0) -- (0.6,0)
    node[right, black, font=\footnotesize] {levels $0,E$, weights $\tfrac12,\tfrac12$};
  \draw[blue!70!black, line width=0.9pt, dash pattern=on 3.5pt off 2pt] (4.6,0) -- (5.2,0)
    node[right, black, font=\footnotesize] {uniform on $[0,E]$};
  \draw[green!45!black, line width=0.9pt, dash pattern=on 4pt off 1.5pt on 1pt off 1.5pt] (8.4,0) -- (9.0,0)
    node[right, black, font=\footnotesize] {levels $0,E$, weights $0.7,0.3$};
\end{scope}
\end{tikzpicture}
\caption{Commutation times as zeros. The modulus of the characteristic
function of the internal-energy distribution,
$|\widehat{\nu_\xi}(\tau/\hbar)|$, as a function of $E\tau/\hbar$, for
two levels $0$ and $E$ with equal weights (solid; the two species of
Remark~\ref{rem:two_species}), for the uniform distribution on $[0,E]$
(dashed), and for two levels with weights $0.7$ and $0.3$ (dash-dotted).
By Theorem~\ref{thm:commutation_set}, slice algebras at time separation
$\tau$ can commute only at the zeros (dots): odd multiples of
$\pi\hbar/E$ in the first case, even ones in the second, none in the
third, which is the generic case (Proposition~\ref{prop:attainable}).
The grey lines mark the Mandelstam--Tamm bound
$\pi\hbar/(2\Delta_\xi U)$ in the first two cases, at $E\tau/\hbar=\pi$
and $\sqrt3\,\pi$; the first case attains it. In both cases the
Margolus--Levitin term of~\eqref{eq:QSL} equals $\pi\hbar/E$, so the grey
lines are the full bound~\eqref{eq:QSL}.}
\label{fig:commutation}
\end{figure}

For a probability measure $\nu$ on $\R$ put
\[
Z_\nu:=\{s>0:\ \widehat\nu(s)=0\},
\]
so that $Z_\xi=\hbar Z_{\nu_\xi}$. By Theorem~\ref{thm:commutation_set}(i)
and~(iv), slice algebras of a net satisfying (H1)--(H4) can commute at
unequal times only at separations in $\hbar Z_{\nu_\xi}$, and for every
$\nu$ with $\inf\supp\nu>-\infty$ there is such a net in which they
commute at exactly the separations in $\hbar Z_\nu$. The set $Z_\nu$ is
closed and bounded away from zero for every $\nu$; by part~(ii) it is
null when $\inf\supp\nu>-\infty$, and by Remark~\ref{rem:central_shift}
we may then translate $\nu$ so that $\inf\supp\nu\geq0$. The following
proposition shows that exact commutation is non-generic, that it can
nevertheless occur at any prescribed finite set of time separations,
that infinite sets with $\sum_j1/z_j<\infty$ are allowed, and, in
part~(e), that without the lower bound on $\nu$ it can occur on a
half-line.

\begin{proposition}[Attainable commutation sets]\label{prop:attainable}
\begin{enumerate}[label=(\alph*)]
\item (Genericity.) The probability measures $\nu$ on $[0,\infty)$ with
$Z_\nu=\emptyset$ form a dense $G_\delta$ subset of the space of
probability measures on $[0,\infty)$ with the topology of weak
convergence. For fixed $E_0<E_1<\dots<E_N$, $N\geq1$, the measure
$\sum_kw_k\delta_{E_k}$ has $Z_\nu=\emptyset$ for Lebesgue-almost every
$(w_0,\dots,w_N)$ in the open simplex $\{w_k>0,\ \sum_kw_k=1\}$.
\item (Realisability.) Let $(z_j)$ be a finite or infinite sequence in
$(0,\infty)$ with $\sum_j1/z_j<\infty$, and let $E_1',\dots,E_M'>0$ and
$b_1,\dots,b_K>0$, $M,K\geq0$. There is a probability measure $\nu$ with
compact support in $[0,\infty)$ such that
\[
Z_\nu=\{z_j\}\cup\bigcup_{i=1}^M\Bigl\{\frac{(2n+1)\pi}{E_i'}:\
n\geq0\Bigr\}\cup\bigcup_{k=1}^K\Bigl\{\frac{2\pi n}{b_k}:\ n\geq1\Bigr\}.
\]
In particular every finite subset of $(0,\infty)$ is of the form $Z_\nu$.
\item (Compact support.) If $\supp\nu\subset[a,a+2L]$, then $Z_\nu$ is
discrete and $\#\bigl(Z_\nu\cap(0,R]\bigr)\leq\rme LR/\pi$ for every
$R>0$.
\item (Accumulation.) There is a probability measure $\nu$ on
$[0,\infty)$ such that $Z_\nu$ has a finite accumulation point. This is
impossible if $\int\rme^{\varepsilon E}\rmd\nu(E)<\infty$ for some
$\varepsilon>0$.
\item (Without positivity.) The construction in the proof of
Theorem~\ref{thm:commutation_set}(iv) applies to every probability
measure $\nu$ on $\R$. It gives a net satisfying (S1)--(S3) and
(H1)--(H3), which satisfies (H4) if and only if $\inf\supp\nu>-\infty$,
and in which, for $\tau\neq0$, the slice algebras of all non-empty
bounded open $\Op_1,\Op_2$ at times $0$ and $\tau$ commute if
$\widehat\nu(\tau/\hbar)=0$, and those of no such pair commute
otherwise. For $\tau_0>0$ the measure
\begin{equation}\label{eq:polya}
\rmd\nu_{\tau_0}(E):=\frac{\tau_0}{2\pi\hbar}\,
\Bigl(\frac{\sin(\tau_0E/2\hbar)}{\tau_0E/2\hbar}\Bigr)^2\,\rmd E
\end{equation}
has $\widehat{\nu_{\tau_0}}(s)=(1-\hbar|s|/\tau_0)_+$, so
$Z_{\nu_{\tau_0}}=[\tau_0/\hbar,\infty)$: the slice algebras at times $0$
and $\tau$ commute exactly when $|\tau|\geq\tau_0$.
\end{enumerate}
\end{proposition}

\begin{proof}
(a) If $\nu_n\to\nu$ weakly, $\widehat{\nu_n}\to\widehat\nu$ uniformly on
compact sets, so for each $m\in\N$ the set $\{\nu:\ \min_{[1/m,m]}
|\widehat\nu|>0\}$ is open, and their intersection over $m$ is the set
of $\nu$ with $Z_\nu=\emptyset$. For the second statement let
$F(s,w):=\sum_kw_k\rme^{\rmi E_ks}\in\C\cong\R^2$ on $(0,\infty)\times
\Sigma_N$, with $\Sigma_N$ the open simplex, of dimension $N$. The
derivative of $F$ in the directions $\delta w$ with $\sum_k\delta w_k=0$
has range spanned by the vectors $\rme^{\rmi E_ks}-\rme^{\rmi E_0s}$. At
a zero of $F$ these span $\R^2$ unless all the points $\rme^{\rmi E_ks}$
lie on one real line; since $0$ is a convex combination of them, that
line passes through $0$, so $\rme^{\rmi(E_k-E_0)s}=\pm1$ for all $k$.
This happens only for $s$ in the discrete set $S_0:=\{s>0:\
(E_k-E_0)s\in\pi\Z\ \text{for all }k\}$. Off $S_0$ the zero set of $F$
is therefore a submanifold of dimension $N-1$, and its image under the
projection $(s,w)\mapsto w$ is a null subset of $\Sigma_N$, as the image of
an $(N-1)$-dimensional manifold under a smooth map into an
$N$-dimensional one. For $s\in S_0$, $F(s,w)=\rme^{\rmi E_0s}\sum_k
\epsilon_kw_k$ with $\epsilon_k=\pm1$ and $\epsilon_0=1$; this is
non-zero if all $\epsilon_k=1$, and otherwise vanishes on a hyperplane
that meets $\Sigma_N$ in a null set. The union of these countably
many null sets contains every $w$ with $Z_\nu\neq\emptyset$. Finally,
finitely supported measures are weakly dense, and by the second
statement their weights can be perturbed so that $Z_\nu=\emptyset$; so
the $G_\delta$ set is dense.

(b) Let $s_0\in(\pi,\frac32\pi)$ be the smallest positive solution of
$\tan s=s$ and $m_*:=\sin s_0/s_0=\cos s_0<0$. The function $\sin s/s$
attains its minimum over $[0,\infty)$ only at $s_0$: on $[0,\pi]$ it is
non-negative, on $[\pi,2\pi]$ it is non-positive with $s_0$ as its only
critical point there, and on $[2\pi,\infty)$ it is at least
$-1/(2\pi)>m_*$.
Put $p:=-m_*/(1-m_*)\in(0,1)$ and let $\nu_0:=p\,\delta_1+(1-p)\,
\tfrac12\one_{[0,2]}(E)\,\rmd E$. Then
\[
\widehat{\nu_0}(s)=\rme^{\rmi s}\Bigl(p+(1-p)\frac{\sin s}{s}\Bigr),
\]
whose modulus vanishes for $s>0$ exactly at $s=s_0$. For $c>0$ let
$\nu_0^{(c)}$ be the image of $\nu_0$ under $E\mapsto cE$, so that
$\widehat{\nu_0^{(c)}}(s)=\widehat{\nu_0}(cs)$ vanishes exactly at
$s_0/c$. The two-point measure $\frac12(\delta_0+\delta_{E'})$ has
$\frac12(1+\rme^{\rmi E's})$, which vanishes exactly at $(2n+1)\pi/E'$,
and the uniform distribution on $[0,b]$ has $\rme^{\rmi bs/2}\sin(bs/2)/
(bs/2)$, which vanishes exactly at $2\pi n/b$, $n\geq1$. With
$c_j:=s_0/z_j$, $\sum_jc_j<\infty$, and the convolution $\nu$ of the
measures $\nu_0^{(c_j)}$ with these finitely many factors exists: it is
the law of a sum of independent random variables, those of the first
kind with values in $[0,2c_j]$, which converges because $\sum_jc_j<
\infty$. Its support lies in $[0,2\sum_jc_j+\sum_iE_i'+\sum_kb_k]$.
Moreover $|1-\widehat{\nu_0}(c_js)|\leq c_js$, because $\nu_0$ has mean
$1$, so the product $\widehat\nu(s)$ of the characteristic functions
converges absolutely and vanishes exactly where one of its factors
does.

(c) The function $F(z):=\int\rme^{\rmi(E-a-L)z}\rmd\nu(E)$ is entire,
$F(0)=1$, $|F(z)|\leq\rme^{L|\mathrm{Im}\,z|}$, and $|F|=|\widehat\nu|$
on $\R$. Its zeros are isolated, and its real zeros come in pairs
$\pm s$ because $F(-s)=\overline{F(s)}$ for real $s$. If $n(r)$ counts
the zeros in the disc of radius $r$, Jensen's formula gives $n(R)\leq
\int_R^{\rme R}n(r)\,\rmd r/r\leq(2\pi)^{-1}\int_0^{2\pi}\ln|F(\rme R
\rme^{\rmi\theta})|\,\rmd\theta\leq2\rme LR/\pi$, and
$\#(Z_\nu\cap(0,R])\leq n(R)/2$.

(d) If $\int\rme^{\varepsilon E}\rmd\nu<\infty$, $\widehat\nu$ extends
holomorphically to $\{|\mathrm{Im}\,z|<\varepsilon\}$, so its real zeros
are isolated. For the example, let $0<\gamma<\frac12$, $\omega>0$, and
let $g$ be real and continuous with support $[-1,1]$, linear on
$[-\frac12,\frac12]$, and
\[
g(1-x)=x^\gamma\cos(\omega\ln x),\qquad g(-1+x)=x^\gamma\sin(\omega\ln x)
\qquad(0<x\leq\tfrac12).
\]
Let $P:=|\hat g|^2$, which is even, and $Q(t):=\int g(x)g(x+t)\,\rmd x$,
so that $P(E)=\int Q(t)\rme^{-\rmi Et}\rmd t$ and $\supp Q\subset[-2,2]$.
The edges of $g$ give $\hat g(E)=O(|E|^{-1-\gamma})$, so $EP(E)$ is
integrable. Let $g_a\geq0$ be the triangle function of support
$[-\frac12,\frac12]$, $P_a:=|\hat g_a|^2$ and $Q_a:=g_a*g_a\geq0$, with
support $[-1,1]$. For $c_0,\lambda\geq0$ let $\nu$ be the normalised
measure $c_0\delta_0+E(P+\lambda P_a)(E)\one_{E>0}\,\rmd E$. For $s>2$,
inserting $P(E)=\int Q(t)\cos(Et)\,\rmd t$ and using the Abel limits
$\int_0^\infty E\rme^{-\eta E}\cos(Eu)\,\rmd E\to-u^{-2}$ and $\int_0^\infty
E\rme^{-\eta E}\sin(Eu)\,\rmd E\to0$ as $\eta\to0$, for $u=s\pm t\neq0$,
we get, up to the normalisation,
\[
\widehat\nu(s)=c_0-G(s)-\lambda G_a(s),\qquad G(s):=\int Q(t)(s-t)^{-2}
\rmd t,
\]
and likewise $G_a$: beyond the support of $Q$, $\widehat\nu$ is real. For
$0<\sigma\leq\frac12$ the substitution $x=-1+\sigma u$ and the
symmetry $u\leftrightarrow1-u$ give exactly $Q(2-\sigma)=\mathrm{Im}
(\beta_0\sigma^\zeta)$, with $\zeta:=1+2\gamma+2\rmi\omega$ and $\beta_0:=\frac12
\mathrm{B}(1+\gamma+\rmi\omega,1+\gamma+\rmi\omega)\neq0$. Since
$\int_0^\infty x^\zeta[(1+x)^{-2}-x^{-2}]\,\rmd x=\pi\zeta/\sin(\pi\zeta)$ for
$1<\mathrm{Re}\,\zeta<2$, it follows that
\[
G(2+\delta)-G(2)=\mathrm{Im}\bigl(A\,\delta^{2\gamma}\rme^{2\rmi\omega
\ln\delta}\bigr)+O(\delta),\qquad A:=\frac{\pi\zeta\beta_0}{\sin(\pi\zeta)}\neq0,
\]
as $\delta\downarrow0$, while $G_a$ is smooth near $2$. Choose $\lambda
\geq0$ and $c_0\geq0$ with $c_0=G(2)+\lambda G_a(2)$, which is possible
because $G_a(2)>0$. Then $\widehat\nu(2+\delta)=-\mathrm{Im}(A\,
\delta^{2\gamma}\rme^{2\rmi\omega\ln\delta})+O(\delta)$ up to the
normalisation, which changes sign infinitely often as $\delta\downarrow0$
because $2\gamma<1$. So $Z_\nu$ contains points $2+\delta_n$ with
$\delta_n\downarrow0$.

(e) In the proof of Theorem~\ref{thm:commutation_set}(iv) the lower
bound $E_*>-\infty$ enters only through (H4), and (S2) follows from
Lemma~\ref{lem:covariance}, which does not use it. Conversely, (H4)
gives $h+\lambda m_0\geq0$ on the one-particle space, and since
$\inf\sigma(h_0)=0$ this means $U_\K\geq-\lambda m_0$, that is,
$\inf\supp\nu\geq-\lambda m_0$. The commutation criterion does not use
(H4) either. If $\widehat\nu(\tau/\hbar)=0$ the algebras commute by
Lemma~\ref{lem:fock_factorization}, as there. If not,
Lemma~\ref{lem:free_nonvanishing} gives $f\in C_c^\infty(\Op_1)$ and
$g\in C_c^\infty(\Op_2)$ with $\braket{a^*(f\otimes\xi)\Omega,\evol{\tau}
a^*(g\otimes\xi)\Omega}=(f,\rme^{\rmi\tau h_0/\hbar}g)\,
\widehat\nu(\tau/\hbar)\neq0$. Let $A$ be the closure of
$a(f\otimes\xi)$ and $B$ that of $\evol{\tau}a^*(g\otimes\xi)
\evolm{\tau}$. Then $A\Omega=0$, $A^*\Omega=a^*(f\otimes\xi)\Omega$,
$B\Omega=\evol{\tau}a^*(g\otimes\xi)\Omega$ and $B^*\Omega=0$, and
Lemma~\ref{lem:truncation} shows that the two algebras do not commute.
Finally, a direct computation gives
\[
\frac1{2\pi}\int\Bigl(1-\frac{\hbar|s|}{\tau_0}\Bigr)_+\rme^{-\rmi sE}
\,\rmd s=\frac{\tau_0}{2\pi\hbar}\Bigl(\frac{\sin(\tau_0E/2\hbar)}
{\tau_0E/2\hbar}\Bigr)^2;
\]
both functions are integrable, so Fourier inversion gives
$\widehat{\nu_{\tau_0}}(s)=(1-\hbar|s|/\tau_0)_+$, and $s=0$ gives
$\nu_{\tau_0}(\R)=1$.
\end{proof}

The order of the bound in~(c) is attained: the uniform distribution on
$[a,a+2L]$ has the $LR/\pi$ zeros $n\pi/L$ in $(0,R]$, up to rounding.
The condition $\sum_j1/z_j<\infty$ in~(b) is not necessary, as this
example shows, but it cannot be avoided in constructions by convolution
in which every factor carries one of the zeros. If a random variable
$X_j\geq0$ has a characteristic function that vanishes at $z_j\geq1$,
then $1=\mathbb{E}[1-\cos(z_jX_j)]\leq2z_j\,\mathbb{E}[\min(X_j,1)]$,
because $1-\cos y\leq2\min(y,1)$ for $y\geq0$; and a sum of independent
non-negative random variables converges only if $\sum_j\mathbb{E}
[\min(X_j,1)]<\infty$.

Exact commutation at unequal times is therefore non-generic. In the
two-point and uniform examples it comes from a reflection symmetry of
the internal-energy distribution about a point $c$, which makes
$\rme^{-\rmi cs}\widehat\nu(s)$ real. Zeros at which this function
changes sign survive perturbations that are small in total variation
and preserve the symmetry, and every uniform distribution on an interval
$[a,b]$ has $Z_\nu=\{2\pi n/(b-a):\ n\geq1\}$. The zero of the
symmetric measure $\nu_0$ in~(b) is a tangency, which a symmetric
perturbation can remove, and the measure in~(d) has no reflection
symmetry. Exact commutation can be arranged at any finite set of time
separations, for instance at a single one. For probability
measures on the whole line, whose characteristic functions need not
extend to a half-plane, every closed set that is symmetric under
$s\mapsto-s$ and does not contain $0$ is the zero set of some
characteristic function~\cite{Sasvari1985,RubinSellke2004}; positivity of
the energy restricts the half-line case to null sets. Which closed null
sets arise on a half-line is only partly known: all finite sets and all
discrete sets $\{z_j\}$ with $\sum_j1/z_j<\infty$ do, and so do some
sets with a finite accumulation point, but only for internal energies
without exponential moments.
Without the lower bound on the internal energy the nullity of $Z_\nu$
fails, as the measure~\eqref{eq:polya} shows, and by part~(e) so does
that of the commutation sets: with the result just quoted, the
commutation sets of the nets of part~(e) are exactly the sets $\hbar F$
with $F\subset(0,\infty)$ closed and bounded away from $0$.

Part~(e) shows that (H4) cannot be dropped from
Corollaries~\ref{thm:sr_selection_analytic} and~\ref{cor:4D}, nor from
the nullity of the commutation set in
Theorem~\ref{thm:commutation_set}(ii). For $\nu=\nu_{\tau_0}$ the net
of~(e) satisfies (S1)--(S3) and (H1)--(H3), its fields annihilate the
vacuum, and its vacuum two-point function has the form~\eqref{eq:GKL}
with $\nu$ in place of $\nu_\xi$. Only the lower bound on $\nu$ is
missing, and the energy is not bounded below even in the one-particle
sector, so the weaker hypothesis of Remark~\ref{rem:central_shift} fails
as well. The set of all pairs of slice regions at time separations of
at least $\tau_0$ is a Galilean unequal-time locality relation for this
net, with $T(\mathcal{L})=[\tau_0,\infty)$, and the algebras of open
spacetime regions whose time projections are at distance at least
$\tau_0$ commute as well, since the one-particle subspaces spanned by
their time-smeared fields are orthogonal; (T) holds in this model. The
hypothesis of Theorem~\ref{thm:sr_selection} that the relation is
realised at short times therefore cannot be weakened to a set of
separations of positive measure. The Mandelstam--Tamm bound
in~\eqref{eq:QSL} does not use positivity and survives when the variance
is finite: the measure with density proportional to
$\bigl(\sin(\tau_0E/2\hbar)/E\bigr)^4$ has variance $(\Delta E)^2=
3\hbar^2/\tau_0^2$,
and its characteristic function vanishes exactly for
$|s|\geq2\tau_0/\hbar$, so commutation holds exactly at
$|\tau|\geq2\tau_0$, while $\pi\hbar/(2\Delta E)=\pi\tau_0/(2\sqrt3)$; the
Margolus--Levitin bound needs a lower bound on the energy.

The commutation in~(e) is an identity between Weyl operators,
$W(k)W(k')=W(k')W(k)$ when $\mathrm{Im}\braket{k,k'}=0$, so it holds in
every representation of the Weyl algebra over the one-particle space,
with $\mathfrak{A}(\Op)$ generated by the $W(g\otimes\xi)$,
$g\in C_c^\infty(\Op)$, and the dynamics
$W(k)\mapsto W(\rme^{\rmi th/\hbar}k)$. This algebra is simple, so by
Proposition~\ref{prop:transfer} the non-commutation for
$\widehat\nu(\tau/\hbar)\neq0$ holds in every non-zero representation as
well. No representation of it satisfies (H), in the sense defined before
Proposition~\ref{prop:transfer}, as soon as $\inf\supp\nu=-\infty$. In
such a representation the fields annihilate the vacuum
(Lemma~\ref{lem:vacuum_annihilation}), so the vacuum expectation of
$[\psio(f),\evol{\tau}\psiod(g)\evolm{\tau}]$, which the Weyl relations
fix as $\hbar(f,\rme^{\rmi\tau h_0/\hbar}g)\widehat\nu(\tau/\hbar)$, equals
the left-hand side of~\eqref{eq:GKL}. By
Lemma~\ref{lem:free_nonvanishing} this gives $\nu_\xi=\nu$, which
contradicts $\inf\supp\nu_\xi\geq-\lambda m_0$.

\section{Consequences}\label{sec:consequences}

\subsection{Open spacetime regions}

For a non-empty bounded open $\Op\subset\R^3\times\R$ we use the
following link between four-dimensional algebras and time-zero fields.

\begin{enumerate}[label=(T)]
\item There are von Neumann algebras $\F(\Op)$ with the following
property. Whenever $B\times I\subset\Op$ with $B\subset\R^3$ and
$I\subset\R$ open, $g\in C_c^\infty(B)$ and $\chi\in C_c^\infty(I)$,
there is a closed operator $\hat\psi(\chi\otimes g)$ affiliated with
$\F(\Op)$ such that $\Omega\in D(\hat\psi(\chi\otimes g))\cap
D(\hat\psi(\chi\otimes g)^*)$, $\hat\psi(\chi\otimes g)\Omega=0$ and
$\hat\psi(\chi\otimes g)^*\Omega=\int\chi(t)\,\evol{t}\psiod(g)\Omega
\,\rmd t$.
\end{enumerate}

For the Fock models of Theorem~\ref{thm:commutation_set}(iv), which
include the free field, and for the models of Corollary~\ref{cor:pair},
(T) holds with $\hat\psi(\chi\otimes g)$ the closure of
$\int\overline{\chi(t)}\,\evol{t}\psio(g)\evolm{t}\,\rmd t$ on
finite-particle vectors. There $\F(\Op)$ is the von Neumann algebra
generated by these operators for $B\times I\subset\Op$, since the
time-translated annihilation fields annihilate the vacuum and the adjoint
contains the corresponding integral of creation fields.
Hypothesis~(T) does not require slice algebras to be contained in the
algebras of open regions.

\begin{corollary}[No commuting open regions]\label{cor:4D}
Assume (H1)--(H4) and (T). Then $[\F(\Op_1),\F(\Op_2)]\neq\{0\}$ for all
non-empty bounded open $\Op_1,\Op_2\subset\R^3\times\R$.
\end{corollary}

\begin{proof}
Choose $B_i\times I_i\subset\Op_i$, with $B_i$ balls and $I_i$ open
intervals, and suppose that the algebras commute. By
Lemma~\ref{lem:free_nonvanishing} there are $f\in C_c^\infty(B_1)$,
$g\in C_c^\infty(B_2)$ and $\tau'\neq0$ with $(f,\rme^{\rmi\tau'h_0/
\hbar}g)\neq0$. For $\chi_i\in C_c^\infty(I_i)$ apply
Lemma~\ref{lem:truncation} to $A:=\hat\psi(\chi_1\otimes f)$ and
$B:=\hat\psi(\chi_2\otimes g)^*$, which are affiliated with $\F(\Op_1)$
and $\F(\Op_2)$. Since $B^*\Omega=0$,
\[
0=\braket{A^*\Omega,B\Omega}=\iint\overline{\chi_1(t)}\chi_2(s)\,
k(s-t)\,\rmd t\,\rmd s,\qquad
k(\tau):=\braket{\psiod(f)\Omega,\evol{\tau}\psiod(g)\Omega}.
\]
The function $k$ is continuous. Letting $\chi_1$ run through an
approximate identity at $t_0\in I_1$ gives $k=0$ on the open interval
$J:=I_2-t_0$. By Lemma~\ref{lem:GKL}, $\tilde k(z):=\hbar(f,
\rme^{\rmi zh_0/\hbar}g)\int\rme^{\rmi(E-E_0)z/\hbar}\rmd\nu_\xi(E)$ is
holomorphic in $\{\mathrm{Im}\,z>0\}$, bounded and continuous on the
closed half-plane. It equals $\rme^{-\rmi E_0\tau/\hbar}k(\tau)$ on
$\R$, and so vanishes on $J$. By the Schwarz reflection principle and
the identity theorem, $\tilde k\equiv0$. The second factor is not
identically zero, since it tends to $1$ as $z\to0$. Zeros of a non-zero
holomorphic function are isolated, so the first factor vanishes
identically, and by continuity $(f,\rme^{\rmi\tau h_0/\hbar}g)=0$ for all
real $\tau$, a contradiction at $\tau=\tau'$.
\end{proof}

A four-dimensional region contains slices at a continuum of times, and
under (H) the commutation set of slice algebras is null. For this reason
the sharp-time commutation of Remark~\ref{rem:two_species} does not
survive for open regions: under (H) and (T), no Galilean unequal-time
locality between open spacetime regions is possible. Without positivity
it is, as shown after Proposition~\ref{prop:attainable}.

\subsection{Many-body systems with pair interactions}

\begin{corollary}[Pair interactions]\label{cor:pair}
Let $V\in L^\infty(\R^3)$ be real with $V(-\bx)=V(\bx)$. On the bosonic
or fermionic Fock space over $L^2(\R^3)$, let $\hH$ act on the
$n$-particle subspace as the self-adjoint operator
$\sum_{i=1}^nh_0^{(i)}+\sum_{i<j}V(\bx_i-\bx_j)$ on the domain of
$\sum_ih_0^{(i)}$. Take $\D$ to be the finite-particle vectors with
Schwartz components, $\psio(g):=\sqrt\hbar\,a(g)$ and $\psiod(g):=
\sqrt\hbar\,a^*(g)$. Then, for all non-empty bounded open
$\Op_1,\Op_2\subset\R^3$ and $t_1\neq t_2$, the slice algebras
$\F(\Op_1,t_1)$ and $\F(\Op_2,t_2)$ do not commute (bosons), respectively
their odd generators do not all anticommute (fermions). If $V$ is also
rotation invariant, (H1)--(H3) hold, (H4) holds for $V\geq0$, and the
weaker hypothesis of Remark~\ref{rem:central_shift} holds for every
bounded $V$, with $\ell(nm_0)=n(n-1)\|V\|_\infty/2$.
\end{corollary}

\begin{proof}
On the $n$-particle subspace the pair potential is a bounded symmetric
operator, and a bounded symmetric perturbation of a self-adjoint operator
is self-adjoint on the same domain. The vacuum is annihilated by $\hH$,
and $\hH$ acts on the one-particle space as $h_0$, since $V$ acts only
on pairs. Take $t_1=0$, $\tau:=t_2$, $f\in C_c^\infty(\Op_1)$ and
$g\in C_c^\infty(\Op_2)$. If the algebras commute,
Lemma~\ref{lem:truncation} applied as in
Theorem~\ref{thm:commutation_set}(i) gives $\braket{a^*(f)\Omega,
\evol{\tau}a^*(g)\Omega}=0$, that is, $(f,\rme^{\rmi\tau h_0/\hbar}g)=0$
for all $f,g$. This contradicts Lemma~\ref{lem:free_nonvanishing}. For
fermions the fields are bounded, and anticommutation of
$a(f)$ with $\evol{\tau}a^*(g)\evolm{\tau}$ gives the same identity. For
radial $V$ the second-quantised Bargmann representation satisfies
(B1)--(B5) with this $\hH$. The potential is a multiplication operator
depending only on relative positions, so it commutes with the boosts
and translations, and the kinetic energy transforms as for free
particles. The set $\D$ is invariant under all the unitaries and fields
in (H1)--(H2), and (H3) holds for the Fock vacuum. If $V\geq0$, then
$\hH\geq0$; in general $\hH\geq-n(n-1)\|V\|_\infty/2$ on the
$n$-particle subspace.
\end{proof}

No stability condition on $V$ is needed: the vacuum witnesses only the
one-particle propagation, which Galilean invariance forces to be free.
In the language of Theorem~\ref{thm:commutation_set}, $\K=\C$ and
$Z_\xi=\emptyset$. The commutators themselves are not c-numbers when
$V\neq0$; in a one-particle state they differ from the free ones.

\subsection{Fermi fields and observable algebras}

\begin{proposition}[Fermi fields and gauge-invariant
observables]\label{prop:fermi_obs}
\begin{enumerate}[label=(\alph*)]
\item Replace~(S3) by the canonical anticommutation relations
$\{\psio(f),\psiod(g)\}=\hbar(f,g)$, $\{\psio(f),\psio(g)\}=0$ on $\D$,
and read the commutation of two slice algebras $\F(\Op_1,t_1)$,
$\F(\Op_2,t_2)$ in~(L1) and in Theorem~\ref{thm:commutation_set} as the
anticommutation of the fields $\evol{t_1}\psio^\sharp(f)\evolm{t_1}$
and $\evol{t_2}\psio^\sharp(g)\evolm{t_2}$ for $f\in
C_c^\infty(\Op_1)$, $g\in C_c^\infty(\Op_2)$, where $\psio^\sharp$
stands for $\psio$ or $\psiod$. Then Theorem~\ref{thm:sr_selection}(i),
Lemmas~\ref{lem:vacuum_annihilation} and~\ref{lem:GKL},
Theorem~\ref{thm:commutation_set} and
Corollaries~\ref{thm:sr_selection_analytic} and~\ref{cor:4D} hold, with
anticommutators in the proofs and bounded operators
$\hat\psi(\chi\otimes g)$ in~(T).
\item Assume (H1)--(H4), with the commutation relations~(S3) or the
anticommutation relations of~(a), and assume that $\Omega$ is, up to
scalars, the only vector invariant under the spatial translations. Let $\A(\Op,t)$ be the fixed-point algebra of
$\F(\Op,t)$ under $\mathrm{Ad}\,\rme^{\rmi\theta\hM/\hbar}$. If
$[\A(\Op_1,0),\A(\Op_2,\tau)]=\{0\}$ for non-empty bounded open
$\Op_1,\Op_2$ and $\tau\neq0$, then $|\tau|\in Z_\xi$.
\end{enumerate}
\end{proposition}

\begin{proof}
(a) The anticommutation relations imply
$\|\psio(f)u\|^2+\|\psiod(f)u\|^2=\hbar\|f\|^2\|u\|^2$, so the fields are
bounded. Replace $F_{f,h}$ in~\eqref{eq:F_witness} by
$\braket{\psiod(f)\Omega,\evol{\tau}\psiod(h)\Omega}+
\braket{\evol{\tau}\psio(h)\Omega,\psio(f)\Omega}$, which equals
$\hbar\|g\|^2$ at $\tau=0$ for $f=h=g$. For bounded operators the
analogue of Lemma~\ref{lem:truncation} for anticommutators is immediate.
In Lemma~\ref{lem:vacuum_annihilation} and Lemma~\ref{lem:GKL} only
$\|\psiod(g)\Omega\|^2=\hbar\|g\|^2-\|\psio(g)\Omega\|^2=\hbar\|g\|^2$
changes. In Theorem~\ref{thm:commutation_set}(iv) the anticommutation
relations give the anticommutation of fields with orthogonal test
vectors directly.

(b) By Lemma~\ref{lem:joint_spectrum}, $\hP=0$ on $\ker\hM$, so every
vector of $\ker\hM$ is invariant under the spatial translations. By
uniqueness, $\ker\hM=\C\Omega$. Let $A_f$ be the closure of $\psio(f)$.
By~(H2), $A_f$ maps $\Hil_{m_0}\cap D(A_f)$ into $\ker\hM$, so
$A_f\Psi=\braket{\Omega,A_f\Psi}\Omega=\braket{\Phi_f,\Psi}\Omega$ with
$\Phi_f:=\psiod(f)\Omega=A_f^*\Omega$. The operator $A_f^*A_f$ is gauge
invariant, so $\Hil_{m_0}$ reduces it, and there it acts as the bounded
operator $|\Phi_f\rangle\langle\Phi_f|$. Hence the bounded,
gauge-invariant operator $X_f:=A_f^*A_f(\one+A_f^*A_f)^{-1}\in
\A(\Op_1,0)$ acts on $\Hil_{m_0}$ as $c_f|\Phi_f\rangle\langle\Phi_f|$
with $c_f:=(1+\|\Phi_f\|^2)^{-1}$, and $Y_g:=\evol{\tau}X_g\evolm{\tau}
\in\A(\Op_2,\tau)$ acts as $c_g|w_g\rangle\langle w_g|$ with
$w_g:=\evol{\tau}\Phi_g$. From $[X_f,Y_g]=0$ on $\Hil_{m_0}$,
$\braket{\Phi_f,w_g}|\Phi_f\rangle\langle w_g|=\braket{w_g,\Phi_f}
|w_g\rangle\langle\Phi_f|$. Applying both sides to $\Phi_f$ shows that
$w_g$ is proportional to $\Phi_f$ whenever $\braket{\Phi_f,w_g}\neq0$.
Suppose $\ell(g):=\braket{\Phi_f,w_g}\neq0$ for some $g_1\in
C_c^\infty(\Op_2)$. The functional $\ell$ is continuous, so there is
$g_2\in C_c^\infty(\Op_2)$, not proportional to $g_1$, with
$\ell(g_2)\neq0$. Then $w_{g_1}$ and $w_{g_2}$ are both proportional to
$\Phi_f$, hence to each other. This forces $g_1\propto g_2$, since
$g\mapsto w_g$ is injective by Lemma~\ref{lem:vacuum_annihilation}.
Therefore $\ell\equiv0$ for every $f\in C_c^\infty(\Op_1)$, and
Lemma~\ref{lem:GKL} and Lemma~\ref{lem:free_nonvanishing} give
$\widehat{\nu_\xi}(\tau/\hbar)=0$.
\end{proof}

The vacuum itself is not a witness for gauge-invariant observables,
because bilinears $\psiod\psio$ annihilate it. Part~(b) uses one-particle
states instead, and it needs the uniqueness of the vector invariant under
spatial translations. Uniqueness under all spacetime translations would
not suffice for this argument, because Lemma~\ref{lem:joint_spectrum}
gives $\hP=0$ on $\ker\hM$ but not $\hH=0$.

\subsection{Distance-independent commutators}

\begin{proposition}[Distance-independent vacuum commutators]
\label{prop:antiLR}\label{rem:norm_form_LR}
Assume (H1)--(H4). For $\tau\neq0$ and non-empty bounded open
$\Op_1,\Op_2\subset\R^3$ let $c(\Op_1,\Op_2;\tau)$ be the supremum of
$|F_{f,g}(\tau)|$, with $F$ as in~\eqref{eq:F_witness}, over
$f\in C_c^\infty(\Op_1)$ and $g\in C_c^\infty(\Op_2)$ with
$\|f\|=\|g\|=1$. Then
\begin{equation}\label{eq:antiLR}
c(\Op_1,\Op_2;\tau)=\hbar\,\bigl\|\one_{\Op_1}\rme^{\rmi\tau h_0/\hbar}
\one_{\Op_2}\bigr\|\,\bigl|\widehat{\nu_\xi}(\tau/\hbar)\bigr|,
\end{equation}
the right-hand side is invariant under independent translations of
$\Op_1$ and $\Op_2$, and $\|\one_{\Op_1}\rme^{\rmi\tau h_0/\hbar}
\one_{\Op_2}\|\to1$ as $\tau\to0$. In particular, whenever
$\widehat{\nu_\xi}(\tau/\hbar)\neq0$, no bound on
$|F_{f,g}(\tau)|/\|f\|\|g\|$ that tends to zero as the distance between
the supports grows can hold.
\end{proposition}

\begin{proof}
By Lemmas~\ref{lem:vacuum_annihilation} and~\ref{lem:GKL},
$F_{f,g}(\tau)=\hbar(f,\rme^{\rmi\tau h_0/\hbar}g)\widehat{\nu_\xi}
(\tau/\hbar)$. The supremum of $|(f,Xg)|$ over unit vectors of
$L^2(\Op_1)$ and $L^2(\Op_2)$ is $\|\one_{\Op_1}X\one_{\Op_2}\|$, and
$C_c^\infty(\Op)$ is dense in $L^2(\Op)$. In the Schr\"odinger
representation, (B4) reads
\[
e_{-\bv}\rme^{\rmi\tau h_0/\hbar}e_\bv=
\rme^{\rmi m_0|\bv|^2\tau/2\hbar}\tau_{-\bv\tau}\rme^{\rmi\tau h_0/\hbar}.
\]
Multiplication by $e_{\pm\bv}$ is unitary and commutes with
$\one_{\Op_i}$. Hence, taking $\bv=-\mathbf d/\tau$,
\[
\bigl\|\one_{\Op_1}\tau_{\mathbf d}\rme^{\rmi\tau h_0/\hbar}
\one_{\Op_2}\bigr\|=\bigl\|\one_{\Op_1}\rme^{\rmi\tau h_0/\hbar}
\one_{\Op_2}\bigr\|\qquad\text{for every }\mathbf d\in\R^3.
\] Since $\tau_{\mathbf d}$ commutes with
$\rme^{\rmi\tau h_0/\hbar}$ and $\one_{\Op_2+\mathbf d}=\tau_{\mathbf d}
\one_{\Op_2}\tau_{\mathbf d}^{-1}$, the norm is invariant under
translations of $\Op_2$, and by taking adjoints under translations of
$\Op_1$. For the limit, choose a ball $B\subset\Op_1$ and $\mathbf d$
with $B\subset\Op_2+\mathbf d$. For a unit vector $u\in L^2(B)$ the norm
is at least $|(u,\rme^{\rmi\tau h_0/\hbar}u)|$, which tends to $1$; and
it is at most $1$.
\end{proof}

Proposition~\ref{prop:antiLR} is the quantitative form of
Lemma~\ref{lem:galilean_transitivity}. Regions a large distance apart
have vacuum commutators as large as coincident ones, and nearly maximal
at short times. Lattice systems escape because they have no boosts. The
continuum Lieb--Robinson bounds of
Refs.~\onlinecite{GebertNachtergaeleReschkeSims2020,
HinrichsLemmSiebert2025} bound the difference between the interacting
and the free anticommutators, or one-body overlaps with test functions of
controlled regularity in space and momentum. Boosts shift momenta
without bound. By Proposition~\ref{prop:antiLR}, such bounds therefore
cannot hold uniformly over a boost-invariant class of test functions.
Gebert, Nachtergaele, Reschke and Sims note that the free contribution
$\braket{f_t,g}$ does not, in general, decay
exponentially~\cite{GebertNachtergaeleReschkeSims2020}.

\subsection{External potentials}\label{sec:external}

Without Galilean covariance the transitivity argument is not available,
and the harmonic trap of Proposition~\ref{prop:trap} shows that a
confining potential can produce exact commutation at unequal times. For
quasi-free fields in an external potential the question reduces to the
one-particle propagator, and the conclusions that do not use the boosts
survive: commutation only on a closed null set of time separations, none
at short separations for bounded potentials, and none between open
spacetime regions. Let $V:\R^3\to\R$ be measurable, with either
\begin{enumerate}[label=(E\arabic*)]
\item $V\in L^\infty_{\rm loc}(\R^3)$ bounded below, and $\Omega_0:=
\R^3$; or
\item $V\in L^2(\R^3)+L^\infty(\R^3)$ bounded outside every ball centred at
the origin, as for Coulomb and Newtonian potentials, and $\Omega_0:=
\R^3\setminus\{0\}$.
\end{enumerate}
In case~(E1) let $h$ be the self-adjoint operator defined by the closure
of the quadratic form of $h_0+V$ on $C_c^\infty(\R^3)$; in case~(E2),
$h:=h_0+V$ on $D(h_0)$, which is self-adjoint and bounded below by the
Kato--Rellich theorem. In both cases $h$ is bounded below, and every
$u\in D(h)$ lies in $H^2_{\rm loc}(\Omega_0)$ and satisfies $hu=
-\frac{\hbar^2}{2m_0}\Delta u+Vu$ on $\Omega_0$, by interior elliptic
regularity, since $Vu\in L^2_{\rm loc}(\Omega_0)$.

\begin{proposition}[External potentials]\label{prop:external}
Consider the free Bose or Fermi field on the Fock space over
$L^2(\R^3)$ with $\hH:=\rmd\Gamma(h)$, the Fock vacuum, the
finite-particle vectors as $\D$ and $\psio(g):=\sqrt\hbar\,a(g)$, let
$\Op_1,\Op_2\subset\R^3$ be non-empty, bounded and open, and let
$Z(\Op_1,\Op_2)$ be the set of $\tau\neq0$ at which $\F(\Op_1,0)$ and
$\F(\Op_2,\tau)$ commute (bosons), respectively their odd generators
anticommute (fermions).
\begin{enumerate}[label=(\roman*)]
\item $\tau\in Z(\Op_1,\Op_2)$ if and only if $\one_{\Op_1}
\rme^{\rmi\tau h/\hbar}\one_{\Op_2}=0$.
\item $Z(\Op_1,\Op_2)$ is a closed subset of $\R\setminus\{0\}$ and has
Lebesgue measure zero. It is bounded away from $0$ if $\Op_1\cap\Op_2
\neq\emptyset$.
\item If $V$ is bounded, let $\mathrm{osc}\,V:=\frac12(\operatorname{ess\,sup}V-
\operatorname{ess\,inf}V)$,
let $r>0$ be such that $\Op_1$ and $\Op_2$ contain balls of radius $r$,
put $\kappa_r(\tau):=\|\one_{B_r}\rme^{\rmi\tau h_0/\hbar}\one_{B_r}\|$
for a ball $B_r$ of radius $r$, and let $\tau_*$ be the smallest
$\tau>0$ with $\tau\,\mathrm{osc}\,V=\hbar\kappa_r(\tau)$, and
$\tau_*=\infty$ if $V$ is constant. Then $Z(\Op_1,\Op_2)\cap(-\tau_*,
\tau_*)=\emptyset$. By scaling, $\kappa_r(\tau)=K(\hbar\tau/m_0r^2)$
with $K(s):=\|\one_{\{|\bx|<1\}}\rme^{-\rmi s\Delta/2}\one_{\{|\bx|<1\}}\|$
in dimensionless coordinates, so $\tau_*$ depends only on
$\mathrm{osc}\,V$, $r$ and $\hbar/m_0$, and not on the distance between
$\Op_1$ and $\Op_2$.
\item Let $\F(\Op)$, for a non-empty bounded open $\Op\subset\R^3\times
\R$, be generated by the closures of $\int\overline{\chi(t)}\,\evol{t}
\psio(g)\evolm{t}\,\rmd t$ with $B\times I\subset\Op$, $g\in
C_c^\infty(B)$, $\chi\in C_c^\infty(I)$, as after~(T). Then no two such
algebras commute (bosons), respectively have mutually anticommuting odd
generators (fermions).
\end{enumerate}
For the many-body systems of Corollary~\ref{cor:pair} placed in the
external potential, that is, with the bounded pair potential of that
corollary and $h_0$ replaced by $h$, the $n$-particle Hamiltonian being
defined on the domain of the self-adjoint operator $\sum_ih^{(i)}$,
(ii)--(iv) hold,
and (i) holds with ``if and only if'' replaced by ``only if''.
\end{proposition}

\begin{proof}
(i) The vacuum is invariant under $\rme^{\rmi\tau\hH/\hbar}$ and
$\evol{\tau}a(g)\evolm{\tau}=a(\rme^{\rmi\tau h/\hbar}g)$. If the
algebras commute, Lemma~\ref{lem:truncation}, applied as in the proof
of Theorem~\ref{thm:commutation_set}(i), gives $(f,\rme^{\rmi\tau h/
\hbar}g)=0$ for $f\in C_c^\infty(\Op_1)$ and $g\in C_c^\infty(\Op_2)$,
hence for all $f\in L^2(\Op_1)$, $g\in L^2(\Op_2)$. Conversely, if
$\one_{\Op_1}\rme^{\rmi\tau h/\hbar}\one_{\Op_2}=0$, the one-particle
subspaces $L^2(\Op_1)$ and $\rme^{\rmi\tau h/\hbar}L^2(\Op_2)$ carrying
the two algebras are orthogonal, and Lemma~\ref{lem:fock_factorization}
(for fermions, the canonical anticommutation relations) gives the
claim.

(ii) Closedness in $\R\setminus\{0\}$ follows from strong continuity.
If $\Op_1\cap\Op_2\neq\emptyset$, $\one_{\Op_1}\rme^{\rmi\tau h/\hbar}
\one_{\Op_2}$ tends strongly to $\one_{\Op_1\cap\Op_2}\neq0$ as $\tau
\to0$, so it does not vanish for small $|\tau|$. Let $c$ be a lower
bound of $h$. For $f\in L^2(\Op_1)$ and $g\in L^2(\Op_2)$ the function
$z\mapsto(f,\rme^{\rmi z(h-c)/\hbar}g)$ is holomorphic in
$\{\mathrm{Im}\,z>0\}$, bounded and continuous on its closure, and on
the real axis its modulus is that of $k(\tau):=(f,\rme^{\rmi\tau h/
\hbar}g)$. If $Z(\Op_1,\Op_2)$ had positive measure, $k$ would vanish on
it, hence identically, by the theorem of F.~and M.~Riesz used in the
proof of Theorem~\ref{thm:commutation_set}(ii). Then the complex measure
$(f,E_h(\cdot)g)$ has vanishing Fourier transform and is zero, so
$(f,\phi(h)g)=0$ for every bounded Borel function $\phi$. Fixing
$g\neq0$ and letting $f$ run through $L^2(\Op_1)$, every $\phi(h)g$
vanishes on $\Op_1$, and Lemma~\ref{lem:spectral_UC}, applied to
$2m_0h/\hbar^2$, gives $g=0$, a contradiction.

(iii) Replacing $V$ by $V-c$ with $c:=\frac12(\operatorname{ess\,sup}V+
\operatorname{ess\,inf}V)$ multiplies
$\rme^{\rmi\tau h/\hbar}$ by a phase, so we may assume $\|V\|_\infty=
\mathrm{osc}\,V$. Duhamel's formula then gives $\|\rme^{\rmi\tau h/\hbar}
-\rme^{\rmi\tau h_0/\hbar}\|\leq|\tau|\,\mathrm{osc}\,V/\hbar$. Choose
balls $B_1\subset\Op_1$ and $B_2\subset\Op_2$ of radius $r$. By the
proof of Proposition~\ref{prop:antiLR}, $\|\one_{B_1}\rme^{\rmi\tau h_0/\hbar}
\one_{B_2}\|=\kappa_r(\tau)$ whatever the positions of the balls. The
function $\kappa_r$ is even; it is continuous for $\tau\neq0$, because
the kernel of $\one_{B_r}\rme^{\rmi\tau h_0/\hbar}\one_{B_r}$ depends
continuously on $\tau\neq0$ in the Hilbert--Schmidt norm; it is positive
by Lemma~\ref{lem:free_nonvanishing}; and it tends to $1$ as $\tau\to0$.
So $\tau_*$ is well defined and positive. For $0<|\tau|<\tau_*$,
\[
\|\one_{\Op_1}\rme^{\rmi\tau h/\hbar}\one_{\Op_2}\|\geq
\|\one_{B_1}\rme^{\rmi\tau h/\hbar}\one_{B_2}\|\geq
\kappa_r(\tau)-|\tau|\,\mathrm{osc}\,V/\hbar>0,
\]
and (i) applies. The formula for $\kappa_r$ follows from the unitary
dilation $\phi\mapsto r^{3/2}\phi(r\,\cdot)$, which takes a ball of
radius $r$ centred at $0$ to the unit ball and $h_0$ to
$\hbar^2h_0'/m_0r^2$, with $h_0':=-\Delta/2$ in dimensionless
coordinates.

(iv) As in the proof of Corollary~\ref{cor:4D}, commutation of the two
algebras forces $k(\tau)=(f,\rme^{\rmi\tau h/\hbar}g)$ to vanish on an
open interval for all $f\in C_c^\infty(B_1)$ and $g\in C_c^\infty(B_2)$,
where $B_i\times I_i\subset\Op_i$. The holomorphic extension used in
(ii) then vanishes identically, and the argument of (ii) gives $g=0$.

For the many-body systems, the $n$-particle Hamiltonian acts as $h$ on
one-particle vectors and annihilates the vacuum. The set $Z(\Op_1,
\Op_2)$ is then contained in that of the one-particle problem, and the
proofs of the ``only if'' part of (i) and of (ii)--(iv) apply unchanged;
closedness in (ii) again follows from strong continuity.
\end{proof}

For the harmonic trap, $Z(\Op_1,\Op_2)$ is the set of $n\pi/\omega$,
$n\in\Z\setminus\{0\}$, with $\Op_1\cap(-1)^n\Op_2=\emptyset$
(Proposition~\ref{prop:trap}), which is closed and null, as part~(ii)
requires, and bounded away from $0$. Whether a bounded potential can
refocus exactly, that is, whether $Z(\Op_1,\Op_2)$ can be non-empty for
bounded $V$, is left open.

\section{Kinematical classification}\label{sec:classification}
\label{sec:scope_conjecture}

Theorem~\ref{thm:sr_selection} uses the Galilei group only through
Lemma~\ref{lem:galilean_transitivity}. Which other kinematics share this
transitivity? Figueroa-O'Farrill and Prohazka classified the
simply-connected spatially isotropic homogeneous $(D+1)$-dimensional
spacetimes~\cite{FigueroaOFarrillProhazka2019}. For $D=3$ their Table~1
lists the lorentzian (S1--S3), riemannian (S4--S6), galilean (S7--S11)
and carrollian (S13--S16) kinematical spacetimes, and their Table~2 the
aristotelian ones (A21--A23), which have no boosts; the remaining entries
exist only for $D\leq2$. In their basis the isotropy algebra is spanned
by the rotations $J$ and the boosts $B$, the spacetime translations by
$H$ and $P$, and $[J,\cdot\,]$ acts in the standard way. For the galilean
spacetimes $[B,B]=[B,P]=[P,P]=0$, $[H,B]=-P$, and $[H,P]$ is a
combination of $B$ and $P$. Hence $\mathrm{ad}_H$ acts on the
coefficients $(b,p)$ of $bB+pP$, component by component, through a real
$2\times2$ matrix $A_M$. Table~\ref{tab:kinematics} lists the brackets
and the functions
\begin{equation}\label{eq:wronskian}
w_M(\Delta t):=\bigl(\rme^{-\Delta t A_M}\bigr)_{PB},
\end{equation}
the $P$-component of $\rme^{-\Delta tA_M}$ applied to $B$, with
$w_M(0)=0$ and $w_M'(0)=1$. Up to a non-zero factor, $w_M(\Delta t)$ is
the determinant of the two-point boundary problem for the boost
trajectories $t\mapsto P\text{-component of }\rme^{-tA_M}(b,p)$: it
decides whether they can join two prescribed positions at two times, and
its zeros are conjugate times.
Units are those of Ref.~\onlinecite{FigueroaOFarrillProhazka2019},
with the characteristic time of the Newton--Hooke spacetimes set to one.

\begin{table}[htb]
\caption{Spacetimes with absolute time in $3+1$ dimensions
(Ref.~\onlinecite{FigueroaOFarrillProhazka2019}, Tables~1--2). For the
galilean spacetimes the last column gives $w_M$ of~\eqref{eq:wronskian}
and whether the Lie algebra admits the central extension
$[B_a,P_b]=\delta_{ab}M$; for the aristotelian ones it gives a function
of a pair of events with time separation $\Delta t$ that is invariant
under the group.}
\label{tab:kinematics}
\begin{ruledtabular}
\begin{tabular}{llll}
Label & Name & Brackets besides $[J,\cdot\,]$ & $w_M(\Delta t)$ / invariant \\
\hline
S7 & galilean & $[H,B]=-P$ & $\Delta t$; central \\
S8 & galilean dS (Newton--Hooke) & $[H,B]=-P$, $[H,P]=-B$ &
 $\sinh\Delta t$; central \\
S9$_\gamma$ & torsional galilean dS & $[H,B]=-P$, $[H,P]=\gamma B+(1+\gamma)P$ &
 $\frac{\rme^{-\gamma\Delta t}-\rme^{-\Delta t}}{1-\gamma}$ ($\Delta t\,\rme^{-\Delta t}$ if $\gamma=1$); not central \\
S10 & galilean AdS (Newton--Hooke) & $[H,B]=-P$, $[H,P]=B$ &
 $\sin\Delta t$; central \\
S11$_\chi$ & torsional galilean AdS & $[H,B]=-P$, $[H,P]=(1+\chi^2)B+2\chi P$ &
 $\rme^{-\chi\Delta t}\sin\Delta t$; not central \\
\hline
A21 & static & --- & $|\bx_2-\bx_1|$ \\
A22 & torsional static & $[H,P]=P$ & $|\bx_2-\rme^{-\Delta t}\bx_1|$ \\
A23$_\pm$ & $\R\times H^3$, $\R\times S^3$ & $[P,P]=\pm J$ &
 leaf distance $d(\bx_1,\bx_2)$ \\
\end{tabular}
\end{ruledtabular}
\end{table}

\begin{theorem}[Kinematical classification]\label{thm:classification}
Let $M$ be one of the spacetimes of Table~\ref{tab:kinematics}, with
simply connected group $G$ and origin $o$.
\begin{enumerate}[label=(\roman*)]
\item (Galilean spacetimes.) Every point of $M$ is
$\exp(tH)\exp(\bx\cdot P)\cdot o$ for a unique $(\bx,t)\in\R^3\times\R$,
and $G$ acts by $(\bx,t)\mapsto(R(\bx+\bp(t)),t+s)$, where $\bp(t)$ is the
$P$-component of $\rme^{-tA_M}(\mathbf{b},\bp)$, $(\mathbf{b},\bp)$
parametrises $\exp(\mathbf{b}\cdot B+\bp\cdot P)$, $R$ is a rotation and
$s$ a time translation. In particular $G$ acts on each time leaf by
Euclidean motions. It acts transitively on the pairs of events with
time separation $\Delta t\neq0$ if and only if $w_M(\Delta t)\neq0$.
The zeros of $w_M$ other than $0$ form the set $\pi\Z\setminus\{0\}$ for
S10 and S11$_\chi$, and are absent for S7, S8 and S9$_\gamma$.
\item (Aristotelian spacetimes.) $G$ preserves the time leaves, and for
no $\Delta t$ does it act transitively on the pairs of events with time
separation $\Delta t$: the functions in Table~\ref{tab:kinematics} are
invariant.
\item (Exclusion.) For the galilean spacetimes,
Lemma~\ref{lem:galilean_transitivity}(b)--(c) holds for
$\tau\in T(\mathcal{L})\setminus\{\tau:w_M(\tau)=0\}$. Consequently
Theorem~\ref{thm:sr_selection} and
Corollary~\ref{cor:no_unequal_time_locality} hold for canonical nets
whose slice algebras are covariant under $G$ acting as in~(i), with
Galilean invariance replaced by $G$-invariance.
\item (Central extension.) The central extension $[B_a,P_b]=
\delta_{ab}M$ exists exactly when $[H,P]$ has no $P$-component, that is,
for S7, S8 and S10.
\end{enumerate}
\end{theorem}

\begin{proof}
(i) The span of $J$, $B$, $P$ is an ideal, on which $H$ acts, so
$G=\exp(\R H)\,K$ with $K$ the normal subgroup it generates, and
$t:G\to\R$ with $\exp(sH)k\mapsto s$ is a homomorphism that is constant
on the isotropy group. Every point has the stated form, with $t$ given
by this homomorphism and $\bx$ determined by the $K$-component. Since
$B$ and $P$ span an abelian ideal on which $H$ acts by $A_M$,
$\exp(-tH)\exp(\mathbf{b}\cdot B+\bp\cdot P)\exp(tH)=\exp(\mathbf{b}(t)
\cdot B+\bp(t)\cdot P)$ with $(\mathbf{b}(t),\bp(t))=\rme^{-tA_M}
(\mathbf{b},\bp)$. The boosts fix $o$ and commute with $P$, so
$\exp(\mathbf{b}\cdot B+\bp\cdot P)$ maps $(\bx,t)$ to $(\bx+\bp(t),t)$.
Rotations commute with $H$ and rotate $P$, and $\exp(sH)$ shifts $t$.
Given pairs with $t_2-t_1=t_2'-t_1'=\Delta t$, choose $s=t_1'-t_1$ and
$R=\one$. It remains to solve $\bp(t_i)=\bx_i'-\bx_i$, $i=1,2$, for
$(\mathbf{b},\bp)$, component by component. The rows
$e_P^{\mathsf T}\rme^{-t_1A_M}$ and $e_P^{\mathsf T}\rme^{-t_2A_M}=
e_P^{\mathsf T}\rme^{-\Delta tA_M}\rme^{-t_1A_M}$ of this $2\times2$
system are independent if and only if $e_P^{\mathsf T}$ and
$e_P^{\mathsf T}\rme^{-\Delta tA_M}=(w_M(\Delta t),\,\cdot\,)$ are, that
is, if and only if $w_M(\Delta t)\neq0$. If $w_M(\Delta t)=0$, then
$\bp(t_2)=\lambda\bp(t_1)$ for every group element, with
$\lambda=(\rme^{-\Delta tA_M})_{PP}\neq0$. Hence $|\bx_2-\lambda\bx_1|$
is invariant, and the action is not transitive. The explicit $w_M$ in
Table~\ref{tab:kinematics} follow by exponentiating the $2\times2$
matrices $A_M$ read off from the brackets. The zero sets follow from
these formulas.

(ii) For A21 the group acts by $(\bx,t)\mapsto(R\bx+\ba,t+s)$. For A22,
in the coordinates of~(i), $\exp(\ba\cdot P)$ acts on the leaf $t$ as
$\bx\mapsto\bx+\rme^{-t}\ba$. Then $\bx_2-\rme^{-\Delta t}\bx_1$ is
mapped to $R(\bx_2-\rme^{-\Delta t}\bx_1)$ because $\rme^{-\Delta t}
\rme^{-t_1}=\rme^{-t_2}$. For A23$_\pm$, $H$ is central and $J,P$ span
the isometry algebra of $H^3$ or $S^3$, so $G$ acts on
$M=\R\times H^3$ or $\R\times S^3$ by time translations and isometries
of the leaf.

(iii) Since $w_M$ is analytic with $w_M(0)=0$ and $w_M'(0)=1$, it has no
zeros in $0<|\Delta t|<\delta$ for some $\delta>0$. For $\tau$ with
$w_M(\tau)\neq0$ the proof of Lemma~\ref{lem:galilean_transitivity}(b)
goes through, since the group maps balls of a leaf to balls of the same
radius by~(i). If $\inf T(\mathcal{L})=0$, there are $\tau_n\in
T(\mathcal{L})$ with $\tau_n\to0$ and $w_M(\tau_n)\neq0$, and the proof
of Theorem~\ref{thm:sr_selection} applies unchanged.

(iv) The Jacobi identity for $(H,B_a,P_b)$ with $[B_a,P_b]=\delta_{ab}M$
central and $[H,P]=\alpha B+\beta P$ gives $\beta\delta_{ab}M=0$.
Conversely, for $\beta=0$ the Jacobi identities for $(H,B_a,B_b)$ and
$(H,P_a,P_b)$ reduce to $\delta_{ab}M-\delta_{ba}M=0$ and
$\alpha(\delta_{ab}-\delta_{ba})M=0$, and the others hold because $M$ is
central, $[B,B]=[P,P]=0$ and the brackets are rotation covariant.
\end{proof}

For S9$_\gamma$ and S11$_\chi$ the Bargmann-type extension exists but is
not central~\cite{FigueroaOFarrillGormezVandenBleeken2023}, so the
analysis of Sec.~\ref{sec:reeh_schlieder_strand} does not apply to them
as it stands. The exceptional time separations of the oscillating case
are realised, and Fig.~\ref{fig:worldlines}(c) shows their classical
origin: the world lines through an event refocus.

\begin{proposition}[Harmonic trap]\label{prop:trap}
Let $\omega>0$, let $h_\omega:=h_0+\frac12m_0\omega^2|\bx|^2$ on
$L^2(\R^3)$, and consider the free Bose field on the Fock space over
$L^2(\R^3)$ with $\hH:=\rmd\Gamma(h_\omega)$, $\psio(g):=
\sqrt\hbar\,a(g)$, the Fock vacuum and the finite-particle vectors.
\begin{enumerate}[label=(\alph*)]
\item $\hH\geq0$. Writing $\Gamma(u)$ for the second quantisation of a
one-particle unitary $u$, one has
$\Gamma(\tau_\ba)\F(\Op,t)\Gamma(\tau_\ba)^{-1}=
\F(\Op+\ba\cos\omega t,t)$ and
$\Gamma(e_\bv)\F(\Op,t)\Gamma(e_\bv)^{-1}=
\F(\Op+\omega^{-1}\bv\sin\omega t,t)$; rotations and time translations
act as in Lemma~\ref{lem:covariance}. The slice algebras are thus
covariant under the group of S10 with characteristic time $1/\omega$,
acting as in Theorem~\ref{thm:classification}(i).
\item For $\tau\neq0$ and non-empty bounded open $\Op_1,\Op_2$,
$[\F(\Op_1,0),\F(\Op_2,\tau)]=\{0\}$ if and only if $\omega\tau=n\pi$
with $n\in\Z\setminus\{0\}$ and $\Op_1\cap(-1)^n\Op_2=\emptyset$.
\end{enumerate}
\end{proposition}

\begin{proof}
(a) $h_\omega\geq\frac32\hbar\omega$, so $\hH\geq0$. Let $X$ and $P$ be
the one-particle position and momentum and
$W(\mathbf q,\mathbf k):=\rme^{\rmi(\mathbf k\cdot X-\mathbf q\cdot
P)/\hbar}$. Up to a constant phase, $W(\mathbf q,\mathbf k)$ is a
translation by $\mathbf q$ followed by multiplication by
$\rme^{\rmi\mathbf k\cdot\bx/\hbar}$. The ladder-operator computation
on finite linear combinations of Hermite functions, which are analytic
vectors for linear combinations of $X$ and $P$, gives
$\rme^{\rmi th_\omega/\hbar}X\rme^{-\rmi th_\omega/\hbar}=X\cos\omega t+
(m_0\omega)^{-1}P\sin\omega t$ and $\rme^{\rmi th_\omega/\hbar}P
\rme^{-\rmi th_\omega/\hbar}=P\cos\omega t-m_0\omega X\sin\omega t$.
Hence
\[
u_t:=\rme^{-\rmi th_\omega/\hbar}W(\mathbf q,\mathbf k)
\rme^{\rmi th_\omega/\hbar}=W\bigl(\mathbf q\cos\omega t+
(m_0\omega)^{-1}\mathbf k\sin\omega t,\ \mathbf k\cos\omega t-
m_0\omega\mathbf q\sin\omega t\bigr).
\]
Since $\Gamma(u)\evol{t}=\evol{t}\Gamma(u_t)$, we get
$\Gamma(u)\F(\Op,t)\Gamma(u)^{-1}=\evol{t}\Gamma(u_t)\F(\Op,0)
\Gamma(u_t)^{-1}\evolm{t}$. By the covariance of Fock fields,
$\Gamma(u_t)\F(\Op,0)\Gamma(u_t)^{-1}$ is the slice algebra at time
zero of the translate of $\Op$ by the first argument of $u_t$. For
$\tau_\ba=W(\ba,0)$ this translation is $\ba\cos\omega t$, and for
$e_\bv=W(0,m_0\bv)$ it is $\omega^{-1}\bv\sin\omega t$.

(b) The field is free, so by Lemma~\ref{lem:truncation} (for $\Rightarrow$)
and Lemma~\ref{lem:fock_factorization} (for $\Leftarrow$) the two
algebras commute if and only if $(f,\rme^{\rmi\tau h_\omega/\hbar}g)=0$
for all $f\in C_c^\infty(\Op_1)$ and $g\in C_c^\infty(\Op_2)$. Suppose
first $\sin\omega\tau\neq0$. By the formula in~(a),
$\rme^{\rmi\tau h_\omega/\hbar}(e_\bv g)$ is, up to a unimodular factor,
$\rme^{\rmi\tau h_\omega/\hbar}g$ translated by $-\omega^{-1}\bv
\sin\omega\tau$ and multiplied by a unimodular function. As in
Lemma~\ref{lem:free_nonvanishing}, vanishing of all matrix elements
would force $\rme^{\rmi\tau h_\omega/\hbar}g=0$ on every translate of
$\Op_1$, hence $g=0$. Now let $\omega\tau=n\pi$, $n\neq0$. The Hermite
functions form an eigenbasis of $h_\omega$ with eigenvalues
$\hbar\omega(|\mathbf k|+\frac32)$ and parity $(-1)^{|\mathbf k|}$,
$\mathbf k\in\N^3$. Hence $\rme^{\rmi n\pi h_\omega/\hbar\omega}=
\rme^{3\pi\rmi n/2}\,\mathsf{P}^n$ with $(\mathsf Pg)(\bx)=g(-\bx)$, and
$(f,\rme^{\rmi\tau h_\omega/\hbar}g)=\rme^{3\pi\rmi n/2}\int
\overline{f(\bx)}g((-1)^n\bx)\,\rmd^3x$. This vanishes for all $f,g$ if
and only if $\Op_1\cap(-1)^n\Op_2=\emptyset$.
\end{proof}

At half periods the trap refocuses $\bx$ onto $-\bx$, and at full
periods onto $\bx$. The resulting relation between slice regions at time
separation $n\pi/\omega$, $\Op_1\cap(-1)^n\Op_2=\emptyset$, is invariant
under the group of~(a), because the leaf shifts at the two times differ
by the factor $(-1)^n$. By Theorem~\ref{thm:classification}(i) these are
exactly the separations at which the group fails to be transitive.
Theorem~\ref{thm:sr_selection}, extended as in
Theorem~\ref{thm:classification}(iii), is therefore sharp for the
oscillating Newton--Hooke kinematics.

\begin{remark}[Kinematics that admit unequal-time locality]
\label{rem:contrast}
For the Minkowski spacetime S1, the free Klein--Gordon field is a
Poincar\'e-covariant canonical net with positive energy whose algebras
commute at spacelike separation~\cite{Haag1996}. Spacelike separation is
an invariant relation that contains pairs of regions at unequal times.
The same net is covariant under the static group A21, which preserves
this relation. For the Carroll spacetime S13 the group acts by
$(\bx,t)\mapsto(R\bx+\ba,\,t+\bv\cdot\bx+s)$. Let
$\hat\phi(\bx,t):=(\hbar/2\omega_0)^{1/2}(\hat a(\bx)\rme^{-\rmi
\omega_0t}+\hat a^*(\bx)\rme^{\rmi\omega_0t})$ on the Fock space over
$L^2(\R^3)$, with $\hH:=\hbar\omega_0\,\rmd\Gamma(\one)\geq0$ and
$\hat\pi:=\partial_t\hat\phi$, so that
$[\hat\phi(\bx,0),\hat\pi(\by,0)]=\rmi\hbar\delta(\bx-\by)$. The Carroll
boost is implemented by $\Gamma(\rme^{\rmi\omega_0\bv\cdot\bx})$, which
maps $\hat\phi(\bx,t)$ to $\hat\phi(\bx,t+\bv\cdot\bx)$ and leaves the
Fock vacuum invariant. For real $F\in C_c^\infty(\R^3\times\R)$ the
smeared field is $(\hbar/\omega_0)^{1/2}$ times the Segal field of the
one-particle vector $\int F(\cdot,t)\rme^{\rmi\omega_0t}\rmd t$, whose
support lies in the spatial projection of $\supp F$. Algebras of regions
with disjoint spatial projections therefore commute, by
Lemma~\ref{lem:fock_factorization}. Disjointness of spatial projections
is a Carroll-invariant relation, and it contains pairs at every time
separation. The Carroll boosts do not map slice regions to slice regions,
so this net is covariant at the level of the algebras of open spacetime
regions rather than in the sense of~(S2); its slice algebras do not
depend on time. In both cases an invariant separates ``there, later'' from
``here, later'': the Lorentzian interval, and the distance of absolute
space. Free Carrollian fields have been analysed algebraically in
Ref.~\onlinecite{FredenhagenProhazkaTiefenbacher2026}.
\end{remark}

\section{Thermal states and media at rest}\label{sec:media}

A state of non-zero density or temperature is not invariant under
unitarily implemented boosts, so it falls outside (H3). The Galilean
dynamics may still act by automorphisms. Theorem~\ref{thm:sr_selection}
then applies unchanged (Remark~\ref{rem:where_positivity_enters}): no
Galilean-invariant commutation relation is realised at short times. For
larger time separations, Theorem~\ref{thm:commutation_set} used the
vacuum twice: through positivity, which makes vacuum correlations
boundary values of functions holomorphic in a half-plane, and through
$\psio(g)\Omega=0$, which makes the commutator one-sided. Neither can be
dropped: without positivity, Proposition~\ref{prop:attainable}(e) gives
commutation on a half-line for fields that annihilate the vacuum, and
Proposition~\ref{prop:KMS_no} below gives it for a positive-energy field
that does not. One might hope to replace both by the analyticity of KMS
states in a strip. The
following example shows that this does not work.

\begin{proposition}[KMS analyticity does not replace one-sidedness]
\label{prop:KMS_no}
For every $a>0$ there is a real finite measure $\nu$ on $\R$ with
$\nu(\R)=1$, $\mathrm{sign}(E)\,\nu\geq0$ and $\supp\widehat\nu\subset
[-a,a]$. For such $\nu$ let $\nu_+$ be the restriction of $\nu$ to
$(0,\infty)$ and $\nu_-$ the image of the restriction of $-\nu$ to
$(-\infty,0)$ under $E\mapsto-E$. Let $\K_\pm:=L^2(\R^3)\otimes
L^2((0,\infty),\nu_\pm)$, let $h$ be multiplication by $E$ on both, let
$\hH:=\rmd\Gamma(h)$ on the bosonic Fock space over $\K_+\oplus\K_-$,
and set $\hat\psi(g):=\sqrt\hbar\,[a((g\otimes1)\oplus0)+
a^*(0\oplus(\bar g\otimes1))]$ on finite-particle vectors. Then
$\hH\geq0$, the Fock vacuum is invariant, $[\hat\psi(f),\hat\psi(g)^*]=
\hbar(f,g)$ and $[\hat\psi(f),\hat\psi(g)]=0$, and
\begin{equation}\label{eq:KMS_no}
[\hat\psi(f),\evol{\tau}\hat\psi(g)^*\evolm{\tau}]=\hbar\,(f,g)\,
\widehat\nu(\tau/\hbar),\qquad
[\hat\psi(f),\evol{\tau}\hat\psi(g)\evolm{\tau}]=0 .
\end{equation}
Let $\mathcal{W}(\Op,t)$ be the von Neumann algebra generated by the
Weyl operators of the self-adjoint fields $\overline{\hat\psi(g)+
\hat\psi(g)^*}$ and $\overline{\rmi(\hat\psi(g)-\hat\psi(g)^*)}$,
$g\in C_c^\infty(\Op)$, translated to time $t$. Then
$\mathcal{W}(\Op_1,0)$ and $\mathcal{W}(\Op_2,\tau)$ commute for all
$\Op_1,\Op_2$ whenever $|\tau|\geq a\hbar$, and for all $\tau$ when
$\Op_1\cap\Op_2=\emptyset$.
\end{proposition}

\begin{proof}
Choose a non-zero real $\chi\in C_c^\infty((-a/2,a/2))$ and $k_0>0$, put
$f_0(s):=\chi(s)\rme^{\rmi k_0s}$ and $\hat f_0(E):=\int\rme^{-\rmi Es}
f_0(s)\,\rmd s$, and let $\nu:=c_0^{-1}E|\hat f_0(E)|^2\,\rmd E$ with
$c_0:=\int E|\hat f_0|^2\rmd E=2\pi(f_0,-\rmi f_0')=2\pi k_0\|\chi\|^2>0$.
Then $\nu(\R)=1$ and $\mathrm{sign}(E)\nu\geq0$. Moreover
$\int\rme^{\rmi Es}|\hat f_0(E)|^2\rmd E=2\pi\int f_0(s+u)
\overline{f_0(u)}\,\rmd u$ is supported in $[-a,a]$, and
$\widehat\nu(s)=-\rmi c_0^{-1}\frac{\rmd}{\rmd s}$ of it. The
commutators follow from the canonical commutation relations of
$a,a^*$, from $\evol{\tau}a^\sharp(k)\evolm{\tau}=
a^\sharp(\rme^{\rmi\tau h/\hbar}k)$, and from
$\int\rme^{\rmi E\tau/\hbar}\rmd\nu_+-\int\rme^{-\rmi E\tau/\hbar}
\rmd\nu_-=\widehat\nu(\tau/\hbar)$. The real-linear combinations of
$\hat\psi(g)$ and $\hat\psi(g)^*$ are $\sqrt{2\hbar}$ times Segal fields
$\Phi(k_u)$, with $k_u:=(u\otimes1)\oplus(\bar u\otimes1)$ and
$u=g$ or $-\rmi g$. Two families of Weyl operators commute if and only if
$\mathrm{Im}\braket{k,k'}=0$ for all pairs, and
\[
\braket{k_u,\rme^{\rmi\tau h/\hbar}k_{u'}}=(u,u')\,\widehat{\nu_+}
(\tau/\hbar)+\overline{(u,u')}\,\widehat{\nu_-}(\tau/\hbar),
\]
with $\widehat{\nu_\pm}(s):=\int\rme^{\rmi Es}\rmd\nu_\pm$. For $(u,u')$
ranging over a complex subspace, the imaginary part vanishes identically
if and only if $\widehat{\nu_+}(\tau/\hbar)=\overline{\widehat{\nu_-}
(\tau/\hbar)}$, that is, $\widehat\nu(\tau/\hbar)=0$, or if $(u,u')=0$
identically.
\end{proof}

The field of Proposition~\ref{prop:KMS_no} has positive energy, a
canonical commutator and quasi-free KMS states at every inverse
temperature: the one-particle density $(\rme^{\beta h}-1)^{-1}$ is finite
on the relevant test vectors because $\nu_\pm$ vanish linearly at
$E=0$~\cite{ArakiWoods1963,BratteliRobinson1997}. Its commutators are
c-numbers, so they take the values~\eqref{eq:KMS_no} in every state, and
they vanish for all time separations of at least $a\hbar$. For a
c-number commutator, the KMS condition constrains the spectral measure
$\nu$ essentially only through its sign, $\mathrm{sign}(E)\nu\geq0$,
and this example satisfies that constraint. What excludes it in the
Galilean setting is the Bargmann structure together with positivity;
without positivity, Proposition~\ref{prop:attainable}(e) gives
commutation at all separations beyond a threshold even for fields that
annihilate the vacuum. By
Lemmas~\ref{lem:joint_spectrum} and~\ref{lem:vacuum_annihilation}, a
field of Bargmann charge $m_0>0$ annihilates the vacuum, since
$\hat\psi(g)\Omega$ would have mass $-m_0$, which (H4) forbids. The
field of Proposition~\ref{prop:KMS_no} does not annihilate the vacuum.
Nor is it Galilean covariant: its commutator is ultralocal in space,
whereas under Galilean covariance the size of the vacuum commutator,
measured as in Proposition~\ref{prop:antiLR}, does not decrease with the
spatial separation.

A positive result follows from the fact that commutation is an algebraic
property. A dynamics of a C*-algebra is here a one-parameter group of
automorphisms, not necessarily continuous in time. We say that a
representation $\pi$ of a C*-algebra $\mathfrak{A}$ with a dynamics
$\alpha$ and an action of the Bargmann group by automorphisms satisfies
(H) if both are implemented by a unitary Bargmann representation on the
space of $\pi$, and if, for a unit vector $\Omega$ and a dense subspace
$\D$, the fields $\pi(\psio(g))$---for the Weyl algebra, in a regular
representation, the restrictions to $\D$ of the corresponding
combinations of the generators
of $t\mapsto\pi(W(tk))$---satisfy (S1), (S3) or the anticommutation
relations of Proposition~\ref{prop:fermi_obs}(a), and (H1)--(H4), with
$\F(\Op,0)=\pi(\mathfrak{A}(\Op))''$.

\begin{proposition}[Transfer]\label{prop:transfer}
Let $\mathfrak{A}$ be a C*-algebra with a one-parameter group of
automorphisms $\alpha$, not necessarily continuous in $t$, and
C*-subalgebras $\mathfrak{A}(\Op)$, and put
$\mathfrak{A}(\Op,t):=\alpha_t(\mathfrak{A}(\Op))$. Let $\pi_0$ be a
representation in which some $X\in\mathfrak{A}(\Op_1,0)$ and
$Y\in\mathfrak{A}(\Op_2,\tau)$ satisfy $\pi_0([X,Y])\neq0$. Then, for
every representation $\pi$ with $\ker\pi\subset\ker\pi_0$, the algebras
$\pi(\mathfrak{A}(\Op_1,0))''$ and $\pi(\mathfrak{A}(\Op_2,\tau))''$ do
not commute. If a grading of $\mathfrak{A}$ commutes with $\alpha$ and
leaves every $\mathfrak{A}(\Op)$ invariant, and $X,Y$ are odd with
$\pi_0(\{X,Y\})\neq0$, then $\pi(X)$ and $\pi(Y)$ do not anticommute in
any such $\pi$. If $\mathfrak{A}$ is simple, all this holds for every
non-zero representation, in particular in the GNS representation of
every state, equilibrium or not.
\end{proposition}

\begin{proof}
$[X,Y]$, respectively $\{X,Y\}$, is not in $\ker\pi_0\supset\ker\pi$, and
every non-zero representation of a simple C*-algebra is faithful.
\end{proof}

\begin{corollary}[Free gases]\label{cor:free_gases}
(a) Let $S\subset L^2(\R^3)$ be a complex subspace containing
$C_c^\infty(\R^3)$ and invariant under $\rme^{\rmi th_0/\hbar}$. In every
representation of the Weyl relations over $S$ in which the free time
evolution acts by $W(f)\mapsto W(\rme^{\rmi th_0/\hbar}f)$, the algebras
generated by $\{W(f):f\in C_c^\infty(\Op_1)\}$ and
$\{W(\rme^{\rmi\tau h_0/\hbar}g):g\in C_c^\infty(\Op_2)\}$ do not commute
for $\tau\neq0$. (b) The same holds for the CAR algebra over $L^2(\R^3)$
with the free dynamics $a(f)\mapsto a(\rme^{\rmi th_0/\hbar}f)$, in every
representation. In particular this covers all equilibrium states of the
free Bose and Fermi gases on these algebras~\cite{ArakiWoods1963,
BratteliRobinson1997}.
\end{corollary}

\begin{proof}
(a) In every representation $W(f)W(g)=\rme^{-\rmi\,\mathrm{Im}(f,g)}
W(g)W(f)$, with $\hbar$ absorbed in the normalisation. If all generators
commuted, then $\mathrm{Im}(f,\rme^{\rmi\tau h_0/\hbar}g)\in2\pi\Z$ for
all $f,g$. Replacing $f$ by $sf$, $s\in\R$, gives
$\mathrm{Im}(f,\rme^{\rmi\tau h_0/\hbar}g)=0$, and replacing $f$ by
$\rmi f$ gives $(f,\rme^{\rmi\tau h_0/\hbar}g)=0$, contradicting
Lemma~\ref{lem:free_nonvanishing}. (b) In every representation
$\{a(f),a^*(\rme^{\rmi\tau h_0/\hbar}g)\}=(f,\rme^{\rmi\tau h_0/\hbar}g)
\one$, and the same lemma applies.
\end{proof}

For interacting systems the transfer requires the dynamics to act by
automorphisms of a C*-algebra that contains the local algebras, not only
in one representation; no continuity in time is needed. For fermions
with UV-regularized pair interactions such a dynamics has been
constructed~\cite{GebertNachtergaeleReschkeSims2020}.

\begin{corollary}[Regularized interacting fermions]\label{cor:GNRS}
Let $\tau_t$ be the strongly continuous group of automorphisms of the
CAR algebra over $L^2(\R^d)$ constructed in
Ref.~\onlinecite{GebertNachtergaeleReschkeSims2020} (Theorem~2.7 there)
for the one-body operator $-\Delta$, that is, their external potential
set to zero, and a pair potential satisfying their Assumption~2.1,
smeared with Gaussians. For $t\neq0$ and
non-empty open $\Op_1,\Op_2\subset\R^d$ there are $f\in
C_c^\infty(\Op_1)$ and $g\in C_c^\infty(\Op_2)$ with
$\{a(f),\tau_t(a^*(g))\}\neq0$. Hence, in every non-zero
representation, the time-zero and time-$t$ field algebras of $\Op_1$ and
$\Op_2$ fail to anticommute.
\end{corollary}

\begin{proof}
In Ref.~\onlinecite{GebertNachtergaeleReschkeSims2020}, $\tau_t$ is the
norm limit of $\tau_t^{\Lambda}=\mathrm{Ad}\,\rme^{\rmi tH_\Lambda}$,
where $H_\Lambda$ is the free second-quantised Laplacian plus a bounded
interaction containing two annihilation operators on its right. The
Fock vacuum $\Omega_F$ is therefore invariant under
$\rme^{-\rmi tH_\Lambda}$, and $\rme^{\rmi tH_\Lambda}a^*(g)\Omega_F=
a^*(\rme^{-\rmi t\Delta}g)\Omega_F$. Hence $\braket{\Omega_F,\{a(f),
\tau_t^\Lambda(a^*(g))\}\Omega_F}=(f,\rme^{-\rmi t\Delta}g)$, uniformly
in $\Lambda$, and the same holds for $\tau_t$. This is non-zero for
suitable $f,g$ by Lemma~\ref{lem:free_nonvanishing}, applied with
$\hbar=1$ and $m_0=1/2$. The CAR algebra is simple, and
Proposition~\ref{prop:transfer} applies.
\end{proof}

These regularised interactions are translation invariant but, because of
the smearing, not boost covariant. The mechanism behind
Corollary~\ref{cor:GNRS} is the free propagation of a single particle,
which in a Galilean theory is enforced by boost covariance
(Lemma~\ref{lem:GKL}).

For a Galilei-covariant interaction the transfer holds as well.
Narnhofer and Thirring~\cite{NarnhoferThirring1990,NarnhoferThirring1991}
constructed a dynamics of fermions with a two-body interaction in normal
order, written with the annihilation operators $a_{q,p}$ smeared with
Gaussian coherent states centred at $(q,p)$ and a potential
$v(q-q',p-p')$ integrable in phase space, which cuts the interaction off
at high relative velocities (Ref.~\onlinecite{NarnhoferThirring1991},
Eqs.~(4.2) and~(4.3)). It is a norm-continuous group of automorphisms of
the CAR algebra over $L^2(\R^3)$ (Ref.~\onlinecite{NarnhoferThirring1991},
Theorem~4.9), and it is covariant under the Galilei group: with the
automorphisms implementing space translations, boosts and gauge
transformations it satisfies the relations of the Bargmann group
(Ref.~\onlinecite{NarnhoferThirring1991}, Definition~5.9 and
Lemma~5.10). In the proof of Corollary~4.10 of
Ref.~\onlinecite{NarnhoferThirring1991} it is the norm limit of the
dynamics generated by the free kinetic energy plus the interaction
restricted to bounded regions of the two-particle phase space, bounded
operators with two annihilation operators on their right, and the weak*
limit points of KMS states of these dynamics are KMS states of the
limit, which therefore exist at every temperature; with a chemical
potential see Ref.~\onlinecite{Narnhofer2020a}.

\begin{corollary}[Normal-ordered cutoff interactions]\label{cor:NT}
Let $\tau_t$ be a group of automorphisms of the CAR algebra over
$L^2(\R^3)$ such that, for all $t$ and all $A$ in the algebra, $\tau_t(A)$
is the norm limit, in the Fock representation and along a sequence of
indices $\Lambda$, of $\tau_t^\Lambda(A):=\rme^{\rmi tH_\Lambda/\hbar}A\,
\rme^{-\rmi tH_\Lambda/\hbar}$, where $H_\Lambda=\rmd\Gamma(h_0)+V_\Lambda$ and
$V_\Lambda$ is bounded, self-adjoint, commutes with the number operator
and vanishes on the vectors with at most one particle, as does a
normal-ordered interaction with two annihilation operators on its right.
This holds for the dynamics of
Refs.~\onlinecite{NarnhoferThirring1990,NarnhoferThirring1991}.
For $t\neq0$ and non-empty open $\Op_1,\Op_2\subset\R^3$ there are
$f\in C_c^\infty(\Op_1)$ and $g\in C_c^\infty(\Op_2)$ with
$\{a(f),\tau_t(a^*(g))\}\neq0$. Hence, in every non-zero representation,
in particular in the GNS representation of every KMS state of $\tau$,
with or without chemical potential, the time-zero and time-$t$ field
algebras of $\Op_1$ and $\Op_2$ fail to anticommute.
\end{corollary}

\begin{proof}
As for Corollary~\ref{cor:GNRS}. Since $V_\Lambda$ vanishes on the vectors
with at most one particle, $\Omega_F$ is invariant under
$\rme^{-\rmi tH_\Lambda/\hbar}$ and $\rme^{\rmi tH_\Lambda/\hbar}a^*(g)\Omega_F
=a^*(\rme^{\rmi th_0/\hbar}g)\Omega_F$. Hence
\[
\braket{\Omega_F,\{a(f),\tau_t^\Lambda(a^*(g))\}\Omega_F}
=(f,\rme^{\rmi th_0/\hbar}g)
\]
for every $\Lambda$, and the same holds for $\tau_t$. This is non-zero for
suitable $f,g$ by Lemma~\ref{lem:free_nonvanishing}. Since $V_\Lambda$
commutes with the number operator, $\tau_t$ preserves the parity, so
$\tau_t(a^*(g))$ is odd. The CAR algebra is simple, and
Proposition~\ref{prop:transfer} applies.
\end{proof}

Without positivity the conclusions can fail, even in thermal
equilibrium.

\begin{proposition}[A Fermi system in equilibrium without positivity]
\label{prop:fermi_medium}
Let $\tau_0>0$, let $\nu:=\nu_{\tau_0}$ be the measure~\eqref{eq:polya},
and let $\mathfrak h$, $h$ and the unitary Bargmann representation on
$\mathfrak h$ be as in the proof of
Theorem~\ref{thm:commutation_set}(iv), with $\K=L^2(\R,\nu)$ and
$\xi\equiv1$. On the CAR algebra over $\mathfrak h$ let
$\alpha_t(a(k)):=a(\rme^{\rmi th/\hbar}k)$, let the Bargmann
representation act by the corresponding Bogoliubov automorphisms, its
central group by gauge transformations, put
$\psio(g):=\sqrt\hbar\,a(g\otimes\xi)$, and let $\mathfrak{A}(\Op)$ be the
C*-algebra generated by the $a(g\otimes\xi)$, $g\in C_c^\infty(\Op)$.
\begin{enumerate}[label=(\roman*)]
\item For every $\beta>0$ and $\mu\in\R$ the gauge-invariant quasi-free
state with one-particle density $(1+\rme^{\beta(h-\mu)})^{-1}$ is a KMS
state at inverse temperature $\beta$ for the dynamics
$a(k)\mapsto a(\rme^{\rmi t(h-\mu)/\hbar}k)$.
\item For $|\tau|\geq\tau_0$ and all $f,g\in L^2(\R^3)$, the operators
$a^\sharp(f\otimes\xi)$ and $\alpha_\tau(a^\sharp(g\otimes\xi))$
anticommute, where $a^\sharp$ stands for $a$ or $a^*$. Hence, in every
representation, the odd generators of the slice algebras of any two
regions at times $0$ and $\tau$ anticommute.
\item No representation of this system satisfies (H).
\item Let $\K$ be a Hilbert space, $U_\K$ a self-adjoint operator on $\K$,
$h:=h_0\otimes\one+\one\otimes U_\K$ and $\xi\in\K$ a unit vector, and
let the CAR algebra over $L^2(\R^3)\otimes\K$ carry the dynamics, the
Bargmann action and the fields defined as above. If, for some $\beta>0$
and $\mu\in\R$, the gauge-invariant quasi-free state with one-particle
density $(1+\rme^{\beta(h-\mu)})^{-1}$ has a finite expected number of
particles in some non-empty bounded open $\Lambda\subset\R^3$, then
$U_\K$ is bounded below, the Fock representation satisfies (H1)--(H4)
in the form of Proposition~\ref{prop:fermi_obs}(a), and in every
non-zero representation the odd generators of the slice algebras at
times $0$ and $\tau\neq0$ all anticommute only if $|\tau|\in Z_\xi$, a
closed null set. The states of~(i) have infinitely many particles in every such
$\Lambda$.
\end{enumerate}
\end{proposition}

\begin{proof}
(i) With $\gamma:=\rme^{-\beta(h-\mu)}$ the one-particle density is
$\gamma(1+\gamma)^{-1}$, and the GNS representation of the state is the
Araki--Wyss representation of the CAR associated with $\gamma$. Its
vacuum is a KMS vector for the dynamics generated by $h-\mu$ by
Ref.~\onlinecite{Derezinski2006}, Theorem~43(13) and~(14), which require
only that $h$ be self-adjoint. (ii)
$\{a(f\otimes\xi),\alpha_\tau(a^*(g\otimes\xi))\}=\braket{f\otimes\xi,
\rme^{\rmi\tau h/\hbar}g\otimes\xi}\one=(f,\rme^{\rmi\tau h_0/\hbar}g)\,
\widehat\nu(\tau/\hbar)\one=0$, and the other anticommutators vanish by
the anticommutation relations. These are identities in the CAR algebra,
so they hold in every representation. (iii) In a representation
satisfying (H), Proposition~\ref{prop:fermi_obs}(a) and
Corollary~\ref{thm:sr_selection_analytic} confine the separations at
which the odd generators anticommute to a null set, whereas by~(ii) they
contain $[\tau_0,\infty)$. (iv) Let $\rho:=(1+\rme^{\beta(h-\mu)})^{-1}$
and let $P$ be the projection onto $L^2(\Lambda)\otimes\K$. The expected
number of particles in $\Lambda$ is $\operatorname{Tr}(P\rho P)$, which
is at least $\sum_k\braket{\varphi\otimes\eta_k,\rho\,\varphi\otimes
\eta_k}$ for every unit vector $\varphi\in L^2(\Lambda)$ and every
orthonormal family $(\eta_k)$ in $\K$. If $U_\K$ is not bounded below,
take $\eta_k$ in the spectral subspaces of $U_\K$ for disjoint intervals
contained in $(-\infty,E_k]$, with $E_k\to-\infty$. Since
$(1+\rme^{\beta(x-\mu)})^{-1}$ decreases in $x$, the $k$-th term is at
least $(\varphi,(1+\rme^{\beta(h_0+E_k-\mu)})^{-1}\varphi)$, which tends
to $1$, and the sum diverges. Hence $U_\K\geq E_*$ for some $E_*\in\R$,
$\rmd\Gamma(h)\geq E_*\,\rmd\Gamma(\one)$, and the Fock representation
satisfies (H1)--(H4) as in the proof of
Theorem~\ref{thm:commutation_set}(iv). With $\nu_\xi$ the spectral
measure of $U_\K$ in $\xi$,
$\{a(f\otimes\xi),\alpha_\tau(a^*(g\otimes\xi))\}=
(f,\rme^{\rmi\tau h_0/\hbar}g)\,\widehat{\nu_\xi}(\tau/\hbar)\one$ is a
c-number, so by Lemma~\ref{lem:free_nonvanishing} the odd generators all
anticommute, in any non-zero representation, only if
$\widehat{\nu_\xi}(\tau/\hbar)=0$; this set is null by
Theorem~\ref{thm:commutation_set}(ii). In~(i), $U_\K$ is multiplication
by $E$ on $L^2(\R,\nu_{\tau_0})$, which is not bounded below.
\end{proof}

The states of~(i) are equilibrium states of a Fermi system at rest, with
internal levels unbounded below filled by the Fermi--Dirac distribution.
By~(iv) their density is infinite, and so is that of every quasi-free
equilibrium state of a Galilean Fermi system of this form whose odd
fields anticommute on a set of time separations of positive measure.
Boosts act by Bogoliubov automorphisms that do not leave these states
invariant. For bosons the construction of~(i) has no gauge-invariant
quasi-free equilibrium state: the density $(\rme^{\beta(h-\mu)}-1)^{-1}$
is positive only if $h\geq\mu$, that is, $U_\K\geq\mu$. For a quasi-free
Galilean Bose system with one-particle structure $L^2(\R^3)\otimes\K$
and $h=h_0\otimes\one+\one\otimes U_\K$, a gauge-invariant quasi-free
equilibrium state therefore forces $U_\K\geq\mu$, hence (H4) in the Fock
representation with $\lambda=\max(0,-\mu/m_0)$, and commutation is
confined to the null set $Z_\xi$ in every non-zero representation, by
the Weyl relations, as in Corollary~\ref{cor:free_gases}(a), or by
Proposition~\ref{prop:transfer}. For quasi-free systems,
equilibrium thus forces positivity for bosons, and equilibrium at finite
density does so for fermions.

\begin{remark}[Media at rest: what remains open]\label{rem:massive_vacuum}
A medium at rest of non-zero density has infinite mass in infinite
volume. It is not a normalisable translation-invariant vector of a
Bargmann representation, and in its GNS representation the boosts act
by automorphisms that need not be unitarily implemented. By
Proposition~\ref{prop:transfer}, the results above extend to every
representation of a simple C*-algebra with a dynamics that has one
representation satisfying (H), and, if the algebra is not simple, to
every representation whose kernel is contained in that of such a
representation. Without such a representation they can fail:
Proposition~\ref{prop:attainable}(e) and
Proposition~\ref{prop:fermi_medium} give Galilean-covariant systems of
canonical fields, the second with equilibrium states at every
temperature, whose slice algebras commute, respectively whose odd
generators anticommute, at all time separations of at least $\tau_0$.
What they lack is positivity, and the equilibrium states of the second
have infinite density. For quasi-free systems, equilibrium forces
positivity for bosons, and equilibrium at finite density does so for
fermions (Proposition~\ref{prop:fermi_medium}(iv)). Which conditions
restore the conclusions for interacting media is open; finite density
is a natural candidate. For locality realised at short times,
Theorem~\ref{thm:sr_selection}, in the form of
Proposition~\ref{prop:fermi_obs}(a) for Fermi fields, excludes all these
cases whenever the system has a time-invariant state in whose GNS
representation (S1) and the canonical (anti)commutation relations hold.

For Galilei-covariant interacting fermions with the cutoff of Narnhofer
and Thirring~\cite{NarnhoferThirring1991} the transfer holds
(Corollary~\ref{cor:NT}). For Bose gases with continuous
pair potentials that vanish at infinity, Buchholz constructs the dynamics
on an extension $\bar{\mathfrak A}$ of the gauge-invariant subalgebra of
the resolvent algebra~\cite{Buchholz2018,Buchholz2020}. It is a
C*-algebra of operators on Fock space, on which the dynamics acts by
automorphisms that are continuous only on the subspaces of fixed
particle number; Proposition~\ref{prop:transfer} needs no more. The
algebra $\bar{\mathfrak A}$ is not simple, and non-commutation cannot hold
in all its representations: the vacuum representation is
one-dimensional~\cite{Buchholz2020}, so everything commutes in it, and
the Fock vacuum is a KMS state at every temperature that arises as a
thermodynamic limit of canonical ensembles~\cite{Buchholz2026}. The first part of
Proposition~\ref{prop:transfer} applies with $\pi_0$ the one-particle
representation $\rho_1$, on which the pair interaction acts trivially.
For normalised $f\in C_c^\infty(\Op_1)$ and $g\in C_c^\infty(\Op_2)$, the
resolvent $(1+N_f)^{-1}$, $N_f:=a^*(f)a(f)$, belongs to the algebra of
$\Op_1$~\cite{Buchholz2018}, and $\rho_1((1+N_f)^{-1})=\one-P_f/2$, with
$P_f$ the projection onto $f$. The commutator of $(1+N_f)^{-1}$ with
$\alpha_t((1+N_g)^{-1})$ is represented by $\frac14[P_f,P_{g_t}]$, where
$g_t:=\rme^{\rmi th_0/\hbar}g$, with $\hbar=1$ and $m_0=1/2$ in the units
of Ref.~\onlinecite{Buchholz2018}. It is non-zero when $(f,g_t)\neq0$ and
$f$ is not a multiple of $g_t$, which Lemma~\ref{lem:free_nonvanishing}
provides after a small change of $f$ if necessary. The kernels of the
$n$-particle representations $\rho_n$ decrease with
$n$~\cite{Buchholz2018}, so non-commutation holds in the Fock
representation, in every $\rho_n$ with $n\geq1$, and in every
representation of $\bar{\mathfrak A}$ whose kernel is contained in that of
$\rho_1$. Which of the equilibrium states of the homogeneous dynamics
obtained in Ref.~\onlinecite{Buchholz2026} as thermodynamic limits of
canonical ensembles have representations with this property is open.
The vacuum, which is among them, has not; for confined ensembles,
attractive forces and proper condensates put ideals into the kernels of
the limit representations, and the resolvents $(1+N_f)^{-1}$ themselves
can be represented by zero~\cite{Buchholz2026}. For fermions with
unregularised pair potentials the dynamics acts on an extension of the
CAR algebra~\cite{Siebert2025}; the same question arises there, and the
existence of KMS states is itself open.
\end{remark}

\section{Conclusion}\label{sec:conclusion}

The Galilei group acts transitively on pairs of non-simultaneous events
with a given time separation. As a result, a Galilean-invariant
commutation relation between local algebras at unequal times cannot
distinguish near from far, and if it holds at short times it contradicts
the canonical commutator in either form
(Theorem~\ref{thm:sr_selection}). With unitary Bargmann covariance and
positivity of the energy modulo the central mass, the vacuum two-point
function of a canonical field of mass $m_0$ factorises into the free
propagator and the characteristic function of an internal-energy
distribution (Lemma~\ref{lem:GKL}). Commutation at a time separation
$\tau\neq0$ then requires the internal state created by the field to be
orthogonal to its own evolution over $\tau$
(Theorem~\ref{thm:commutation_set}). This set of separations is closed
and null, bounded below by the Mandelstam--Tamm and Margolus--Levitin
times, and empty for sharp internal energy. For every internal-energy
distribution some net commutes at exactly these separations, so the
inclusion is optimal; commutation at unequal times is thus possible,
but only at a closed null set of separations. In Galilean many-body
systems with free kinetic energy and a bounded pair interaction no two
slice algebras at unequal times commute, and for fermions their odd
generators do not all anticommute (Corollary~\ref{cor:pair}); with
positivity and the link~(T) with time-smeared fields, algebras of open
spacetime regions never commute (Corollary~\ref{cor:4D}). Positivity is
essential: without it, the construction of
Theorem~\ref{thm:commutation_set}(iv) commutes, for slice algebras and
for open regions, at all separations beyond any $\tau_0>0$
(Proposition~\ref{prop:attainable}(e)), and a quasi-free Fermi system of
infinite density has equilibrium states at every temperature whose odd
fields anticommute at these separations in every representation
(Proposition~\ref{prop:fermi_medium}). For quasi-free systems,
equilibrium forces positivity for bosons, and equilibrium at finite
density does so for fermions. Among the spatially isotropic
homogeneous spacetimes with absolute time, transitivity holds exactly
for the galilean ones, up to the conjugate times at which a two-point
determinant vanishes. These are non-trivial only for the two oscillating
galilean kinematics, S10 and S11$_\chi$. For S10, the Newton--Hooke case, a harmonic trap realises
them (Theorem~\ref{thm:classification},
Proposition~\ref{prop:trap}).
Minkowski, Carroll and static kinematics admit unequal-time locality
(Remark~\ref{rem:contrast}). The results make no empirical prediction.
They say which ingredients of relativistic locality survive in each
kinematics, and they do not select among the non-galilean ones.

Several questions remain open. Exact commutation is non-generic; all
finite sets, all discrete sets $\{z_j\}$ with $\sum_j1/z_j<\infty$ and
some sets with a finite accumulation point occur as commutation sets
under positivity (Proposition~\ref{prop:attainable}), but which closed
null sets occur is only partly known; without positivity every closed
set bounded away from zero occurs. Without global
symmetry the conclusions persist in external potentials
(Proposition~\ref{prop:external}), but whether a bounded potential can
refocus exactly is open, and so is the
case of the torsional galilean spacetimes, whose Bargmann-type extension
is not central. Two further questions are the local-coupling limit of
the Galilean Lee model and interacting media. The conclusions hold in
every non-zero representation for the cutoff dynamics of Narnhofer and
Thirring (Corollary~\ref{cor:NT}), whose Fock vacuum is the witness, and
for Bose gases in the resolvent-algebra framework in every representation
whose kernel is contained in that of the one-particle representation.
Without a representation satisfying~(H) they can fail
(Propositions~\ref{prop:attainable}(e) and~\ref{prop:fermi_medium}).
Which conditions restore them for interacting media, and which
equilibrium states of these gases satisfy the condition on the kernel,
is open (Remark~\ref{rem:massive_vacuum}).

\appendix
\section{Auxiliary results}\label{app:aux}

\begin{alemma}[Vacuum expectation of commutators]\label{lem:truncation}
Let $\mathcal{M}_1,\mathcal{M}_2$ be commuting von Neumann algebras, let
$A$ be a closed operator affiliated with $\mathcal{M}_1$ and $B$ a closed
operator affiliated with $\mathcal{M}_2$, and let $\Omega\in D(A)\cap
D(A^*)\cap D(B)\cap D(B^*)$. Then $\braket{A^*\Omega,B\Omega}=
\braket{B^*\Omega,A\Omega}$. In particular this holds if $A$ and $B$ are
affiliated with the same abelian von Neumann algebra.
\end{alemma}

\begin{proof}
Let $B=V|B|$ be the polar decomposition, $E_m$ the spectral projection
of $|B|$ for $[0,m]$, and $C_m:=V|B|E_m\in\mathcal{M}_2$. Then
$C_m\Omega\to B\Omega$, and $C_m^*\Omega=E_m|B|V^*\Omega=E_mB^*\Omega\to
B^*\Omega$ because $\Omega\in D(B^*)$. Since $C_m\in\mathcal{M}_1'$, we
have $C_mA\subset AC_m$, so $C_m\Omega\in D(A)$ and $AC_m\Omega=
C_mA\Omega$. Hence $\braket{A^*\Omega,C_m\Omega}=\braket{\Omega,
AC_m\Omega}=\braket{\Omega,C_mA\Omega}=\braket{C_m^*\Omega,A\Omega}$.
Let $m\to\infty$.
\end{proof}

\begin{alemma}[Free propagation reaches everywhere]
\label{lem:free_nonvanishing}
Let $h_0=-\hbar^2\Delta/2m_0$ on $L^2(\R^3)$. For $\tau\neq0$ and
non-empty open $\Op_1,\Op_2\subset\R^3$ there are $f\in C_c^\infty(\Op_1)$
and $g\in C_c^\infty(\Op_2)$ with $(f,\rme^{\rmi\tau h_0/\hbar}g)\neq0$.
The same holds in any dimension.
\end{alemma}

\begin{proof}
In momentum space $e_\bv$ shifts the momentum by $m_0\bv$. A direct
computation gives $e_{-\bv}\rme^{\rmi\tau h_0/\hbar}e_\bv=
\rme^{\rmi m_0|\bv|^2\tau/2\hbar}\tau_{-\bv\tau}\rme^{\rmi\tau h_0/\hbar}$,
which is relation~(B4) for one particle. Suppose all the matrix elements
vanish, and fix $g\neq0$ in $C_c^\infty(\Op_2)$. Since $e_\bv g\in
C_c^\infty(\Op_2)$ and $e_{-\bv}C_c^\infty(\Op_1)=C_c^\infty(\Op_1)$,
the function $\tau_{-\bv\tau}\rme^{\rmi\tau h_0/\hbar}g$ vanishes almost
everywhere on $\Op_1$. Thus $\rme^{\rmi\tau h_0/\hbar}g=0$ almost
everywhere on $\Op_1+\bv\tau$, for every $\bv$. These sets cover $\R^3$
as $\bv$ runs through a countable dense set, since $\tau\neq0$. Hence
$\rme^{\rmi\tau h_0/\hbar}g=0$, so $g=0$, a contradiction.
\end{proof}

\begin{alemma}[Quantum speed limits]\label{lem:speed_limits}
Let $\nu$ be a probability measure on $\R$ and $\widehat\nu(s)=
\int\rme^{\rmi Es}\rmd\nu(E)$.
\begin{enumerate}[label=(\alph*)]
\item If $\int E^2\rmd\nu<\infty$ and $\Delta E$ is the standard
deviation, then $\widehat\nu(s)\neq0$ for $|s|<\pi/(2\Delta E)$.
\item If $E_0:=\min\supp\nu>-\infty$ and $\braket{E}:=\int E\,\rmd\nu<
\infty$, then $\widehat\nu(s)\neq0$ for $|s|<\pi/(2(\braket{E}-E_0))$.
\item For $\nu=\frac12(\delta_{E_0}+\delta_{E_0+\epsilon})$ both bounds
equal $\pi/\epsilon$, the first zero of $\widehat\nu$.
\end{enumerate}
\end{alemma}

\begin{proof}
Since $\widehat\nu(-s)=\overline{\widehat\nu(s)}$, it suffices to take
$s>0$. (a) Realise $\nu$ as the spectral measure of multiplication by
$E$, denoted $U$, in the vector $\xi\equiv1$ of $L^2(\R,\nu)$, so that
$c(s):=\widehat\nu(s)=\braket{\xi,\rme^{\rmi sU}\xi}$ is $C^1$. Using
$\braket{(U-\braket{E})\xi,\xi}=0$ one finds $|\frac{\rmd}{\rmd s}|c|^2|
\leq2|c|\,\Delta E\,(1-|c|^2)^{1/2}$. Let $s_*$ be the first zero of $c$
in $(0,\infty)$ and $s_1<s_*$ the last point before it with $|c(s_1)|=1$.
On $(s_1,s_*)$ the function $\theta:=\arccos|c|$ is $C^1$ with
$|\theta'|\leq\Delta E$, and it is continuous on $[s_1,s_*]$ with
$\theta(s_1)=0$ and $\theta(s_*)=\pi/2$. Hence $\pi/2\leq\Delta E\,s_*$.
(b) For $x\geq0$, $\cos x\geq1-\frac2\pi(x+\sin x)$. With $x=2y$ this is
$q(y):=2y+\sin2y-\pi\sin^2y\geq0$ for $0\leq y\leq\pi/2$, and
$q(0)=q(\pi/2)=0$ with $q'(y)=2\cos y\,(2\cos y-\pi\sin y)$, so $q$
increases and then decreases. For $x\geq\pi$, $x+\sin x\geq\pi$. Hence
$\mathrm{Re}(\rme^{-\rmi E_0s}\widehat\nu(s))\geq1-\frac2\pi
\bigl((\braket{E}-E_0)s+\mathrm{Im}(\rme^{-\rmi E_0s}\widehat\nu(s))
\bigr)$, and $\widehat\nu(s)=0$ forces $s\geq\pi/(2(\braket{E}-E_0))$.
(c) Here $\Delta E=\braket{E}-E_0=\epsilon/2$ and $|\widehat\nu(s)|=
|\cos(\epsilon s/2)|$.
\end{proof}

These are the bounds of Mandelstam and Tamm~\cite{MandelstamTamm1991}
and of Margolus and Levitin~\cite{MargolusLevitin1998}; see also
Refs.~\onlinecite{Fleming1973,DeffnerCampbell2017}.

\begin{alemma}[Divergence form]\label{lem:divergence}
Every real $f\in C_c^\infty(\R^n)$ with $\int f=0$ is of the form
$\sum_i\partial_ih_i$ with real $h_i\in C_c^\infty(\R^n)$.
\end{alemma}

\begin{proof}
For $n=1$ take $h(x)=\int_{-\infty}^xf$. For $n>1$ write $x=(x_1,x')$,
$F(x'):=\int f(x_1,x')\,\rmd x_1$, and fix $\theta\in C_c^\infty(\R)$
with $\int\theta=1$. Then $f-\theta(x_1)F(x')$ has zero integral in
$x_1$ for every $x'$, so it is $\partial_1h_1$ with $h_1(x):=
\int_{-\infty}^{x_1}[f(s,x')-\theta(s)F(x')]\,\rmd s\in
C_c^\infty(\R^n)$. Moreover $\int F=0$, so by induction
$F=\sum_{i\geq2}\partial_iH_i$ with $H_i\in C_c^\infty(\R^{n-1})$, and
$\theta(x_1)F(x')=\sum_{i\geq2}\partial_i(\theta(x_1)H_i(x'))$.
\end{proof}

\begin{alemma}[Intertwiners of the Schr\"odinger pair]\label{lem:schur}
Let $\K$ be a Hilbert space and $V:L^2(\R^3)\to L^2(\R^3)\otimes\K$ a
bounded operator with $V\tau_\ba=(\tau_\ba\otimes\one)V$ and
$Ve_\bv=(e_\bv\otimes\one)V$ for all $\ba,\bv$. Then there is
$\zeta\in\K$ with $Vg=g\otimes\zeta$ for all $g$.
\end{alemma}

\begin{proof}
For $\eta\in\K$ let $V_\eta:=(\one\otimes\langle\eta|)V$, where
$(\one\otimes\langle\eta|)(g\otimes k)=\braket{\eta,k}g$. Then $V_\eta$
commutes with the irreducible pair $\{\tau_\ba,e_\bv\}$~\cite{Folland1989},
so $V_\eta=c(\eta)\one$. The map $\eta\mapsto c(\eta)$ is antilinear and
bounded, so $c(\eta)=\braket{\eta,\zeta}$ for a unique $\zeta$. Then
$\braket{h\otimes\eta,Vg}=\braket{h,V_\eta g}=\braket{h\otimes\eta,
g\otimes\zeta}$ for all $h,\eta$, and product vectors are total.
\end{proof}

\begin{alemma}[Orthogonal modes give commuting algebras]
\label{lem:fock_factorization}
Let $\K_1,\K_2$ be orthogonal closed subspaces of a Hilbert space $\K$.
For $i=1,2$ let $\mathcal{M}_i$ be the von Neumann algebra generated on
the bosonic Fock space over $\K$ by the closures, on finite-particle
vectors, of $a(k)$, $a^*(k)$ and $a(k)+a^*(k)$, $k\in\K_i$. Then
$\mathcal{M}_1$ and $\mathcal{M}_2$ commute.
\end{alemma}

\begin{proof}
Let $\K_3$ be the orthogonal complement of $\K_1\oplus\K_2$. By the
exponential property of the symmetric Fock space~\cite{Guichardet1972}
there is a unitary map onto the tensor product of the Fock spaces over
$\K_1$, $\K_2$ and $\K_3$. It maps the finite-particle vectors onto a
subspace that contains the algebraic tensor product of the
finite-particle subspaces, and $a(k)$, $k\in\K_1$, to $a(k)\otimes\one
\otimes\one$. Since $a(k)$ is bounded on every subspace of fixed
particle number, the closures on the two domains agree, and the closure
of $X\otimes\one\otimes\one$ on the algebraic tensor product, for
$X=a(k)$, $a^*(k)$ or $a(k)+a^*(k)$, is $\overline X\otimes\one\otimes
\one$, which is affiliated with $\mathcal{B}\otimes\one\otimes\one$. Hence $\mathcal{M}_1\subset
\mathcal{B}\otimes\one\otimes\one$ and, likewise,
$\mathcal{M}_2\subset\one\otimes\mathcal{B}\otimes\one$.
\end{proof}

\begin{alemma}[Spectral unique continuation]\label{lem:spectral_UC}
Let $\Omega_0\subset\R^d$ be a connected open set whose
complement is a Lebesgue null set, let $W\in L^\infty_{\rm
loc}(\Omega_0)$ be real, and let $h$ be a self-adjoint operator on
$L^2(\R^d)$, bounded below, such that every $u\in D(h)$ lies in
$H^2_{\rm loc}(\Omega_0)$ and satisfies $hu=-\Delta u+Wu$ on $\Omega_0$.
Let $U\subset\R^d$ be open and non-empty, and let $\Psi\in L^2(\R^d)$ be
such that $\phi(h)\Psi$ vanishes almost everywhere on $U$ for every
bounded Borel function $\phi$. Then $\Psi=0$. The hypothesis on $\Psi$ holds as soon as
$\rme^{-\rmi th}\Psi$ vanishes almost everywhere on $U$ for all $t$ in
a non-empty open interval.
\end{alemma}

\begin{proof}
Adding a constant to $h$ and to $W$ changes neither the hypotheses nor
the conclusion, so let $h\geq1$. For $y>0$ put $v(y):=
\rme^{-yh^{1/2}}\Psi$. Since $\lambda^k\rme^{-y\lambda^{1/2}}$ is
bounded on $[1,\infty)$ for every $k$, dominated convergence shows that
$y\mapsto v(y)$ is infinitely differentiable from $(0,\infty)$ into
$D(h)$ with the graph norm, with $\partial_y^2v=hv$. By the closed graph
theorem, $D(h)$ embeds continuously into $H^2(K)$ for every open $K$
with compact closure in $\Omega_0$.
Hence $v(y,x):=v(y)(x)$ lies in $H^2_{\rm loc}((0,\infty)\times\Omega_0)$
and satisfies
\[
\partial_y^2v+\Delta_xv=Wv\qquad\text{on }(0,\infty)\times\Omega_0,
\]
so that $|\Delta_{y,x}v|\leq C_K|v|$ on every compact subset $K$ of the
connected open set $(0,\infty)\times\Omega_0\subset\R^{d+1}$. With
$\phi(\lambda)=\rme^{-y\lambda^{1/2}}$ the hypothesis says that $v$
vanishes on the open set $(0,\infty)\times(U\cap\Omega_0)$, which is
non-empty because $U$ has positive measure. By the unique continuation
theorem for second-order elliptic inequalities with locally bounded
lower-order terms (Ref.~\onlinecite{Hormander2007}, Sec.~17.2), $v=0$ on
$(0,\infty)\times\Omega_0$, hence almost everywhere, and $\Psi=
\lim_{y\to0}v(y)=0$.
For the last statement, let $\phi_0\in L^2(U)$. The function
$t\mapsto\braket{\phi_0,\rme^{-\rmi th}\Psi}$ is the Fourier transform
of the complex measure $\braket{\phi_0,E_h(\cdot)\Psi}$, which is
supported in $[1,\infty)$. It is therefore the boundary value of a
bounded holomorphic function on $\{\mathrm{Im}\,t<0\}$, continuous up
to the real axis, and it vanishes on an interval. By the Schwarz
reflection principle it vanishes identically, so the measure is zero
and $\braket{\phi_0,\phi(h)\Psi}=0$ for every bounded Borel $\phi$.
\end{proof}

For the Schr\"odinger evolution, the last statement is essentially
Masuda's unique continuation theorem~\cite{Masuda1967}, whose proof
continues $\rme^{-\rmi th}\Psi$ holomorphically and extends it to an
elliptic problem in more variables; the Poisson extension
$\rme^{-yh^{1/2}}$ used here is a variant. Requardt derived the same conclusion from the unique
continuation property of the generalised eigenfunctions of $h$, for
Schr\"odinger operators with a generalised eigenfunction expansion,
which by Agmon's theorem exists for short-range potentials, and stated
the converse~\cite[Theorem~2]{Requardt1986}; the Poisson extension needs
no eigenfunction expansion. Making an evolution problem
elliptic and then invoking unique continuation is also the idea of
Strohmaier's proof of the
Reeh--Schlieder property for quantum fields on stationary spacetimes,
where it is applied to generalised eigenvectors of the time
translations~\cite{Strohmaier2000}; for semibounded spectral subspaces of
Dirac-type operators see Refs.~\onlinecite{Bar2000,
BarStrohmaier2001}. Working on $\Omega_0$ rather than on $\R^d$ allows
$W$ to be singular on a closed null set whose complement is connected. For $n\geq2$ particles with
Coulomb pair interactions, $W$ is not in $L^{(3n+1)/2}_{\rm loc}$ of
$\R^{3n}$, the integrability under which unique continuation holds for
$|\Delta u|\leq|Wu|$ in $3n+1$ dimensions~\cite{JerisonKenig1985}, but
it is locally bounded away from the collision set.

\begin{acknowledgments}
This work was supported by the R+D+I efforts from guane Enterprises.
\end{acknowledgments}

\section*{Data Availability Statement}

Data sharing is not applicable to this article as no new data were
created or analyzed in this study.

\bibliography{algebraic_qkinematics}

\end{document}